\documentclass[11pt,a4paper,twoside]{article}

\usepackage[utf8]{inputenc}
\usepackage[margin=2.5cm]{geometry}
\usepackage[shortlabels]{enumitem}
\usepackage[page,toc,title]{appendix}
\usepackage{amsmath,amssymb,graphicx,xcolor,subcaption,float,amsthm,mathrsfs,mathtools,bbold,framed,url,amscd,soul,mathabx,tikz,tikz-cd}
\usepackage[color=yellow]{todonotes}
\usepackage{comment}
\usepackage{mathtools}
\usepackage{graphicx}
\setstcolor{red}
\usepackage{cancel}

\theoremstyle{plain}
\newcommand{\thistheoremname}{}
\newtheorem*{genericthm*}{\thistheoremname}
\newenvironment{namedthm*}[1]
  {\renewcommand{\thistheoremname}{#1}%
   \begin{genericthm*}}
  {\end{genericthm*}}

\usepackage{hyperref}
\hypersetup{
colorlinks=true,
}
\numberwithin{equation}{section}

\newcommand{\nocontentsline}[3]{}
\newcommand{\tocless}[2]{\bgroup\let\addcontentsline=\nocontentsline#1{#2}\egroup}
\renewcommand\vec[1]{\boldsymbol{#1}}

\DeclareMathOperator*{\argmax}{arg\,max}

\newtheorem{theorem}{Theorem}
\newtheorem{corollary}[theorem]{Corollary}
\newtheorem{lemma}{Lemma}
\newtheorem{proposition}{Proposition}

\theoremstyle{definition}
\newtheorem{definition}{Definition}
\theoremstyle{remark}
\newtheorem{remark}{Remark}[section]

\usepackage{fancyhdr}
\emergencystretch=\maxdimen
\def\p{{\partial}}

\newcommand{\mbb}[1]{\mathbb{#1}}

\newcommand{\scp}[2]{\left<#1\,,\,#2\right>}

\newcommand{\ad}{\operatorname{ad}}
\newcommand{\dd}{\mathrm{~d}}

\newcommand{\intprod}{\mathbin{\raisebox{\depth}{\scalebox{1}[-2]{$\lnot$}}}}

\def\diff{\mathrm{d}}
\def\dd{\mathbf{d}}
\def\Ad{\mathrm{Ad}}

\def\p{{\partial}}

\begin{document}
	\title{Semimartingale driven mechanics and reduction by symmetry for stochastic and dissipative dynamical systems}
	\author{Oliver D. Street\thanks{Department of Mathematics, Imperial College London, email: o.street18\@imperial.ac.uk } \and So Takao\thanks{Computing and Mathematical Sciences, California Institute of Technology.}}
	\date{\today}
	\maketitle
	\begin{abstract}
        The recent interest in structure preserving stochastic Lagrangian and Hamiltonian systems raises questions regarding how such models are to be understood and the principles through which they are to be derived.
        By considering a mathematically sound extension of the Hamilton-Pontryagin principle, we derive a stochastic analogue of the Euler-Lagrange equations, driven by independent semimartingales. Using this as a starting point, we can apply symmetry reduction carefully to derive non-canonical stochastic Lagrangian / Hamiltonian systems, including the stochastic Euler-Poincar\'e / Lie-Poisson equations, studied extensively in the literature. Furthermore, we develop a framework to include dissipation that balances the structure-preserving noise in such a way that the overall stochastic dynamics preserves the Gibbs measure on the symplectic manifold, where the dynamics effectively takes place. In particular, this leads to a new derivation of double-bracket dissipation by considering Lie group invariant stochastic dissipative dynamics, taking place on the cotangent bundle of the group.
	\end{abstract}

\tableofcontents

\section{Introduction}
The 20\textsuperscript{th} century saw significant developments in
the theoretical foundation of
Lagrangian and Hamiltonian mechanics from the perspective of differential geometry.
As has become increasingly apparent,
the rich structures present in classical mechanics can be elegantly described using ideas from differential geometry. In paricular, the theory of Lie groups has played a central role in the modern developments of classical mechanics, since Poincaré's seminal work \cite{poincare1901}. This work is important for introducing the concept of {\em symmetry reduction} (more specifically, {\em Euler-Poincaré reduction}, as we call it today), where one can recast Hamilton's equations of motion onto the Lie algebra of a group acting transitively on the original configuration space. In particular, it is shown that Euler's equations for rigid body motion, and much later, also Euler's equations for ideal fluid motion \cite{arnold1966sur, arnold1969hamiltonian} fall under this class of equations.

This interpretation of non-canonical equations as arising from a canonical system by applying symmetry reduction is not only elegant from a theoretical perspective, but has also been shown to be useful from a practical point of view. Especially in geophysical fluid dynamics (GFD), where it is common to make asymptotic approximations to Euler's fluid equations, it is crucial for said approximations to retain certain properties of the original system, such as the conservation laws. Indeed, many successful GFD models such as the quasigeostrophic equation and the Boussinesq equation are of Euler-Poincaré form \cite{salmon1988hamiltonian}, resulting from different asymptotic approximations of the ideal fluid Lagrangian / Hamiltonian \cite{HMR1998}. Conversely, this suggests that in order to develop structure-preserving GFD models, rather than taking approximations at the level of the equation, one can instead make approximations at the level of the principles used to derive the equation (e.g. a variational principle), to directly yield structure-preserving models.
Numerically, these models are often augmented by so-called {\em parameterisation schemes} \cite{kalnay2003atmospheric} to compensate for discretisation errors, which can be stochastic \cite{berner2017stochastic}, to capture these errors statistically. Motivated by this, a new type of variational principle was introduced by Holm in \cite{holm2015variational} to yield a stochastic extension of the Euler-Poincaré equations, with the goal of developing structure-preserving stochastic parameterisation schemes.

Holm's stochastic variational principle \cite{holm2015variational} for ideal fluid dynamics proceeds by modifying Hamilton's least action principle by adding certain stochastic constraint terms to the action functional. Adopting a later version of this principle presented in \cite{gay2018stochastic}, a modified action reads
\begin{align}\label{eq:stochastic-action}
    \mathcal{S} = \int^{t_1}_{t_0} \underbrace{L(g_t, V_t)}_{\text{Lagrangian}} \,\diff t + \underbrace{\left<p_t, \circ \diff g_t - V_t \,\diff t - \sum_{i=1}^N \Xi_i(g_t) \circ \diff W_t^i\right>}_{\text{Stochastic constraint}},
\end{align}
where $(g_t, V_t)$, $(g_t, p_t)$ are curves on $TG$ and $T^*G$ respectively, i.e., the tangent/cotangent bundle of a Lie group $G$, $\{W_t^i\}_{i = 1}^N$ is a collection of i.i.d. Brownian motion, $\circ$ denotes Stratonovich integration, and $\Xi_i$ is a right-invariant vector field on $G$, which we formally write as $\Xi_i(g_t) = \xi_i \cdot g_t$ for some $\xi_i$ in the Lie algebra $\mathfrak{g}$ of $G$. Assuming that $L$ is a right $G$-invariant Lagrangian, which is a reasonable assumption in ideal fluid dynamics due to the {\em particle relabelling symmetry} (in this case, $G$ is the group of diffeomorphisms over a manifold), the action can be formally re-expressed as
\begin{align}\label{eq:reduced-stochastic-action}
    \mathcal{S} = \int^{t_1}_{t_0} \ell(v_t)\,\diff t + \left<\alpha_t, \circ \diff g_t \cdot g_t^{-1} - v_t \,\diff t - \sum_{i=1}^N \xi_i \circ \diff W_t^i\right>_{\mathfrak{g}^* \times \mathfrak{g}},
\end{align}
where $v_t := V_t \cdot g_t^{-1} \in \mathfrak{g}$, $\alpha_t := p_t \cdot g_t^{-1} \in \mathfrak{g}^*$, and $\ell(v_t) := L(g_t\cdot g_t^{-1}, V_t\cdot g_t^{-1}) = L(g_t, V_t)$ is the reduced Lagrangian. Optimising the action \eqref{eq:reduced-stochastic-action} over all continuous curves $g_t \in G$ with fixed endpoints and curves $(v_t, \alpha_t) \in \mathfrak{g} \times \mathfrak{g}^*$ satisfying certain constraints, one can formally deduce the {\em stochastic Euler-Poincaré equations} for right-invariant Lagrangians, given by
\begin{equation}\label{eq:stochastic-EP-intro}
    \diff \frac{\delta\ell}{\delta v_t} + \ad^*_{v_t}\frac{\delta\ell}{\delta v_t}\, \diff t + \sum_{i = 1}^N \ad^*_{\xi_i}\frac{\delta\ell}{\delta v_t}\circ \diff W_t^i = 0 \,.
 \end{equation}
where $\ad^* : \mathfrak{g} \rightarrow \mathrm{End}(\mathfrak{g}^*)$ is the coadjoint representation.
Since its introduction in \cite{holm2015variational}, this system has found applications in areas such as shape registration \cite{arnaudon2019geometric}, MCMC sampling \cite{arnaudon2019irreversible, barp2021unifying} and data assimilation \cite{cotter2020particle, cotter2020data}. Equations of the form \eqref{eq:stochastic-EP-intro} also lend geometric structures to related models derived from them, for example a family of stochastic geophysical fluid equations can be derived from the Euler equation with noise of this form \cite{holm-luesink2021}, and a stochastic Hamiltonian structure exists for the classical water wave equations when derived from this stochastic Euler equation \cite{street2023waterwaves}. In the PDE analysis literature, a type of noise similar to that in \eqref{eq:stochastic-EP-intro}, referred to as {\em transport noise}, had already been studied actively, due to its regularising effect on the original PDE \cite{flandoli2010well, fedrizzi2013noise, bethencourt2022transport}.

Despite the growing interest in this stochastic system, the variational principle itself raises several questions from a mathematical perspective.
For example, does it make rigorous sense for a path $g_t$ to have fixed endpoints, without contradicting the adaptedness assumption necessary to define the stochastic integral? Another glaring issue is whether the term $\circ \diff g_t \cdot g_t^{-1}$ in \eqref{eq:reduced-stochastic-action} makes any sense, with ``$\circ \diff g_t$" being a mere notation rather than a proper mathematical object.
For a theoretical soundness of the approach, it is important to provide a concrete understanding of this principle,
especially in applications that make direct use of it, such as variational integrators \cite{bou2009stochastic, holm2018stochastic, kraus2021variational}.

Our first goal of the work is thus to address these issues to make rigorous sense of this variational principle. We will work in the more general setting of general driving semimartingales $S_t^i$ to replace the Brownian motion $W_t^i$ in \eqref{eq:stochastic-EP-intro} and establish the necessary conditions on the driving semimartingale to ensure that the variational principle is well defined. Moreover, while previous works have mostly focused on the case where the configuration space is a Lie group $G$, we start by considering arbitrary configuration spaces, and deduce a stochastic extension to the classic Euler-Lagrange equations. In this setting, we show that the notion of `compatibility' of the variational processes with respect to the driving semimartingale \cite{street2021semi}, is key to establishing the variational principle.
While this appears to be a basic starting point, to our knowledge, the stochastic Euler-Lagrange equations have not yet been considered explicitly in the literature. Implicitly though, this system can be expressed as Bismut's symplectic diffusion process \cite{bismut1981mechanique}, extending Hamilton's equations of motion to the stochastic setting.
We then consider symmetry-reduction of this process in the presence of a symmetry Lie group $G$. From the Hamiltonian side, symmetry reduction in the stochastic setting is straightforward -- this simply requires one to substitute the original Poisson structure with its symmetry-reduced counterpart. On the Lagrangian side, the reduction procedure is more subtle, which involves the careful construction of a variational family to rigorously establish a reduced variational principle. Our contribution here is thus to make sense of the symmetry reduction process to derive the stochastic Euler-Poincaré equations \eqref{eq:stochastic-EP-intro}.

Finally, we consider an extension of the symplectic diffusion process to include a dissipative term, following \cite{arnaudon2018noise}. This prior work considers the stochastic Euler-Poincaré equations \eqref{eq:stochastic-EP-intro} augmented by the so-called selective-decay dissipation\footnote{In the work, they refer to this as the double-bracket dissipation, but strictly speaking this is a different dissipative structure.} \cite{gay2013selective}, which is shown to preserve the Boltzmann distribution $p_\infty \propto e^{-\beta h}$, where $h : \mathfrak{g}^* \rightarrow \mathbb{R}$ is the energy of the system, when the underlying Lie algebra $\mathfrak{g}$ is compact and semi-simple. The preservation of the Boltzmann distribution is a desirable property from the point of view of statistical mechanics, being the probability distribution associated with the canonical ensemble.
However, the choice of dissipation that they add to yield the preservation of $p_\infty$ is seemingly ad hoc and does not provide any insights into how to generalise this result, for example, beyond the compact semi-simple setting. We address this by adopting a new perspective on the problem, where we demonstrate that this system can be obtained by applying symmetry reduction to a dissipative extension of the symplectic diffusion process considered in \cite{arnaudon2019irreversible}. This perspective has several implications. First, it reveals the base measure that the Boltzmann distribution $p_\infty$ is defined with respect to, which we identify as the {\em Liouville measure} $\lambda$ (i.e., the measure induced by the symplectic volume form). Thus, we obtain a more complete picture of the result in \cite{arnaudon2019irreversible}, where in fact their stochastic-dissipative system preserves the {\em Gibbs measure} $\mathbb{P}_\infty := p_\infty \lambda$. Second, our result extends beyond the Euler-Poincaré system on compact semi-simple Lie algebras. In particular, on a general Lie algebra, we identify that the selective decay dissipation is {\em not} actually the correct term to add in order to yield preservation of the Gibbs measure, but instead, the correct choice is the closely-related {\em double bracket dissipation} \cite{bloch1994dissipation, bloch1996euler}. 
As an unintended consequence, our result also reveals a new derivation of the double bracket dissipation from the viewpoint of stochastic geometric mechanics, which we believe is a novel contribution in itself.

It is worth noting that our work only considers the finite dimensional setting, as we place special emphasis on making mathematical sense of the procedures used to obtain structure-preserving stochastic-dissipative equations. Clearly, this is much more feasible in the finite-dimensional setting than in the infinite dimensional setting. However, of course this does not prevent us from formally seeing how we can get analogous equations in the infinite dimensional setting. In fact, we will look at an infinite-dimensional example towards the end, where we consider ideal fluids with noise and dissipation of the form considered in this work, and see connections with existing models. In particular, considering a finite-dimensional approximation of this model using point vortices, we arrive at a stochastic-dissipative point vortex system not previously seen in the literature.

\paragraph{Related works.} We note that the idea of a stochastic variational principle itself is not entirely new in the literature. For example, the stochastic mechanics of Nelson \cite{nelson1966derivation, nelson2020quantum, zambrini1982semi} uses a stochastic variational principle to derive equations in quantum physics, and the related stochastic Euler-Poincaré reduction in \cite{ACC2014} is used to derive dissipative systems in classical mechanics such as the Navier-Stokes equation. While these works consider stochasticity in the framework, the resulting action itself is deterministic as it is the expectation of a stochastic integral. Thus, the resulting equations are deterministic, which clearly differs from the type of systems we consider here. However, we take inspirations from their methodology, and in particular, the techniques considered in \cite{ACC2014} can be used in a similar fashion to how we derive the stochastic Euler-Poincaré equations. Regarding the unreduced case, the work \cite{lazaro2007stochastic} derives Bismut's symplectic diffusion process \cite{bismut1981mechanique} using a stochastic action principle that is not dissimilar to ours. While their work focuses entirely on the Hamiltonian side, we consider variational principles from the Lagrangian perspective to consider a stochastic extension of the Euler-Poincaré reduction later on. We also note the similarity with the work \cite{crisan2022variational} that considers Euler-Poincaré reduction for processes driven by geometric rough paths. In particular, the authors apply a similar technique as we do for the reduction procedure, using the form of variations found in \cite{ACC2014}.

\paragraph{Plan of the paper.}
\begin{itemize}
    \item In Section \ref{sec:classical_mechanics}, we illustrate the inclusion of structure-preserving stochastic terms into classical mechanics through its geometric structure. This is performed for unreduced systems whose configuration space is a manifold, and we consider both the Lagrangian and Hamiltonian viewpoints in their abstract form. On the Lagrangian side, the stochasticity is included through the Hamilton-Pontryagin principle. The corresponding Hamiltonian system can be demonstrated to be equivalent to Bismut's symplectic diffusions \cite{bismut1981mechanique}, for a particular choice of the stochastic Hamiltonians.  
    \item In Section \ref{sec:symmetry_reduction}, we consider symmetry reduction for the stochastic equations derived in the previous section. This is first considered for Hamiltonian systems defined on arbitrary symplectic manifolds, before we specialise to Lie-Poisson systems by assuming that the configuration of the system can be described by a Lie group. The corresponding Euler-Poincar\'e equations are derived, and in Theorem \ref{thm:EP} we demonstrate how these can be obtained by applying symmetry reduction to the Lagrangian.
    \item Dissipative systems are considered in Section \ref{sec:dissipation}. In particular, a dissipative term is added to the stochastic (unreduced) Hamiltonian system on a symplectic manifold that preserve the Gibbs measure canonically defined on the space. In the Lie-Poisson case, it is shown that this dissipative term is equivalent to the double bracket dissipation under symmetry reduction. Furthermore, we show ergodicity of the process in the Lie-Poisson case, so that asymptotically, any phase-space measure evolving with respect to the stochastic dynamics converge to the Gibbs measure defined on the symplectic leaves.
    \item Our symmetry-reduced stochastic dissipative system is illustrated through a series of examples in Section \ref{sec:examples}. The examples are based on well known problems in physics, and have been chosen to illustrate a range of the structures considered. 
\end{itemize}




\section{Preliminaries and notations}
We introduce some definitions and notations that will be employed throughout the paper. The purpose here is not to provide a thorough  background on these topics but rather to establish the notations and conventions that we will be using in the remainder of the text. For the differential geometric background used in Lagrangian / Hamiltonian mechanics, we refer the readers to \cite{marsden2013introduction}. We also refer the readers to the classic text \cite{oksendal2013stochastic} for necessary background in stochastic differential equations (SDEs), and \cite{kunita1984stochastic, norris1992complete} for details on SDEs defined over manifolds. 

\subsection{Exterior calculus}
Here, we introduce notations for elementary operations defined over smooth manifolds.
Given a smooth manifold $Q$, we denote by $\mathfrak{X}(Q)$ the space of vector fields and $\Omega^1(Q)$ the space of one-forms over $Q$. These are the space of smooth sections of $TQ$ and $T^*Q$ respectively. More generally, we denote by $\Omega^k(Q)$ the space of differential $k$-forms, which are smooth sections of the $k$-th exterior bundle $\bigwedge^k T^*Q$. By convention, we also denote by $\Omega^0(Q)$ the space of smooth functions $Q \rightarrow \mathbb{R}$.
The exterior derivative on $\Omega^k(\Omega)$ is denoted $\dd : \Omega^k(Q) \rightarrow \Omega^{k+1}(Q)$ and the interior product by $\intprod : \mathfrak{X}(Q) \times \Omega^k(Q) \rightarrow \Omega^{k-1}(Q)$. For $X \in \mathfrak{X}(Q)$, the {\em Lie derivative} of $k$-forms with respect to $X$, denoted $\mathcal L_X : \Omega^k(Q) \rightarrow \Omega^{k}(Q)$ is characterised algebraically by Cartan's formula
\begin{align}\label{eq:lie-deriv}
    \mathcal{L}_{X} \alpha = X \intprod \dd \alpha + \dd (X \intprod \alpha).
\end{align}
This describes the infinitesimal evolution of a $k$-form $\alpha$ along the vector field $X$, that is, $\mathcal{L}_{X} \alpha := \left.\frac{\diff}{\diff t}\right|_{t=0}\phi_t^* \alpha$, where $\phi_t : Q \rightarrow Q$ is the integral curve of $X$ and $\phi_t^* : \Omega^k(Q) \rightarrow \Omega^{k}(Q)$ denotes the pullback of $k$-forms with respect to $\phi_t$.

\subsection{Lie group action and representation}
A Lie group $G$ is a smooth manifold equipped with a group structure, such that the group multiplication and inversions are smooth maps. As a group, one can define its left or right action on a smooth manifold $Q$, denoted $L$ or $R : G \times Q \rightarrow Q$ respectively, which are smooth in both arguments. A group can also act on itself through the group operation, $G \times G \rightarrow G$, in which case the aforementioned actions are referred to as left and right translations. For a fixed $g \in G$, we employ the notation $L_g = L(g, \cdot)$ and $R_g = R(g, \cdot)$. We define the {\em tangent-lifted action} $TL : TQ \rightarrow TQ$ and $TR : TQ \rightarrow TQ$ of $G$ on $TQ$ as differentials of maps $L$ and $R$ respectively in the second argument. That is,
\begin{equation}\label{eqn:tangent_lifted_action}
    \begin{aligned}
    T_q L_g := \left. DL_g \right|_{q} : T_q Q &\rightarrow T_{L_g(q)} Q \,,
    \\
    V &\rightarrow \frac{d}{ds}\bigg |_{s=0} L_g(\gamma_1(s))
    \end{aligned}
    \begin{aligned}
    \ \hbox{and}\quad T_q R_g := \left. DR_g \right|_{q} : T_q Q &\rightarrow T_{R_g(q)} Q \,,
    \\
    V &\rightarrow \frac{d}{ds}\bigg |_{s=0} R_g(\gamma_1(s))
    \end{aligned}
\end{equation}
for any $g \in G$, $q \in Q$, and $V\in T_qQ$, where $\gamma_1:[0,1]\rightarrow Q$ is any smooth curve in $Q$ such that $\gamma_1(0)=q$ and $\gamma_1'(0) = V$. Likewise, we can define the {\em cotangent-lifted action} to be the corresponding fibre-wise adjoint operations $T^*L, T^*R : T^*Q \rightarrow T^*Q$. These are defined by considering the adjoint actions $T_{L_g(q)}^*L_{g^{-1}} : T^*_qQ \rightarrow T^*_{L_g(q)}Q$ and $T_{R_g(q)}^*R_{g^{-1}} : T^*_qQ \rightarrow T^*_{R_g(q)}Q$, which satisfy
\begin{equation}\label{eqn:cotangent_lifted_action}
   \scp{T_{L_g(q)}^*L_{g^{-1}}p }{ v_1 }  = \scp{p}{T_{L_g(q)}L_{g^{-1}}v_1 } \,,\quad\hbox{and}\quad  \scp{T_{R_g(q)}^*R_{g^{-1}}p }{ v_2 }  = \scp{p}{T_{R_g(q)}R_{g^{-1}}v_2 } \,,
\end{equation}
where $v_1 \in T_{L_g(q)}Q$, $v_2 \in T_{R_g(q)}Q$ and $p\in T^*_gG$ are arbitrary, and the above pairings are between $T_q^*Q$ and $T_qQ$. Note that the lift of the action by the inverse is necessary to ensure that the lifted action $T_qL_g$ and the cotangent lifted action $T^*_{L_g(q)}L_{g^{-1}}$ map between the same fibres of the bundle $TQ$.

Let $e \in G$ be the identity element of $G$ and take $\xi \in T_eG$. We now consider any smooth curve $\gamma_2 : [0, 1] \rightarrow G$ such that $\gamma_2(0) = e$ and $\gamma_2'(0) = \xi$.
We define the {\em left / right infinitesimal generator} (or the {\em infinitesimal action} of $G$ on $Q$) as vector fields $\xi_Q^L, \xi_Q^R \in \mathfrak{X}(Q)$, given by
\begin{equation*}
    \xi_Q^L(q) := \left.\frac{\diff}{\diff t} \right|_{t=0} L_{\gamma_2(t)}(q), \qquad \xi_Q^R(q) := \left.\frac{\diff}{\diff t} \right|_{t=0} R_{\gamma_2(t)}(q)
\end{equation*}
for any $q \in Q$. If it is clear from context whether the action is from the left or right, we omit the superscript $L$ / $R$ and simply denote the infinitesimal generator by $\xi_Q$.

As an important special case, we consider a group action (and its tangent-lifted counterpart) onto itself, giving rise to the concept of the adjoint representation of the group. First, define the {\em Lie algebra} $\mathfrak{g}$ of a Lie group as the tangent space at the identity, that is, $\mathfrak{g} = T_eG$.
This is canonically equipped with the {\em Lie bracket} $[\cdot, \cdot] : \mathfrak{g} \times \mathfrak{g} \rightarrow \mathfrak{g}$, defined by $[\xi, \eta] = \left.\left(\xi_G\eta_G - \eta_G \xi_G\right)\right|_e$ for any $\xi, \eta \in \mathfrak{g}$, where the latter is the vector field commutator of the left or right infinitesimal generators, evaluated at $e \in G$. We define the {\em adjoint representation of $g \in G$}, denoted $\Ad_g : \mathfrak{g} \rightarrow \mathfrak{g}$, as
\begin{align*}
    \Ad_g \xi := \left.\frac{\diff}{\diff t} \right|_{t=0} L_{g}(R_{g^{-1}}(\gamma_2(t))) = T_{g^{-1}}L_g \cdot (T_e R_{g^{-1}} \cdot \xi)
\end{align*}
for any $\xi \in \mathfrak{g}$, where $\gamma_2 : [0, 1] \rightarrow G$ is again a smooth curve such that $\gamma_2(0) = e$ and $\gamma_2'(0) = \xi$.
This is a linear invertible map with the inverse given by $(\Ad_g)^{-1} = \Ad_{g^{-1}}$. We may also define the adjoint representation of $\eta \in \mathfrak{g}$ as a linear map $\ad_\eta : \mathfrak{g} \rightarrow \mathfrak{g}$ defined by
\begin{align*}
    \ad_\eta \xi := \left.\frac{\diff}{\diff t} \right|_{t=0} \Ad_{\gamma_3(t)}
 \xi,
\end{align*}
for any $\xi \in \mathfrak{g}$, where $\gamma_3 : [0, 1] \rightarrow G$ is a smooth curve such that $\gamma_3(0) = e$ and $\gamma_3'(0) = \eta$. In fact, the adjoint representation of the Lie algebra is related to the Lie bracket by $\ad_\eta \xi = \pm[\eta, \xi]$, where the sign depends on whether the action is from the left ($+$) or the right ($-$). We also define the {\em coadjoint representations} $\Ad_g^* : \mathfrak{g}^* \rightarrow \mathfrak{g}^*$ and $\ad_\eta^* : \mathfrak{g}^* \rightarrow \mathfrak{g}^*$ by $\left<\Ad_g^* \mu, \xi\right>_{\mathfrak{g}^* \times \mathfrak{g}} = \left<\mu, \Ad_g\xi\right>_{\mathfrak{g}^* \times \mathfrak{g}}$ and $\left<\ad_\eta^* \mu, \xi\right>_{\mathfrak{g}^* \times \mathfrak{g}} = \left<\mu, \ad_\eta\xi\right>_{\mathfrak{g}^* \times \mathfrak{g}}$, respectively.

\subsection{Lagrangian mechanics}\label{sec:Lagrangian_prelim}
In the Lagrangian formulation of classical mechanics, the equations of motion are derived entirely from the {\em Lagrangian function}, $L:TQ \rightarrow \mathbb{R}$, defined on the tangent bundle of the configuration space $Q$, representing the space of positions and velocities. This is achieved through {\em Hamilton's Principle}, which seeks to find the extrema of the action functional $S[q] := \int^{t_1}_{t_0} L(q(t), \dot{q}(t)) \diff t$ among all curves $q : [t_0, t_1] \rightarrow Q$ such that $q(t_0) = a$ and $q(t_1) = b$ for some fixed points $a, b \in Q$, and $\dot{q}(t) \in T_{q(t)}Q$ denotes the derivative of the curve $q(t)$. More precisely, consider the variational problem
\begin{equation}\label{eq:hamilton's-principle}
	0 = \delta S := \left. \frac{\diff}{\diff t}\right|_{\epsilon = 0} \int_{t_0}^{t_1} L(q_\epsilon(t), \dot q_\epsilon(t)) \,\diff t,
\end{equation}
where $\{q_\epsilon\}_{\epsilon \in [0,1]}$ is a smoothly parameterised family of $C^2$ curves in $Q$ with $q_\epsilon(t_0) = a$ and $q_\epsilon(t_1) = b$ for all $\epsilon \in [0, 1]$. The solution to \eqref{eq:hamilton's-principle} is equivalent to solving the {\em Euler-Lagrange equations}, which, on a local chart, has the expression
\begin{align}\label{eq:euler-lagrange}
    \frac{\diff}{\diff t} \frac{\partial L}{\partial \dot{q}^i} = \frac{\partial L}{\partial q^i}, \quad i = 1, \ldots, \mathrm{dim}(Q).
\end{align}
In particular, on $Q = \mathbb{R}^n$ and taking $L(q, \dot{q}) = \frac{m}{2} \|\dot{q}\|^2 - V(q)$, equation \eqref{eq:euler-lagrange} recovers Newton's second law $F = m\ddot{q}$ under conservative forcing $F = \nabla V(q)$. A benefit of the Lagrangian perspective is that unlike Newton's second law, the expression \eqref{eq:euler-lagrange} is invariant under change of coordinates, providing a coordinate-free description of classical conservative mechanics.

\subsection{Hamiltonian mechanics}
The Hamiltonian formulation of mechanics provides yet another viewpoint of classical mechanics, this time derived from the Hamiltonian function $H : T^*Q \rightarrow \mathbb{R}$ defined on the cotangent bundle of the configuration space. The cotangent bundle, which represents the space of position and momenta, is equipped with the so-called {\em canonical symplectic form} $\omega \in \Omega^2(T^*Q)$. In local coordinates, this is given by $\omega = \sum_i \diff q^i \wedge \diff p_i$, where $q^i, p_i$ are the coordinates for position and momenta respectively. Given a Hamiltonian $H : T^*Q \rightarrow \mathbb{R}$, we define the corresponding {\em Hamiltonian vector field} to be a vector field $X_H \in \mathfrak{X}(T^*Q)$ satisfying
\begin{align} \label{eq:ham-vec-field}
X_H \intprod \omega = \dd H.
\end{align}
The equations of motion are then given by the ODE $\frac{\diff}{\diff t}(q, p) = X_H(q, p)$, which, in local coordinates can be written
\begin{align}\label{eq:ham-eq}
    \frac{\diff q^i}{\diff t} = \frac{\partial H}{\partial p_i}, \qquad \frac{\diff p_i}{\diff t} = -\frac{\partial H}{\partial q^i}, \quad i = 1, \ldots, \mathrm{dim}(Q).
\end{align}
Note that unlike the Euler-Lagrange equation, which is a second order ODE, Hamilton's equation of motion \eqref{eq:ham-eq} is a system of first order ODEs. However, one can show that under certain regularity conditions, \eqref{eq:ham-eq} is equivalent to \eqref{eq:euler-lagrange} by means of the {\em Legendre transform}. This is achieved by considering {\em fibre derivative} $\mathbb{F}L : TQ \rightarrow T^*Q$ of the Lagrangian, defined in local coordinates as $\mathbb{F}L(q, \dot{q}) = (q, {\partial L}/{\partial \dot{q}})$ and taking $H \circ \mathbb{F} L(q, \dot{q}) = \left<\mathbb{F} L(q, \dot{q}), \dot{q}\right> - L(q, \dot{q})$.

The Hamiltonian formulation of mechanics leads us to consider far-reaching generalisations of classical mechanical systems beyond those defined over tangent/cotangent bundles.
By the definition of Hamiltonian vector fields given above, one can easily extend this notion to arbitrary {\em symplectic manifolds} $(P, \omega)$, an even dimensional manifold $P$ equipped with a closed, non-degenerate two-form $\omega \in \Omega^2(P)$. By the closedness of $\omega$, one can verify the following essential property 
\begin{align}\label{eq:omega-invariance}
    \mathcal{L}_{X_H} \omega \stackrel{\eqref{eq:lie-deriv}}{=} X_H \intprod \cancel{\dd \omega} + \dd (X_H \intprod \omega) = \dd (X_H \intprod \omega) \stackrel{\eqref{eq:ham-vec-field}}{=} \dd \dd H = 0,
\end{align}
which states that $\omega$ is invariant under the Hamiltonian dynamics.
We can also go beyond the symplectic setting by considering dynamics on a {\em Poisson manifold} $(P, \{\cdot, \cdot\})$, where $\{\cdot, \cdot\} : \Omega^0(P) \times \Omega^0(P) \rightarrow \Omega^0(P)$ is the {\em Poisson bracket}, an antisymmetric derivation in both arguments satisfying the Jacobi identity $\{f, \{g, h\}\} + \{g, \{h, f\}\} + \{h, \{f, g\}\} = 0$. Given a Hamiltonian $H : P \rightarrow \mathbb{R}$, we can define the Hamiltonian vector field $X_H \in \mathfrak{X}(P)$ as a derivation $X_H = \{\cdot, H\} \in \mathrm{Der}(\Omega^0(P)) \cong \mathfrak{X}(P)$. The space of all such Hamiltonian vector fields, denoted $\mathfrak{X}_{ham}(P)$, is closed under the vector field commutator, due to the well-known relation $[X_f, X_g] = -X_{\{f, g\}}$.
In the case when $(P, \omega)$ is a symplectic manifold, we can define a Poisson bracket by $\{f, g\} := X_g \intprod (X_f \intprod \omega)$. Hence, the Poisson formulation naturally generalises the symplectic formulation of Hamiltonian mechanics, allowing us to, for example, define Hamiltonian systems over odd-dimensional manifolds $P$.

Another benefit of the Hamiltonian framework is that it provides an elegant interpretation of conservation laws. Given a Poisson manifold $(P, \{\cdot, \cdot\})$, let $G$ be a Lie group acting on $P$ with the corresponding Lie algebra denoted $\mathfrak{g}$. We say that the action is {\em Hamiltonian} if for any $\xi \in \mathfrak{g}$, the infinitesimal generator $\xi_P$ is a Hamiltonian vector field for some Hamiltonian $J^\xi : P \rightarrow \mathbb{R}$, i.e., $\xi_P = X_{J^\xi}$. The {\em momentum map} $J : P \rightarrow \mathfrak{g}^*$ is then defined by $\left<J(z), \xi\right>_{\mathfrak{g}^* \times \mathfrak{g}} = J^\xi(z)$ for any $z \in P$ and $\xi \in \mathfrak{g}$. When the Hamiltonian $H : P \rightarrow \mathbb{R}$ is invariant under $G$-action, then {\em Noether's theorem} states that the momentum map is invariant under the corresponding dynamics, i.e., $J \circ \phi_t = J$, where $\phi_t$ is the integral curve of the Hamiltonian vector field $X_H$. Thus, the momentum map generalises the concept of conserved quantities, such as linear and angular momentum in classical mechanics.

\subsection{Stochastic calculus on manifolds}
Let $(\Omega, \mathcal{F}, \mathbb{P})$ be a probability triple, where $\Omega$ is a sample space, $\mathcal{F}$ is a $\sigma$-algebra and $\mathbb{P}$ is a probability measure. For $T>0$, a {\em stochastic process} is a collection of random variables $S_t : \Omega \rightarrow \mathbb{R}$ indexed by $t \in [0,T]$ such that the mapping is $\mathbb{P}$-measurable. For each $t \in [0, T]$, we let $\mathcal{F}_t$ denote the $\sigma$-algebra generated by the random variable $S_t$, forming a filtration $\mathcal{F}_t \subset \mathcal{F}_s \subset \mathcal{F}$ for all $t < s$. A process $\{M_t\}_{t \in [0, T]}$ is said to be {\em $\mathcal{F}_t$-adapted} if $M_t$ is $(\mathcal{F}_t, \mathbb{P})$-measurable and is furthermore a {\em martingale} if $M_t = \mathbb{E}[M_s | \mathcal{F}_t]$ holds for any $s > t$.
We denote the covariation of two stochastic processes $S_t^1$ and $S_t^2$ by $[S_\cdot^1, S_\cdot^2]_t$, and in the special case $S_t^1 = S_t^2 = S_t$, we set $[S_\cdot]_t := [S_\cdot, S_\cdot]_t$, which is the {\em quadratic variation} of $S_t$. We also consider the {\em total variation} of the process as the limit $V(S) := \lim_{\Delta t \rightarrow 0} \sum_{n=1}^{N} |S_{t_{n}} - S_{t_{n-1}}|$ and say that the process has {\em bounded variations} if $V(S) < \infty$. The space of bounded variation processes is the largest subset of continuous functions where one can define the Riemann-Stieltjes integral. In particular, $[S_\cdot^1, S_\cdot^2]_t$ and $[S_\cdot]_t$ are bounded variation processes.

A {\em semimartingale} is a stochastic process $\{S_t\}_{t \in [0, T]}$ that can be decomposed into a local martingale $M_t$ and a bounded variation process $A_t$, i.e., $S_t = M_t + A_t$. This forms the largest class of processes where one can define {\em stochastic integrals}. In particular, we will mainly use Stratonovich's definition of stochastic integrals, which we denote by $X_t \mapsto \int_0^T X_t \circ \diff S_t$ for some $\mathcal{F}_t$-adapted semimartingale $X_t$, with a ``$\circ$" between the integrand and the integrator. For It\^o's definition of stochastic integrals, we simply omit the ``$\circ$" symbol. A major advantage of Stratonovich's definition of stochastic integrals is that it satisfies the usual chain rule
$f(S_T) - f(S_0) = \int^T_0 \frac{\partial f}{\partial S_t} \circ \diff S_t$,
for any smooth function $f : \mathbb{R} \rightarrow \mathbb{R}$. This allows us to extend the definition of semimartingales and Stratonovich integrals from $\mathbb{R}$ to manifolds naturally by working on local charts \cite{norris1992complete}. That is, for a smooth manifold $Q$, we say that a collection of measurable functions $Z_t : \Omega \rightarrow Q$, $t \in [0, T]$ is a {\em $Q$-valued semimartingale} if on a local chart, each coordinate $Z_t^i$ is a semimartingale, and for any $T^*Q$-valued semimartingale $\{\alpha_t\}_{t \in [0, T]}$, we define the Stratonovich integral $\int^T_0 \left<\alpha_t, \circ \,\diff Z_{t}\right>$ by $\sum_{i=1}\int^T_0 \alpha_i(t) \circ \diff Z_{t}^i$ on local charts. The latter expression is invariant under change of coordinates due to the chain rule for Stratonovich integrals.

Now, given a manifold $Q$, a set of vector fields $X_1, \ldots, X_N \in \mathfrak{X}(Q)$ and a set of $\mathbb{R}$-valued semimartingales $S_t^1, \ldots, S_t^N$, we consider a $Q$-valued semimartingale $\{Z_t\}_{t \in [0, T]}$ satisfying the system
\begin{align}\label{eq:sde-integral-form}
    \int^t_0 \left<\alpha, \circ \,\diff Z_{t'}\right> = \sum_{i=1}^N \int^t_0 \left<\alpha, X_i(Z_{t'})\right> \circ \diff S_{t'}^i, \quad t \in [0, T],
\end{align}
for any $\alpha \in \Omega^1(Q)$.
We say that $\{Z_t\}_{t \in [0, T]}$ is a solution to a {\em Stratonovich stochastic differential equation (SDE)}, which, despite its name, is in fact an integral equation. However, we often adopt a shortened ``differential" notation, expressing the SDE \eqref{eq:sde-integral-form} in the form
\begin{align}\label{eq:sde-differential-form}
\diff Z_t = \sum_{i=1}^N X_i(Z_t) \circ \diff S_t^i,
\end{align}
for convenience. For any $f \in C^\infty(Q; Q)$, and $\{Z_t\}_{t \in [0, T]}$ solving \eqref{eq:sde-differential-form}, one can check that the following identity holds
\begin{align}\label{eq:strat-chain-rule}
    \diff f(Z_t) = \sum_{i=1}^N f_* X_i(Z_t) \circ \diff S_t^i,
\end{align}
which we refer to as the {\em Stratonovich chain rule}. Here, $f_* : \mathfrak{X}(Q) \rightarrow \mathfrak{X}(Q)$ denotes the pushforward of vector fields along $f$.

\section{Stochastic Hamilton-Pontryagin principle}\label{sec:classical_mechanics}

In this section, we develop principles within which one can derive stochastic equations of motion that will form the basis of the remaining parts of the work. By incorporating stochastic effects within the high-level principles of classical mechanics, one is able to naturally preserve certain geometric structures of mechanical systems, such as symplecticity and conservation laws. In Section \ref{sec:stoch-lagrangian}, we consider a stochastic extension of Hamilton-Pontryagin principle, a control-theoretic reformulation of Hamilton's principle, to constrain our processes to satisfy an SDE via Lagrange multipliers. This has already been considered in \cite{gay2018stochastic} at a formal level on Lie group-valued processes, however our emphasis here is in making rigorous sense of it to justify the procedure. We then show equivalence of the resulting system with Bismut's stochastic Hamiltonian diffusion \cite{bismut1981mechanique} in Section \ref{sec:stoch-ham}, which immediately implies the preservation of symplectic volume under the stochastic dynamics and the momentum map, when the system is invariant under a Lie group.


\subsection{Stochastic Lagrangian mechanics}\label{sec:stoch-lagrangian}



In Hamilton's principle, as stated in Section \ref{sec:Lagrangian_prelim}, one seeks to find an extrema of the action $\mathcal{S}$ (a time integral of the Lagrangian $L : TQ \rightarrow \mathbb{R}$), over $C^2$-paths $q : [t_1, t_2] \rightarrow Q$ with fixed endpoints. By the differentiability assumption, the path $q_t$ directly lifts to a path $(q_t, \dot{q}_t)$ on $TQ$. Thus, taking variations of $\mathcal{S}$ over a family of paths on $Q$ induce variations over a family of paths on $TQ$, which is crucial for determining the Euler-Lagrange equations, as the Lagrangian is a function of $TQ$. It is also important that we are not considering variations of arbitrary paths $(q_t, V_t)$ on $TQ$, but rather, only those induced by differentiable curves $q_t$ on $Q$, as the derivation of the Euler-Lagrange equations from Hamilton's principle heavily depends on the relation $\delta V_t = \frac{\diff}{\diff t} \delta q_t$. In particular, this only holds when $V_t = \dot{q}_t$. We see that this poses immediate problem when we try to extend Hamilton's principle to accommodate $C^0$-stochastic processes $q_t$, as its time-derivative is not necessarily defined. Moreover, even if we consider smooth approximations to the process and try to invoke a Wong-Zakai type limiting argument, it is still not clear how this yields an SDE in the limit (e.g. how must the LHS in \eqref{eq:euler-lagrange} converge in the limit if it ever does?).

This leads us to consider a variational principle of a slightly different nature, wherein the equation satisfied by the path $q_t$ is {\em directly imposed} via Lagrange multipliers, which is the idea behind {\em Hamilton-Pontryagin principle} \cite{YOSHIMURA2007381}. In this framework, variations of the action are not taken over curves in $Q$, but rather on an extended space referred to as the {\em Pontryagin bundle}, defined as follows.

\begin{definition}
    For the configuration manifold, $Q$, the \emph{Pontryagin bundle}, denoted by $TQ\oplus T^*Q$, is defined as the Whitney sum of the vector bundles $TQ$ and $T^*Q$. That is, it is a vector bundle with base manifold $Q$, whose fibre over $q\in Q$ is $T_qQ \oplus T^*_qQ$.
\end{definition}
Locally, the coordinates $(q, V, p)$ of $TQ\oplus T^*Q$ describe respectively, the position/configuration, velocity and momenta of the system. The momenta in particular acts as a Lagrange multiplier enforcing the relation between the position and velocity variables, similar to the role of the adjoint state variable in optimal control. Thus, under the Hamilton-Pontryagin framework, the variational principle \eqref{eq:hamilton's-principle} can be re-expressed as
\begin{equation}\label{eqn:HP_action}
	0 = \delta \int_{t_0}^{t_1} L(q_t, V_t) + \langle p_t, \dot q_t - V_t \rangle \,\diff t \,,
\end{equation}
where the pairing $\langle \,\cdot \, , \, \cdot \,\rangle : T^*_qQ\times T_q Q \rightarrow \mathbb{R}$ is the natural nondegenerate pairing on this space. Here, the variation is taken over arbitrary $C^1$-curves $[t_1, t_2] \rightarrow TQ \oplus T^*Q$, such that the base curve $q_t \in Q$ is $C^2$-smooth with fixed endpoints.
An advantage of the Hamilton-Pontryagin framework is that one now has the flexibility to impose arbitrary relations satisfied by $q_t$, beyond that considered in \eqref{eqn:HP_action}.
In particular, can can ask if it would be possible to impose a stochastic relation such as
\begin{equation}\label{eqn:stochastic_position_velocity}
	\diff q_t = V_t \, \diff t + \sum_{i \geq 1} \Xi_i(q_t) \circ \diff S_t^i \,,
\end{equation}
where $\{\Xi_i\}_{i \geq 1}$ is a collection of smooth vector fields over $Q$ and $S_t = (t,S_t^1,S_t^2,\dots)$ is the \emph{driving semimartingale}, as in \cite{street2021semi}. As a result, one hopes to obtain stochastic equations of motion, such that certain properties of the original deterministic system are retained. 
In the remainder of this section, we discuss how to make sense of this variational principle under stochastic constraints and subsequently derive an extension of the Euler-Lagrange equations that is stochastic. 


We begin by imposing some assumptions on the sequence of semimartingales $\mathcal{S}_t$ that is used to drive the process $q_t$. This is given in the following definition. 

\begin{definition}[Driving semimartingale]\label{def:driving_semimartingale}
 	Suppose we have a sequence of $\mathbb{R}$-valued continuous semimartingales $S_t := (S_t^0, S_t^1,\dots)$, and consider the Doob-Meyer decomposition of each component,
    \begin{equation}
	   S_t^i = A_t^i + M_t^i \,,
    \end{equation}
    where each $M_t^i$ is a local martingale and $A_t^i$ is a c\`adl\`ag adapted finite variation process. Then we say that $S_t$ is a \emph{driving semimartingale} if:
	\begin{enumerate}
        \item $S_t^0 = t$.
		\item For all $i > 0$, we have that if $M_t^i \equiv 0$, then $A_t^i \equiv 0$ (i.e., if $S_t^i$ is a non-zero process, then it must contain a non-zero martingale part). Furthermore, we assume that each $M_t^i$ is a square-integrable martingale, i.e., $\sup_{t\geq 0} \mathbb{E}[(M_t^i)^2] < \infty$ and $A_t^i$ is an increasing process unless $A_t^i \equiv 0$. In particular, the former condition implies that $[M^i]_t$ for $M_t^i \not\equiv 0$ is strictly increasing in $t$.
        \item The covariation between distinct components of the semimartingale vanishes, i.e., for all $i \neq j$, we have $[S^i,S^j]_t = 0 \,$ for all $t > 0$.
	\end{enumerate}
 For convenience, when there are processes $S_t^i$ that are identically zero, we exclude them from $S_t$. Thus, if $S_t \in c_{00}$ (an eventually-zero sequence), then we represent it as a finite vector $S_t \in \mathbb{R}^N$. We also denote by $\mathcal{F}_t$ the sigma-algebra generated by $S_t$, forming a filtration.
\end{definition}

\begin{remark}
    These conditions are imposed such that the stochastic variational principle is well-defined. Hamiltonian systems, which can exist outside of the variational principle, can be defined with no such limitations on the form of each $S_t^i$. Furthermore, we make use of the the first condition, $S_t^0 = t$, only when performing stochastic Lagrangian reduction. It is possible to study the stochastic Euler-Lagrange  equations under the weaker assumption that $S_t^0$ is a strictly increasing c\`adl\`ag adapted finite-variation process.
\end{remark}

We note that under certain conditions on the vector fields $\{\Xi_i\}_{i \geq 1}$ and the driving semimartingale $S_t$, we can make sense of the stochastic integral in \eqref{eqn:stochastic_position_velocity} for {\em countably infinite} diffusion terms. For example, let $V \subset Q$ be a compact set and $U \subset V$ be a local chart on $Q$. Further, let $q(0) \in U$, $\tau_U$ be the stopping time $\tau_U := \inf \{t \in [0, T] : q(t) \notin U\}$ and assume that $\{\Xi_i\}_{i\geq 1}$ satisfies $\sum_{i = 1}^\infty \|\Xi_i\|_{C^0(U)}^2 \big(1+\|\Xi_i\|_{C^1(U)}^2\big) 
\mathbb{E}\big[([M_\cdot^i]_t - [M_\cdot^i]_0) + (A_t^i - A_0^i)] < \infty,$
for $t \in \tau_U$. Then for $N \geq 1$, we have
\begin{align*}
    \mathbb{E}\left[\left|\int^t_0 \Xi_N(q_r) \circ \diff S_r^N\right|^2\right] &= \mathbb{E}\left[\left|\frac12 \int^t_0 \sum_{j=1}^{d} \Xi_N^j(q_r) \frac{\partial \Xi_N(q_r)}{\partial q_r^j}\diff [M^N_\cdot]_r + \int^t_0 \Xi_N(q_r) \,\diff S_r^N\right|^2\right] \\
    &\leq \mathbb{E}\left[\frac12 \int^t_0 \left(\sum_{j=1}^{d} \Xi_N^j(q_r) \frac{\partial \Xi_N(q_r)}{\partial q_r^j} \right)^2 + |\Xi_N(q_r)|^2 \,\diff [M_\cdot^N]_r + \int^t_0 |\Xi_N(q_r)|^2 \,\diff A_r^N\right] \\
    &\lesssim \|\Xi_N\|_{C^{0}_U}^2 \big(1+\|\Xi_N\|_{C^1_U}^2\big) 
\mathbb{E}\big[([M_\cdot^N]_t - [M_\cdot^N]_0) + (A_t^N - A_0^N)\big] \\
& \rightarrow 0,
\end{align*}
as $N \rightarrow \infty$, where we used the Stratonovich-to-It\^o conversion in the first line, and Cauchy's inequality, together with the  It\^o isometry in the second line.
This implies that $Z^N_s := \sum_{i=1}^N \int^s_0 \Xi_i(q_r) \circ \diff S_r^i$ forms a Cauchy sequence in $L^2(\Omega \times [0, t]; U)$, which allows us to rigorously talk about the limiting process $Z^\infty_s := \lim_{N \rightarrow \infty} Z^N_s \in L^2(\Omega \times [0, t]; U)$, thus making sense of \eqref{eqn:stochastic_position_velocity} for countably infinite diffusion terms for all $0 < t < \tau_U$. We can furthermore extend the time beyond $\tau_U$ by switching to an overlapping chart and assuming a similar condition for the $\Xi_i$'s on this chart.

Hereafter, we will take the driving semimartingale $S_t$ to be of the form described in Definition \ref{def:driving_semimartingale} and the $\Xi_i$'s to be regular enough for the SDEs to make sense. This is since the choice commonly made in the literature, $S_t^0 = t$, $S_t^i = W_t^i$ for i.i.d. Brownian motions $W_t^i$, is admissible within this context.

Following \cite{street2021semi}, we also introduce the useful notion of \emph{compatibility} with respect to a driving semimartingale. This definition formalises the understanding that the evolution of variables within our system is to be governed by stochastic-in-time equations with respect to the driving semimartingale.

\begin{definition}[Compatibility with the driving semimartingale]\label{def:compatibility}
 We say that a stochastic process $f_t : \Omega \rightarrow Q$ is \emph{compatible} with the driving semimartingale, $\{ S_t^i\}_{i \geq 0}$, if there exists a collection of $TQ$-valued semimartingales $\{F_t^i\}_{i \geq 0}$, such that
	\begin{equation}\label{eqn:compatibility_definition}
		\int^t_0 \left<\alpha_{t'}, \,\circ\, \diff f_{t'}\right> = \sum_{i\geq 0} \int^t_0 \left<\alpha_{t'}, F_{t'}^i\right> \circ \diff S_{t'}^i \,,
	\end{equation}
 holds for any $T^*Q$-valued semimartingale $\alpha_t$.
\end{definition}
\begin{remark}\label{rmk:compatibility}
    In the case $S_t = (t, W_t^1, \ldots, W_t^N)$ on $Q = \mathbb{R}^n$, the martingale representation theorem \cite[Theorem 4.3.4]{oksendal2013stochastic} states that a sufficient condition for $f_t$ to admit a unique decomposition of the form \eqref{eqn:compatibility_definition} is for $f_t$ to be $\mathcal{F}_t$-adapted, where $\{\mathcal{F}_t\}_{t \geq 0}$ is the filtration generated by $W_t = (W_t^1, \ldots, W_t^N)$. 
    Moreover, by the Clark-Ocone theorem, $F_t^i$ in \eqref{eqn:compatibility_definition} has the explicit expression $\mathbb{E}[D_t^i f | \mathcal{F}_t]$, where $D_t f$ is the {\em Malliavin derivative}, formalising the notion of derivatives with respect to the Brownian motion.
    We conjecture that an analogous result holds for manifolds and general driving semimartingales (and perhaps already known in the literature),
    however until this is known, we impose this as an assumption.
\end{remark}

It remains to demonstrate that taking variations of a stochastic action integral, where the time integration is taken against a driving semimartingale in the Stratonovich sense, is formally reasonable and results in usable equations and relationships. This issue is handled in the following Lemma.

\begin{lemma}[Fundamental lemma of the stochastic calculus of variations]\label{lemma:fundamental}
Let $E$ be a vector bundle and suppose we have a collection of $E$-valued continuous semimartingales $\{F_i\}_{i\geq 0}$, and a driving semimartingale, $S_t$, as defined in Definition \ref{def:driving_semimartingale}. Further, for $0 < t_0 < t_1 < \infty$, let $X$ be any subset of continuous curves $[t_0, t_1] \rightarrow E^*$ with compact support that is moreover dense in $L^2([t_0, t_1]; E^*)$. If, for any family of $E^*$-valued continuous semimartingales $\{\alpha_i(t)\}_{i \geq 0}$ with $\alpha_i \in X$, we have
\begin{equation}\label{eqn:lemma_assumption}
    \sum_{i\geq 0}\int_{t_0}^{t_1}\scp{\alpha_i(t)}{F_i(t)}_{E^* \times E}\circ \diff S_t^i = 0 \,,
\end{equation}
then
\begin{equation}\label{eqn:lemma_conclusion}
    F_i(t) = 0 \,,
\end{equation}
on the fibres for each $i \geq 0$ and all $t \in [t_0,t_1]$.
\end{lemma}
\begin{remark}
    A version of this Lemma has also appeared previously in the literature \cite{street2021semi}, with the difference being that this version of the Lemma requires each $F_i(t)$ to vanish independently as opposed to the sum of their stochastic integrals being almost surely equal to zero. Furthermore, whilst this result has been stated here for vector bundle-valued processes for maximal generality, we can clearly also state this for processes that take values in a given vector space. In the deterministic case (see e.g. \cite[Lemma 1]{gelfand2000calculus}), the set $X$ is typically taken to be the set of compactly supported smooth curves, but of course this also holds for any family $X$ that is dense in $L^2$, e.g. compactly supported $C^1$-curves.
\end{remark}
\begin{proof}
Taking the quadratic variation of \eqref{eqn:lemma_assumption}, we have
\begin{align*}
    0 &= \sum_{i,j\geq 0}\int_{t_0}^{t_1}\scp{\alpha_i(t)}{F_i(t)}\scp{\alpha_j(t)}{F_j(t)} \diff [S^i,S^j]_t
    \\
    &= \sum_{i\geq 0}\int_{t_0}^{t_1}\scp{\alpha_i(t)}{F_i(t)}^2 \diff[M^i]_t \,,
\end{align*}
where we have used the fact that $[S^i,S^j]_t$ is assumed, in Definition \ref{def:driving_semimartingale}, to be zero when $i\neq j$. Furthermore, we have that $[M^i]_t$ is strictly increasing and is assumed to be zero when $i=0$. Thus, we have that $F_i(t) \equiv 0$ on $(t_0,t_1)$ for all $i\geq 1$ by the non-degeneracy of the pairing, which can be further extended to $[t_0, t_1]$ by the continuity of $F_i$. It remains to show this for $i=0$. Substituting this back into equation \eqref{eqn:lemma_assumption}, we have
\begin{equation*}
   \int_{t_0}^{t_1}\scp{\alpha_0(t)}{F_0(t)}_{E^*\times E} \diff t = 0 \,.
\end{equation*}
This directly implies that $F_0(t) \equiv 0$ in the same manner as proving the deterministic fundamental lemma of the calculus of variations \cite[Lemma 1]{gelfand2000calculus}.
\end{proof}


As an immediate consequence, we can show that if a stochastic process is compatible with the driving semimartingale in the sense of Definition \ref{def:compatibility}, then the decomposition \eqref{eqn:compatibility_definition} is unique.

\begin{corollary} \label{cor:uniqueness-of-decomposition}
    Let $f_t$ be a $Q$-valued continuous semimartingale that is compatible with a driving semimartingale $S_t$. Then the decomposition \eqref{eqn:compatibility_definition} is unique.
\end{corollary}
\begin{proof}
    Suppose that there exists two collections of $TQ$-valued semimartingales $\{F_t^i\}_{i \geq 0}$ and $\{\tilde{F}_t^i\}_{i \geq 0}$ such that
    \begin{equation}
		\int^{t_1}_{t_0} \left<\alpha_t, \,\circ\, \diff f_t\right> = \sum_{i\geq 0} \int^{t_1}_{t_0} \left<\alpha_t, F_t^i\right> \circ \diff S_t^i = \sum_{i\geq 0} \int^{t_1}_{t_0} \big<\alpha_t, \tilde{F}_t^i\big> \circ \diff S_t^i \,,
	\end{equation}
 for an arbitrary $T^*Q$-valued semimartingale $\alpha_t$. Then, we have
 \begin{align}
     \sum_{i\geq 0}\int_{t_0}^{t_1} \scp{\alpha_t}{F_t^i - \tilde F_t^i}_{T^*Q \times TQ} \circ \diff S_t^i = 0,
 \end{align}
 which implies that $F_t^i = \tilde{F}_t^i$ for all $i \geq 0$ and for all $t \in [t_0,t_1]$, by Lemma \ref{lemma:fundamental}.
\end{proof}

This uniqueness of decomposition is key to deducing stochastic differential equations from a stochastic action principle. To this end, we make use of the following result.

\begin{corollary}\label{cor:fundamental_lemma}
    Let $Q$ be a finite-dimensional manifold, $0 < t_0 < t_1 < \infty$, and $X$ a subset of continuous curves $[t_0, t_1] \rightarrow T^*Q$ with compact support that is moreover dense in $L^2([t_0, t_1]; T^*Q)$. Suppose that for a $TQ$-valued semimartingale $f_t$ that is compatible with a driving semimartingale $S_t$ in the sense of Definition \ref{def:compatibility}, we have that
    \begin{equation}\label{eqn:corollary_assumption}
        \int_{t_0}^{t_1} \scp{\alpha_t}{\circ \,\diff f_t - \sum_{i\geq 0}\tilde F_t^i \circ \diff S_t^i}_{T^*Q \times TQ} = 0 \,,
    \end{equation}
    for a collection of $TQ$-valued semimartingales $\{\tilde{F}_t^i\}_{i\geq 0}$ and an arbitrary $T^*Q$-valued semimartingale $\alpha_t \in X$. Then $f_t$ satisfies
	\begin{equation}\label{eqn:corollary_conclusion}
		\diff f_t = \sum_{i \geq 0} \tilde F_t^i \circ \diff S_t^i \,.
	\end{equation}
\end{corollary}
\begin{proof}
    This follows directly from the uniqueness of decomposition \eqref{eqn:compatibility_definition} by Corollary \ref{cor:uniqueness-of-decomposition}.
\end{proof}

Returning to equation \eqref{eqn:stochastic_position_velocity}, for a driving semimartingale corresponding to Definition \ref{def:driving_semimartingale}, we will consider a stochastic variational principle in which the tangent vector $V$ performs the same role as the time derivative of the configuration, $\dot q_t$, does in Hamilton's principle.
The vector fields $\{\Xi_i\}_{i \geq 1}$ in \eqref{eqn:stochastic_position_velocity} are taken to be exogenous, determined from data for example, to model stochastic deviations from the ``primary'' process $\dot{q}_t = V_t$.
Within the Hamilton-Pontryagin principle, we may additionally include stochastic integral terms $\sum_{i} \int \Gamma_i(q)\circ\diff S_t^i$, which augment the term $\int L(q,V)\,\diff t$ in \eqref{eqn:HP_action}. These extra terms $\Gamma_i$, referred to as \emph{stochastic potentials} model stochastic forces that act vertically on the fibres (a vector field $F \in \mathfrak{X}(TQ)$ is {\em vertical} iff $\pi_* F \equiv 0$, where $\pi : TQ \rightarrow Q$ is the natural projection).
Putting this together, we arrive at the following variational principle, generalising Hamilton-Pontryagin principle to yield stochastic equations of motion.

\begin{theorem}[Stochastic Hamilton-Pontryagin Principle]\label{thm:stochastic_HP}
    For fixed $0 < t_0 < \infty$ and stopping time $t_1 > t_0$, consider the stochastic action functional
    \begin{align}\label{eq:HP-action}
        \mathcal{S} = \int_{t_0}^{t_1} L(q_t, V_t) \,\diff t + \sum_{i \geq 1} \Gamma_i(q_t)\circ \diff S_t^i + \scp{p_t}{\circ\,\diff q_t - V_t \,\diff t - \sum_{i \geq 1} \Xi_i(q_t) \circ \diff S_t^i},
    \end{align}
    where $(q_t, V_t, p_t)$ is a $\mathcal{F}_t$-adapted curve in $TQ \oplus T^*Q$ (here, $\mathcal{F}_t$ denotes the filtration generated by the driving semimartingale $S_t$). Then, taking the extrema of $\mathcal{S}$ among a family of $\mathcal{F}_t$-adapted $TQ \oplus T^*Q$-valued continuous semimartingales $Z_t = (q_t, V_t, p_t)$ that is compatible with $\{S_t^i\}_{i \geq 0}$ and such that the endpoints of the base process $q_t$ satisfy $q(t_0) = a$ and $q(t_1) = b$ for some $a \in Q$ and random variable $b : \Omega \rightarrow Q$, we obtain the implicit stochastic Euler-Lagrange equations, expressed on local charts as
	\begin{align}
		\diff p_t &= \frac{\p L}{\p q_t}\, \diff t + \sum_{i\geq 1} \frac{\p \Gamma_i}{\p q_t} \circ \diff S_t^i -  \sum_{i\geq 1} \frac{\p}{\p q_t} \scp{p_t}{\Xi_i}\circ \diff S_t^i
    \label{eq:stoch-euler-lagrange-mom}
		\,,\\
		p_t &= \frac{\p L}{\p V_t}
		\,,\\
		\diff q_t &= V_t \,\diff t + \sum_{i \geq 1} \Xi_i(q_t) \circ \diff S_t^i 
    \label{eq:stoch-euler-lagrange-pos}
		\,,
	\end{align}
    where we assume sufficient regularity on $L,$ $\{\Gamma_i\}_{i \geq 1}$ and $\{\Xi_i\}_{i \geq 1}$ for the right hand sides to make sense.
\end{theorem}
\begin{remark}
    A similar variational principle was considered in \cite{bou2009stochastic}, where the Hamilton-Pontryagin approach was used to include the stochastic potential terms. However, the above theorem is more general in that it includes the noise vector fields $\Xi_i$. In the following section, the additional terms corresponding to $\{\Xi_i \}_{i\geq 1}$ will be shown to reduce in the Euler-Poincar\'e setting such that they encompass the `transport noise' considered for fluid mechanics in \cite{holm2015variational}.
\end{remark}
\begin{proof}
    For simplicity of presentation, we shall only prove the result in local coordinates. That is, we fix a local chart $U \subset Q$ with $q_{t_0} \in U$ and choose $t_1$ to be the  first exit time of $U$ such that $q_t \in U$ for all $t \in [t_0, t_1)$. A global extension is made possible by considering the intrinsic form of the stochastic Euler-Lagrange equations, using a slight extension of the argument given in \cite[Proposition 3.2]{yoshimura2006diracII}. We refer the readers to Appendix \ref{app:intrinsic-stoch-EL} for more details on the global extension.
    
    We start by taking the infinitesimal variation of the action functional $\mathcal{S}$ among all admissible paths $(q_t,V_t,p_t)$.
    To do so, we take a family of processes $\{q_{t,\epsilon}\}_{\epsilon \in (-1, 1)}$ depending smoothly on $\epsilon$ and is chosen such that, almost surely, we have $q_{t,\epsilon}|_{\epsilon=0} = q_t$ for all $t$ and $q_{t,\epsilon}|_{t=t_0,t_1} = q_t|_{t=t_0,t_1}$ for all $\epsilon$. We also define arbitrary processes $\delta V_t$ and $\delta p_t$ in the fibre of the Pontryagin bundle at $q_t$. The infinitesimal variation of $(q_t,V_t,p_t)\in TU \oplus T^*U \cong U \times \mathbb{R}^n \times \mathbb{R}^n$ (here, $n := \mathrm{dim}(Q)$) is defined as the process
    \begin{equation}
        (\delta q_t,\delta V_t,\delta p_t) := \frac{\p}{\p\epsilon}\bigg|_{\epsilon=0}(q_{t,\epsilon},V_t+\epsilon\delta V_t,p_t + \epsilon\delta p_t) \in T_{(q_t,V_t,p_t)}(TU\oplus T^*U) \cong \mathbb{R}^n \times \mathbb{R}^n \times \mathbb{R}^n\,.
    \end{equation}
    Now, taking $0 = \delta \mathcal{S}[q_t,V_t,p_t] := \left.\frac{\p}{\p\epsilon}\right|_{\epsilon=0} \mathcal{S}[q_{t,\epsilon},V_t+\epsilon\delta V_t,p_t + \epsilon\delta p_t]$, we get
    \begin{equation}
    \begin{aligned}\label{eq:local-chart-variations-HP}
        0 &= \int_{t_0}^{t_1} \scp{\frac{\p L}{\p q_t}\,\diff t + \sum_{i\geq 1} \frac{\p \Gamma_i}{\p q_t} \circ \diff S_t^i - \circ \diff p_t -\sum_{i\geq 1}\frac{\p}{\p q_t} \scp{p_t}{\Xi_i}\circ\diff S_t^i }{\delta q_t}  + \scp{p_t}{\delta q_t} \Big|_{t_0}^{t_1}
        \\
        &\qquad + \scp{\frac{\p L}{\p V_t} - p_t}{\delta V_t}\, \diff t + \scp{\delta p_t}{\circ\,\diff q_t - V_t \circ \diff S_t^0 - \sum_{i \geq 1} \Xi_i(q_t) \circ \diff S_t^i} \,,
    \end{aligned}
    \end{equation}
    where $\left<\cdot, \cdot\right>$ is the Euclidean inner product in $\mathbb{R}^n$. To obtain this, we made use of the identity
    \begin{align*}
        &\delta \int_{t_0}^{t_1} \left<p_t, \circ \diff q_t\right> = \sum_{i \geq 0} \delta \int_{t_0}^{t_1} \left<p_t, F_i^q(t)\right> \circ \diff S_t^i = \sum_{i \geq 0} \int_{t_0}^{t_1} \Big( \left<\delta p_t, F_i^q(t)\right> + \left<p_t, \delta F_i^q(t)\right> \Big) \circ \diff S_t^i \\
        &= \sum_{i \geq 0} \int_{t_0}^{t_1} \Big(\left<\delta p_t, \circ \diff q_t\right> + \left<p_t, \circ \diff \delta q_t\right> \Big) = \sum_{i \geq 0} \int_{t_0}^{t_1} \Big(\left<\delta p_t, \circ \diff q_t\right> - \left<\circ \diff p_t, \delta q_t \right> \Big)+ \scp{p_t}{\delta q_t} \Big|_{t_0}^{t_1},
    \end{align*}
    where $\{F_i^q(t)\}_{i \geq 0}$ is a family of $TQ$-valued semimartingales satisfying $\diff q_t = \sum_{i \geq 0} F_i^q(t) \circ \diff S_t^i$ (this follows from the compatibility of the process $q_t$ with the driving semimartingale $S_t = (S^0_t, S^1_t, \ldots)$), and the last equality follows from the Stratonovich product rule $\diff \left<p, \delta q\right> = \left<p, \circ \diff \delta q\right> + \left<\circ \diff p, \delta q\right>$.
    Finally, invoking Corollary \ref{cor:fundamental_lemma}, and noting that $\left.\scp{p_t}{\delta q_t} \right|_{t_0}^{t_1} = 0$ due to the endpoint conditions imposed on $q_t$, we obtain the following relationships
	\begin{align}
		\diff p_t &= \frac{\p L}{\p q_t}\, \diff t + \sum_{i\geq 1} \frac{\p \Gamma_i}{\p q_t} \circ \diff S_t^i -  \sum_{i\geq 1} \frac{\p}{\p q_t} \scp{p_t}{\Xi_i}\circ \diff S_t^i
		\,,\\
		p_t &= \frac{\p L}{\p V_t}
		\,,\\
		\diff q_t &= V_t \,\diff t + \sum_{i \geq 1} \Xi_i(q_t) \circ \diff S_t^i 
		\,.
	\end{align}
\end{proof}

In the above result, we have assumed that there exists a variational family of curves $Z_t^\epsilon := (q^\epsilon_t, V^\epsilon_t, p^\epsilon_t)$ such that the endpoint conditions $q^\epsilon(t_0) = a$ and $q^\epsilon(t_1) = b$ hold for any $\epsilon$, which, at first sight seems to contradict our other assumption that $Z_t^\epsilon$ is $\mathcal{F}_t$-adapted (this is necessary for the stochastic integrals to make sense). How is it that we can assume a condition on a future time $t_1$, when our process is adapted? To answer this, we emphasise a key difference with the deterministic variational principle, where in our setting, the final time $t_1$ and endpoint $b$ are chosen to be random. Thus, for an $\mathcal{F}_t$-adapted path $q_t$ such that $q_{t_1} = b$ for random variables $t_1$ and $b$ that are a priori unknown, we need only to construct $\mathcal{F}_t$-adapted perturbed paths $q_t^\epsilon$ from $q_t$ with $q_{t_0}^\epsilon = a$, such that they agree at some stopping time $t_1$, i.e., $q_{t_1}^\epsilon = q_{t_1}$ for all $\epsilon$. Now, there are several possible approaches for constructing such a perturbation. The one that we discuss below is based on \cite[Section 4.1]{lazaro2007stochastic}.

Let $K \subset Q$ be a compact set with $a \in K$, and $t_1$ be the first exit time of the process $q_t$ leaving the set $K$. Then for an arbitrary smooth vector field $X_K \in \mathfrak{X}(Q)$ that vanishes on the set $\{a\} \cup \partial K$, consider an $\epsilon$-perturbation of the process $q_t$ by
\begin{align}
    \frac{\partial q_t^\epsilon}{\partial \epsilon} = X_K(q_t^\epsilon), \quad q^0_t = q_t,
\end{align}
for all $t \in [t_0, t_1]$. By construction, we see that indeed $q_{t_0}^\epsilon = a$ and $q_{t_1}^\epsilon = q_{t_1}$ holds true for all $\epsilon$ and furthermore, since the construction of $q_t^\epsilon$ does not require any information of $S_t$ after time $t$, it is $\mathcal{F}_t$-adapted.

Later on, we will see another construction of such family in the special case $Q = G$, where $G$ is a Lie group. In this case, one can explictly construct perturbations of paths in the group using {\em deterministic} curves on the corresponding Lie algebra, with vanishing endpoints. We will see that this also leads to a variational family of processes that are $\mathcal{F}_t$-adapted and satisfy the endpoint conditions.

\begin{remark}
    In the deterministic setting, the time-differentiability of the variational curves is essential in ensuring the uniqueness of the action functional extrema, since if we only impose time-continuity, then one can construct infinitely many solutions called the {\em broken extremal} solutions, that satisfy Hamilton's principle \cite[Section 15]{gelfand2000calculus}. In our case, we do not encounter such issue despite our processes being only continuous, due to our semimartingale compatibility assumption (Definition \ref{def:compatibility}). Thus, the compatibility assumption is crucial in establishing Theorem \ref{thm:stochastic_HP}, and is moreover quite a natural assumption in the stochastic setting, by Remark \ref{rmk:compatibility}.
\end{remark}


\subsection{Stochastic Hamiltonian mechanics}\label{sec:stoch-ham}
The Hamiltonian counterpart of the stochastic system \eqref{eq:stoch-euler-lagrange-mom}--\eqref{eq:stoch-euler-lagrange-pos} can be traced back to Bismut's foundational work \cite{bismut1981mechanique}, which he refers to as {\em diffusions symplectiques} (or {\em symplectic diffusions}). On $T^*\mathbb{R}^n$, Bismut's symplectic diffusion is described as the following set of SDEs
\begin{align}
    \diff q_t &= \sum_{i \geq 0} \frac{\partial H_i}{\partial p_t} \circ \diff S_t^i \,,
    \label{eqn:ham_q}\\
    \diff p_t &= - \sum_{i \geq 0} \frac{\partial H_i}{\partial q_t} \circ \diff S_t^i,
    \label{eqn:ham_p}
\end{align}
echoing the deterministic system \eqref{eq:ham-eq}.
As in the deterministic case, it is straightforward to generalise this system to arbitrary symplectic manifolds $(P, \omega)$. This is given by the following SDE on $Z_t \in P$
\begin{align} \label{eq:general-symplectic-diffusion}
    \diff Z_t = \sum_{i \geq 0} X_{H_i}(Z_t) \circ \diff S_t^i \,,
\end{align}
where $X_{H_i} \in \mathfrak{X}_{ham}(P)$ denotes the Hamiltonian vector field corresponding to a Hamiltonian $H_i$.
Below, we show that the system \eqref{eq:general-symplectic-diffusion} on $P = T^*Q$ is indeed equivalent to the stochastic Euler-Lagrange equations \eqref{eq:stoch-euler-lagrange-mom}--\eqref{eq:stoch-euler-lagrange-pos} under appropriate conditions for which we can apply the Legendre transform.

\begin{proposition}
    When the symplectic manifold is taken to be a cotangent bundle of a configuration space,  $P = T^*Q$, and the Lagrangian $L : TQ \rightarrow \mathbb{R}$ is hyperregular, then the stochastic Euler-Lagrange equations \eqref{eq:stoch-euler-lagrange-mom}--\eqref{eq:stoch-euler-lagrange-pos} and Bismut's symplectic diffusion \eqref{eq:general-symplectic-diffusion} are equivalent.
\end{proposition}
\begin{proof}
Since the Lagrangian is hyperregular, by definition, the fibre derivative $\mathbb{F}L:TQ\rightarrow T^*Q$ is a diffeomorphism. As in the deterministic case, we will show that a stochastic Hamiltonian flow is induced on $T^*Q$. Indeed, for finite dimensional manifolds where $(q, V)$ denote the local coordinates on $TQ$, the fibre derivative defines the momenta, $p$, by
\begin{equation}\label{eq:conjugate-momenta}
    p := \mathbb{F}L_q(V) = \frac{\p L}{\p V}(q, V) \,.
\end{equation}
on local charts. This defines the Hamiltonian $H_0 : T^*Q \rightarrow \mathbb{R}$ from the Lagrangian $L$ via the Legendre transform:
\begin{equation}
	H_0(q,p) := \scp{p}{V} - L(q, V), \quad \text{where} \quad V = (\mathbb{F}L_{q})^{-1}(p)
	\,. \label{eq:legendre-transform-L}
\end{equation}
To demonstrate that our flow is indeed a stochastic Hamiltonian system in the sense of Bismut, we must also define the stochastic Hamiltonians, $H_i:T^*Q\rightarrow \mathbb{R}$ for $i \geq 1$. To do this, we take $p$ as defined by \eqref{eq:conjugate-momenta}, and set
\begin{equation}
	H_i(q, p) := \scp{p}{\Xi_i(q)} - \Gamma_i(q)
	\,. \label{eq:legendre-transform-Gamma}
\end{equation}
Notice that $H_0$ is defined through the Legendre transform, however each $H_i$ is defined merely through an algebraic relation resembling the Legendre transform. With these definitions for the Hamiltonians, the Hamilton-Pontryagin action functional \eqref{eq:HP-action} then reads
\begin{equation}\label{eqn:unreduced_stochastic_hamiltonian_action}
	\int_{t_0}^{t_1} -H_0(p_t, q_t) \, \diff t - \sum_{i\geq 1} H_i(p_t, q_t)\circ \diff S_t^i + \scp{p_t}{\circ\,\diff q_t},
\end{equation}
which, upon taking variations in $p_t$ and $q_t$ on a local chart and invoking Corollary \ref{cor:fundamental_lemma}, yields the system \eqref{eqn:ham_q}--\eqref{eqn:ham_p} as expected.
\end{proof}


A phase space variational principle starting from an action integral of the form \eqref{eqn:unreduced_stochastic_hamiltonian_action} has been previously considered in \cite{lazaro2007stochastic} to derive the stochastic Hamiltonian system \eqref{eqn:ham_q}--\eqref{eqn:ham_p}. Interestingly, the authors do not consider the Lagrangian viewpoint in their work; as mentioned earlier, there are certain obstacles to extend Hamilton's principle directly to the stochastic case, which we overcome by the Hamilton-Pontryagin formulation. Our stochastic variational principle (Theorem \ref{thm:stochastic_HP}) starting from the Lagrangian, thus serves as a missing link between the stochastic Hamilton-Pontryagin principle of \cite{bou2009stochastic} and the phase space variational principle of \cite{lazaro2007stochastic}. We note that the former is only capable of generating stochastic dynamics that are $C^1$-smooth (since stochasticity only appears in the velocity component), which form only a subset of the stochastic Euler-Lagrange equations \eqref{eq:stoch-euler-lagrange-mom}--\eqref{eq:stoch-euler-lagrange-pos} that we consider in this work.

The stochastic system \eqref{eq:general-symplectic-diffusion} retains many of the fundamental properties of Hamiltonian dynamics. In particular, we can show the invariance of the symplectic form $\omega$ under the stochastic dynamics.
\begin{lemma}
The symplectic form $\omega$ is invariant with respect to the flow of \eqref{eq:general-symplectic-diffusion}.
\end{lemma}
\begin{proof}
Let $\phi_t$ be the flow of \eqref{eq:general-symplectic-diffusion}. Then using the formula for the pullback of tensor fields along a stochastic flow of diffeomorphism (see Theorem 4.9.3 in \cite{kunita1984stochastic})\footnote{This formula is a natural extension of equation \eqref{eq:strat-chain-rule}.}, we have
\begin{align*}
    \phi_t^*\omega - \omega = \sum_{i \geq 0} \int^t_0 \phi_r^* \mathcal{L}_{X_{H_i}} \omega \circ \diff S_r^i = 0 \,,
\end{align*}
where we used that $\mathcal{L}_{X_{H}} \omega = 0$ for any Hamiltonian $H$ (see \eqref{eq:omega-invariance}).
\end{proof}

As an easy corollary, we obtain a stochastic analogue of the classic Liouville's theorem.
\begin{corollary}[Stochastic Liouville's theorem]
Given a $2n$-dimensional symplectic manifold $(P, \omega)$, the symplectic volume form $\omega^n := \omega \wedge \cdots \wedge \omega$ ($n$ times) is preserved under the flow of \eqref{eq:general-symplectic-diffusion}.
\end{corollary}

Now, we can further generalise the system \eqref{eq:general-symplectic-diffusion} to be defined on a Poisson manifold $(P, \{\cdot, \cdot\})$, giving a process $Z_t \in P$ that satisfies the SDE
\begin{align}\label{eq:poisson-diffusion}
    \diff F(Z_t) = \sum_{i \geq 0} \{F,H_i\}(Z_t)\circ \diff S_t^i \,,
\end{align}
for any smooth function $F \in \Omega^0(P)$. This formulation allows a simplified proof of a stochastic counterpart to Noether's theorem, which we state as follows.
\begin{theorem}[Stochastic Noether's theorem]\label{prop:stoch-noether}
Let $G$ be a Lie group acting on a Poisson manifold $(P, \{\cdot, \cdot\})$ such that its action is Hamiltonian. If all of the Hamiltonians $H_i$ for $i \geq 0,$ in \eqref{eq:poisson-diffusion} are $G$-invariant, then the momentum map $J : P \rightarrow \mathfrak g^*$ corresponding to this action is preserved under the flow of \eqref{eq:poisson-diffusion}.
\end{theorem}
\begin{proof}
Fix $\xi \in \mathfrak g$ and define $J^\xi(Z_t) := \left<\xi,J(Z_t)\right>$. By definition of the momentum map, we have $X_{J^\xi} = \xi_P$ (See Definition 11.2.1 in ref.\cite{marsden2013introduction}). Then from \eqref{eq:poisson-diffusion}, we have
\begin{align*}
    \diff J^\xi(Z_t) &= \sum_{i \geq 0} \{J^\xi,H_i\}(Z_t) \circ \diff S_t^i = - \sum_{i \geq 0} \{H_i, J^\xi\}(Z_t) \circ \diff S_t^i \\
    &= - \sum_{i \geq 0} X_{J^\xi} H_i(Z_t) \circ \diff S_t^i = - \sum_{i \geq 0} \xi_PH_i(Z_t) \circ \diff S_t^i \\
    &= 0 \,,
\end{align*}
where we used that $\xi_P H_i = 0$ for all $i \geq 0$, due to $G$-invariance of the Hamiltonians. Since $\xi$ was chosen arbitrarily, the momentum map $J$ is preserved.
\end{proof}

The Poisson bracket $\{\cdot,\cdot\}$ may be degenerate and we call the central elements the {\em Casimir functions}.
\begin{definition}[Casimirs]
The Casimir function $C \in \Omega^0(P)$ of the Poisson bracket $\{\cdot,\cdot\}$ is any function that commutes with any other functions over $P$, i.e., $\{F,C\} = 0$ for all $F \in \Omega^0(P)$.
\end{definition}
Clearly, if $C$ is a Casimir function for the bracket $\{\cdot,\cdot\}$, then it is conserved by \eqref{eq:poisson-diffusion} automatically  regardless of the choice of Hamiltonians. We note that for Poisson brackets generated by symplectic forms, the only Casimirs are the constant functions due to the non-degeneracy of symplectic forms.

Finally, we can also express \eqref{eq:poisson-diffusion} in It\^o form as follows.
\begin{proposition}
In It\^o form, equation \eqref{eq:poisson-diffusion} can be expressed as
\begin{align}
    \diff F(Z_t) = \sum_{i \geq 0} \{F,H_i\}(Z_t) \,\diff S_t^i + \frac12 \sum_{i \geq 0} \{\{F,H_i\},H_i\}(Z_t) \,\diff [S^i]_t \,.
\end{align}
\end{proposition}
\begin{proof}
Applying the first order linear differential operator $\{\,\cdot\,,H_i\}$ to both sides of \eqref{eq:poisson-diffusion}, we get
\begin{align*}
    \diff \{F,H_i\} = \sum_{j \geq 0}\{\{F,H_j\},H_i\} \circ \diff S_t^j \,.
\end{align*}
Hence, the Stratonovich-to-It\^o correction term reads
\begin{align*}
    \diff [\{F,H_i\}, S^i]_t = \{\{F,H_i\},H_i\} \,\diff [S^i]_t \,,
\end{align*}
where we used our assumption that $[S_\cdot^i, S_\cdot^j]_t = 0$ for all $i \neq j$.
\end{proof}

\section{Symmetry reduction for stochastic mechanics}\label{sec:symmetry_reduction}
A striking feature of Lagrangian and Hamiltonian systems is that in the presence of symmetries, we can characterise its dynamics by a reduced Lagrangian/Hamiltonian systems taking place on lower dimensional manifolds.
In this section, we derive analogous results for stochastic systems with symmetries at different levels of generality. In particular, we will conisder symplectic and Poisson reduction for stochastic Hamiltonian systems in Sections \ref{sec:symplectic-reduction} and \ref{sec:poisson-reduction}, and then investigate the special case when the phase space is given by the cotangent bundle of a Lie group in Section \ref{sec:lie-poisson}, giving us a stochastic analogue of the Lie-Poisson/Euler-Poincar\'e reduction \cite{marsden1986reduction, HMR1998}.
We observe that the derivation is fairly straightforward from the Hamiltonian perspective, merely echoing the procedures in the deterministic case. However, from the Lagrangian perspective (i.e., the Euler-Poincar\'e reduction), the reduction procedure requires a more careful treatment. In Section \ref{sec:euler-poincare-reduction}, we illuminate how this can be achieved by restricting our variations to be of a particular form, enabling us to define a stochastic analogue of the Lin constraint.

\subsection{Symplectic reduction} \label{sec:symplectic-reduction}
Consider a stochastic Hamiltonian system 
\eqref{eq:general-symplectic-diffusion} defined over a symplectic manifold $(P, \omega)$, and let $G$ be a Lie group acting on $P$. By Noether's theorem (Theorem \ref{prop:stoch-noether}), if the Hamiltonians $\{H_i\}_{i \geq 0}$ are all invariant under the group action with respect to $G$, then for a regular value $\mu \in \mathfrak{g}^*$, the pathwise dynamics of \eqref{eq:general-symplectic-diffusion} effectively take place on the submanifold $J^{-1}(\mu) \subseteq P$, where $J : P \rightarrow \mathfrak{g}^*$ is the momentum map corresponding to the $G$-action (note that the value of $\mu$ is entirely determined by the initial condition).
Furthermore, we can consider a reduced stochastic system on the so-called {\em reduced space} $P_\mu := J^{-1}(\mu) / G_\mu$, where $G_\mu \leq G$ is the isotropy subgroup $G_\mu := \{g \in G : \mathrm{Ad}^*_g\mu = \mu\}$.

Letting $\pi_\mu : J^{-1}(\mu) \rightarrow J^{-1}(\mu)/G_\mu$ be the projection and $\iota_\mu: J^{-1}(\mu) \xhookrightarrow{} P$ be the natural inclusion, we see that the space $P_\mu$ inherits a symplectic form $\omega_\mu$ from $P$ through the relation 
\begin{align} \label{eq:reduced-symplectic-form}
    \iota_\mu^* \omega = \pi_\mu^*\omega_\mu \,.
\end{align}
Defining the {\em reduced Hamiltonians} $h_i : P_\mu \rightarrow \mathbb{R}$ for all $i \geq 0$ by $\pi_\mu^* h_i = \iota_\mu^*H_i$ and the corresponding Hamiltonian vector field on the symplectic manifold $(P_\mu, \omega_\mu)$ by $X_{h_i}^\mu \in \mathfrak{X}_{ham}(P_\mu)$, we consider a stochastic Hamiltonian system on $P_\mu$ given by
\begin{align} \label{eq:reduced-symplectic-diffusion}
    \mathrm{d}z_t^\mu = \sum_{i \geq 0} X^\mu_{h_i}(z_t^\mu) \circ \mathrm{d}S_t^i \,.
\end{align}
The system \eqref{eq:reduced-symplectic-diffusion} on $(P_\mu, \omega_\mu)$ is related to the original system \eqref{eq:general-symplectic-diffusion} on $(P, \omega)$ via {\em symplectic reduction} \cite{marsden1974reduction}, which we state below.

\begin{theorem} \label{prop:stoch-symplectic-reduction}
    Consider the stochastic Hamiltonian system \eqref{eq:general-symplectic-diffusion} on the symplectic manifold $(P, \omega)$ such that all $\{H_i\}_{i\geq 0}$ are $G$-invariant.
    Further, for some $T>0$, let $\{\Phi_t\}_{t \in [0, T]}$ be the stochastic flow of \eqref{eq:general-symplectic-diffusion} on $P$ and $\{\phi_t^\mu\}_{t \in [0, T]}$ be the stochastic flow of \eqref{eq:reduced-symplectic-diffusion} on $P_\mu$. Then, there exists a flow $\psi_t$ on $J^{-1}(\mu)$ such that $\Phi_t \circ \iota_\mu = \iota_\mu \circ \psi_t$ and $\phi_t^\mu \circ \pi_\mu = \pi_\mu \circ \psi_t$ for all $t \in [0, T]$.
\end{theorem}
\begin{proof}
The proof holds almost exactly as in the deterministic case. See Appendix \ref{app:stoch-symplectic-reduction} for details.
\end{proof}

Intuitively, Theorem \ref{prop:stoch-symplectic-reduction} states that in the presence of Lie group symmetries, the original stochastic Hamiltonian system \eqref{eq:general-symplectic-diffusion} on a symplectic manifold $P$ projects down via the mapping $\pi_\mu$ to the stochastic Hamiltonian system \eqref{eq:reduced-symplectic-diffusion} on the reduced symplectic space $P_\mu$, with lower dimension.

\subsection{Poisson reducion}\label{sec:poisson-reduction}
Under the same setting as before, we can also derive a reduced stochastic system on the space $P/G$ by projecting the system \eqref{eq:general-symplectic-diffusion} down via the quotient map $\pi : P \rightarrow P/G$. This is known as {\em Poisson reduction} \cite{marsden1986reduction}, and the resulting system will in general inherit a Poisson structure that is not necessarily symplectic. To see this, if we let $\{\cdot, \cdot\}$ be the symplectic Poisson bracket associated to the symplectic manifold $(P, \omega)$, then the projection $\pi: P \rightarrow P/G$ induces a Poisson bracket $\{\cdot, \cdot\}_{P/G} : \Omega^0(P/G) \times \Omega^0(P/G) \rightarrow \Omega^0(P/G)$, given by
\begin{align} \label{eq:poisson-map}
    \{f, g\}_{P/G}\circ \pi := \{f \circ \pi, g \circ \pi\} \,,
\end{align}
for all $f, g \in \Omega^0(P/G)$ (see \cite[Chapter 10.5]{marsden2013introduction}).
Now, for $Z_t$ solving \eqref{eq:general-symplectic-diffusion}, we can compute the dynamics of the projected system $\pi(Z_t)$ and deduce that it satisfies a stochastic Hamiltonian system in the Poisson sense. 

To see this, we first fix an arbitrary $f \in \Omega^0(P/G)$. Then using the Stratonovich chain rule \eqref{eq:strat-chain-rule}, extended to consider countably infinite diffusion terms, we have
\begin{align}
    \mathrm{d}f(\pi(Z_t)) &= \sum_{i \geq 0} (f \circ \pi)_* X_{H_i}(Z_t) \circ \mathrm{d}S_t^i = \sum_{i\geq 0} X_{H_i}(f \circ \pi) (Z_t) \circ \mathrm{d}S_t^i = \sum_{i\geq 0} \{f \circ \pi, H_i\} (Z_t) \circ \mathrm{d}S_t^i \,. \label{eq:poisson-reduction-computation}
\end{align}
Defining the {\em reduced Hamiltonian} $h_i \in \Omega^0(P/G)$ by $H_i = h_i \circ \pi$, we then find
\begin{align*}
    \eqref{eq:poisson-reduction-computation} &= \sum_{i\geq 0} \{f \circ \pi, h_i \circ \pi\} (Z_t) \circ \mathrm{d}S_t^i = \sum_{i\geq 0} \{f, h_i\}_{P/G}(\pi(Z_t)) \circ \mathrm{d}S_t^i \,,
\end{align*}
where we used the definition \eqref{eq:poisson-map} in the last line. The projected dynamics $[Z_t] := \pi(Z_t)$ thus evolve according to the stochastic Poisson-Hamiltonian system
\begin{align} \label{eq:stochastic-poisson-system}
    \diff f([Z_t]) = \sum_{i\geq 0} \{f, h_i\}_{P/G}([Z_t]) \circ \mathrm{d}S_t^i \,.
\end{align}

\begin{remark} \label{rmk:symplectic-leaf-preservation}
The Poisson reduced system \eqref{eq:stochastic-poisson-system} is in fact explicitly related to the symplectic reduced system \eqref{eq:reduced-symplectic-diffusion} by the symplectic stratification theorem \cite[Theorem 10.4.4]{marsden2013introduction}. This states that Poisson manifolds can be expressed as a disjoint union of {\em symplectic leaves}, which, in the current setting, can be further identified with the reduced spaces $P_\mu := J^{-1}(\mu) / G_\mu$ (see \cite{marsden1992lectures}).
Moreover, the Hamiltonian vector fields $X_h := \{\cdot, h\}_{P/G}$  for any $h \in C^\infty(P/G; \mathbb{R})$ span the tangent spaces of the symplectic leaves \cite[Theorem 10.4.4]{marsden2013introduction}, implying that the dynamics of the Poisson-Hamiltonian system \eqref{eq:stochastic-poisson-system} effectively take place on the symplectic leaves $\equiv$ reduced spaces. This restricted dynamics on the reduced space is exactly expressed by \eqref{eq:reduced-symplectic-diffusion}.
\end{remark}

An advantage of considering the Poisson reduced system \eqref{eq:stochastic-poisson-system} over the symplectic reduced system \eqref{eq:reduced-symplectic-diffusion} is that often the Poisson structure of the former is more amenable to explicit computations than the latter. Moreover, in the special case $P = T^*G$ which we will consider next in further details, the Poisson manifold $P/G$ can be identified with the dual Lie algebra $\mathfrak{g}^*$, and is therefore simpler to work with, both analytically and computationally.

\subsection{Lie-Poisson reduction} \label{sec:lie-poisson}
In this part and the next, we will consider the special case $P = T^*G$, where the special structure of the phase space enables explicit computation of the reduced system \eqref{eq:stochastic-poisson-system}. Provided the Hamiltonians $H_i$ are all left-invariant under the cotangent-lifted action of $G$, then by Lie-Poisson reduction theory \cite[Chapter 13]{marsden2013introduction}, the left momentum map $J_L(\alpha_g) = T^*_eR_g \cdot \alpha_g$ is preserved by Noether's theorem, while the right momentum map $J_R(\alpha_g) = T^*_eL_g \cdot \alpha_g$ evolves according to a Poisson-Hamiltonian system on $P / G \cong \mathfrak{g}^*$, equipped with the $(-)$ {\em Lie-Poisson bracket}
\begin{align}\label{eq:lie-poisson}
    \{f, g\}_{\mathfrak{g}^*_-}(\mu) := -\left<\mu, \left[\frac{\delta f}{\delta \mu}, \frac{\delta g}{\delta \mu}\right]\right>,
\end{align}
for any $\mu \in \mathfrak{g}^*$. Here, we denoted by ${\delta f}/{\delta \mu}$ for $\mu \in \mathfrak{g}^*$ an element in $\mathfrak{g}$ satisfying
\begin{align}\label{eq:variational-derivative}
    Df(\mu) \cdot \alpha = \left<\alpha, \,\frac{\delta f}{\delta \mu}\right>_{\mathfrak{g}^* \times \mathfrak{g}},
\end{align}
for any $\alpha \in \mathfrak{g}^*$. Likewise, for any $\xi \in \mathfrak{g}$ and $f : \mathfrak{g} \rightarrow \mathbb{R}$, we can define the object ${\delta f}/{\delta \xi} \in \mathfrak{g}^*$ by swapping the roles of $\mathfrak{g}$ and $\mathfrak{g}^*$ around in \eqref{eq:variational-derivative}.
Note that \eqref{eq:lie-poisson} is precisely the bracket obtained by Poisson reduction, i.e.,
\begin{align*}
    \{f, g\}_{\mathfrak{g}^*_-} \circ J_R = \{f \circ J_R,\, g \circ J_R\} \,,
\end{align*}
with the right momentum map $J_R$ playing the role of the group projection $J_R : P \rightarrow P / G \cong \mathfrak{g}^*$ and where $\{\cdot, \cdot\}$ is the canonical Poisson bracket on $T^*G$.
Defining the reduced Hamiltonians $h_i$ by $h_i \circ J_R = H_i$ for all $i \geq 0$, \eqref{eq:stochastic-poisson-system} implies
\begin{align*}
    \diff f(\mu_t) &= - \sum_{i \geq 0} \left<\mu_t, \left[\frac{\delta f}{\delta \mu_t}, \frac{\delta h_i}{\delta \mu_t}\right]\right> \circ \diff S_t^i = \sum_{i \geq 0} \left<\mu_t, \,\ad_{{\delta h_i}/{\delta \mu_t}}\frac{\delta f}{\delta \mu_t}\right> \circ \diff S_t^i 
    \\
    &= \sum_{i \geq 0} \left<\ad^*_{{\delta h_i}/{\delta \mu_t}} \mu_t, \,\frac{\delta f}{\delta \mu_t}\right> \circ \diff S_t^i \,.
\end{align*}
Then, by the Stratonovich chain rule, we see that this is satisfied by the {\em stochastic Lie-Poisson system}
\begin{align}\label{eq:stochastic-lie-poisson}
    \diff \mu_t = \sum_{i \geq 0} \ad^*_{{\delta h_i}/{\delta \mu_t}} \mu_t \circ \diff S_t^i.
\end{align}
\begin{remark}
    The symplectic leaves of the Poisson manifold $\mathfrak{g}^*_-$ are precisely the connected components of the coadjoint orbits $\mathcal{O}_{\mu_0} := \{\Ad^*_{g}\mu_0 : g \in G\}$ (see \cite[§14.3]{marsden2013introduction}) and these are equipped with the Kirillov-Kostant-Souriau (KKS) symplectic form $\omega_\mu(\xi_{\mathfrak{g}^*}(\mu), \eta_{\mathfrak{g}^*}(\mu)) = -\left<\mu, [\xi, \eta]\right>$.
\end{remark}

\subsection{Euler-Poincar\'e reduction}\label{sec:euler-poincare-reduction}
Next, we consider the problem of symmetry reduction from a Lagrangian perspective when the configuration space is given by a Lie group $G$, to derive a stochastic counterpart to Euler-Poincar\'e reduction theory \cite{HMR1998}.
Effectively, this yields a variational principle from which we can deduce the system \eqref{eq:stochastic-lie-poisson}. In the deterministic case, Euler-Poincar\'e reduction applies in settings where the Lagrangian $L : TG \rightarrow \mathbb{R}$ is left (or right) $G$-invariant, i.e., $L(TL_{h} (g_t, \dot{g}_t)) = L(g_t, \dot{g}_t)$ for any $(g_t, \dot{g}_t) \in TG$ and $h \in G$. In particular, taking $h = g_t^{-1}$, we have $v_t := TL_{h} (g_t, \dot{g}_t) \in T_eG \cong \mathfrak{g}$, which allows us to define the {\em reduced Lagrangian} $\ell : \mathfrak{g} \rightarrow \mathbb{R}$, given by $\ell(v_t) = L(g_t, \dot{g}_t)$. Taking arbitrary variations of a curve $g_t \in G$ lead to a constrained variation on the reduced process $v_t \in \mathfrak{g}$ satisfying
\begin{align}\label{eq:lin-constraint}
    \delta v_t = \dot{\eta}_t + \ad_{v_t} \eta_t \,,
\end{align}
which is the so-called {\em Lin constraint} \cite{cendra1987lin}, where the process $\eta_t := T_{g_t}L_{g_t^{-1}} \delta g_t \in \mathfrak{g}$ can be taken as arbitrary. Then, Hamilton's principle \eqref{eq:hamilton's-principle} yields
\begin{equation}
	0 = \delta \int_{t_0}^{t_1} L(g_t, \dot g_t) \,\diff t = \delta \int_{t_0}^{t_1} \ell(v_t) \,\diff t = \int_{t_0}^{t_1} \left<\frac{\delta \ell}{\delta v_t}, \,\delta v_t\right>_{\mathfrak{g}^* \times \mathfrak{g}} \,\diff t \,,
\end{equation}
which, upon using the Lin constraint \eqref{eq:lin-constraint} and the fundamental lemma of the calculus of variations, yields \eqref{eq:stochastic-lie-poisson} with $\mu_t = \frac{\delta \ell}{\delta v_t}$, $h_i = 0$ for all $i \geq 1$ and $S_t^0 = t$.
Now, extending this reduction procedure to the stochastic setting is not entirely trivial as we cannot necessarily make sense of the process $\dot{g}_t$ due to a lack of time-differentiability of the process.
However, we will show that reduction can be achieved soundly through the Hamilton Pontryagin approach, paralleling the procedures in the deterministic setting \cite{YOSHIMURA2007381}.

\subsubsection{Reduction of the stochastic Hamilton-Pontryagin action by $G$}

The lifted actions of a group on its tangent and cotangent bundles are central to the concept of reducing a system by its symmetries. Indeed, we may think of such actions as providing a mechanism for deducing a Lagrangian on the Lie algebra $\mathfrak{g}$ from that defined on the tangent bundle of the group, $TG$. In the deterministic case, only the lifted action of $G$ on $TQ$ is necessary\footnote{This action corresponds to the Maurer-Cartan form, which is a $\mathfrak{g}$-valued one form on $G$.}. In the stochastic case, we must reduce a stochastic action integral defined on the Pontryagin bundle. To do so, we must first define the left group action $G \times TG\oplus T^*G \rightarrow TG\oplus T^*G$. An element $h\in G$ acts from the left on a curve $(g,V,p) \in TG\oplus T^*G$, by
\begin{equation}\label{eqn:Pontryagin_left_action}
    h\cdot(g,V,p) = (h\cdot g, \, T_{g}L_{h}V, \,T^*_{h\cdot g}L_{h^{-1}}p) \,,
\end{equation}
where dot denotes the group operation, $L$ denotes left translation, and the tangent and cotangent lifted actions are as defined in equations \eqref{eqn:tangent_lifted_action} and \eqref{eqn:cotangent_lifted_action}. Notice that the action on the element of the cotangent bundle is by the inverse, this serves to ensure that fibres of the bundle are mapped to other fibres of the bundle. Indeed, notice that the action as defined by equation \eqref{eqn:Pontryagin_left_action} maps the fibre $T_gG\oplus T_g^*G$ to another fibre $T_{h\cdot g}G \oplus T_{h\cdot g}^*G$. 


Now, recall the stochastic Hamilton-Pontryagin action functional from Theorem \ref{thm:stochastic_HP}, which reads
\begin{align*}
    \mathcal{S}[g_t, V_t, p_t] = \int_{t_0}^{t_1} L(g_t, V_t)\, \diff t + \sum_{i \geq 1} \Gamma_i(g_t)\circ \diff S_t^i + \scp{p_t}{\circ\,\diff g_t - V_t \, \diff t - \sum_{i \geq 1} \Xi_i(g_t) \circ \diff S_t^i}_{T^*Q \times TQ},
\end{align*}
for $(g_t, V_t, p_t) \in TG \oplus T^*G$. In particular, taking $\Xi_i(g_t)$ to be a left-invariant vector field, i.e., $\Xi_i(g_t) =T_eL_{g_t} \xi_i$ for some $\xi_i \in \mathfrak{g}$, and $\Gamma_i \equiv 0$\footnote{We chose $\Gamma_i \equiv 0$  here since the only functions $\Gamma_i : G \rightarrow \mathbb{R}$ that are $G$-invariant are the constant functions, which have no effect on the dynamics.} for all $i \geq 1$, one may naively consider a reduced action functional on $G \times \mathfrak{g} \times \mathfrak{g}^*$ of the form
\begin{align}\label{eq:reduced-HP-action}
    \mathcal{S}_{red.}[g_t, v_t, \mu_t] = \int_{t_0}^{t_1} \ell(v_t)\, \diff t + \sum_{i \geq 1} \scp{\mu_t}{T_{g_t}L_{{g_t}^{-1}}(\circ\,\diff g_t) - v_t \, \diff t - \sum_{i \geq 1} \xi_i \circ \diff S_t^i}_{\mathfrak{g}^* \times \mathfrak{g}},
\end{align}
where $(e, v_t, \mu_t) := g^{-1}_t \cdot (g_t, V_t, p_t) \in G \times \mathfrak{g} \times \mathfrak{g}^*$, adopting our notation in \eqref{eqn:Pontryagin_left_action}. Note that at this stage, we do not understand what the term $T_{g_t}L_{{g_t}^{-1}}(\circ\,\diff g_t)$ means in the expression, as $\diff g_t$ is not a well-defined object on $TG$. However, by assuming that the process $g_t$ is compatible with the driving semimartingale $S_t = (S_t^0, S_t^1, \ldots)$ so that there exists a family $\{F_t^i\}_{i \geq 0}$ of $TG$-valued semimartingales such that $\diff g_t = \sum_{i \geq 0} F_t^i \circ \diff S_t^i$ (see Definition \ref{def:compatibility}), we can {\em define} $T_{g_t}L_{{g_t}^{-1}}(\circ\,\diff g_t)$ by
\begin{equation}\label{eq:define-g-inv-dg}
	\int^T_0 \left<\alpha_t, T_{g_t}L_{{g_t}^{-1}}(\circ \,\diff g_t)\right>_{\mathfrak{g}^* \times \mathfrak{g}} := \sum_{i\geq 0} \int^T_0 \left<\alpha_t, T_{g_t}L_{{g_t}^{-1}}F_t^i\right>_{\mathfrak{g}^* \times \mathfrak{g}} \circ \diff S_t^i \,,
\end{equation}
for any $\mathfrak{g}^*$-valued process $\alpha_t$ and $T > 0$. This is well-defined by the uniqueness of the decomposition of $g_t$ into the processes $\{F_t^i\}_{i \geq 0}$ (Corollary \ref{cor:uniqueness-of-decomposition}).
Thus, the compatibility assumption allows us to make sense of the reduced action functional \eqref{eq:reduced-HP-action} and one can easily check that $\mathcal{S}[g_t, V_t, p_t] = \mathcal{S}_{red.}[g_t, v_t, \mu_t]$. In the following, we derive a stochastic extension of the Lin constraint \eqref{eq:lin-constraint}, enabling us to understand how taking variations on the curve $g_t$ affect the term involving $T_{g_t}L_{{g_t}^{-1}}(\circ\,\diff g_t)$ in the reduced action.
    

\subsubsection{Stochastic Lin constraints}

To derive a stochastic generalisation of the Lin constraint \eqref{eq:lin-constraint}, we adopt an approach similar to that considered in \cite{ACC2014}.
To this end, let $\eta_t \in \mathfrak{g}$ be an arbitrary, deterministic, $C^1$-differentiable curve that vanishes at the endpoints, i.e., $\eta_{t_0} = \eta_{t_1} = 0$ for fixed $0 < t_0 < t_1 < \infty$. Furthermore, let us define a group-valued process, $e_{\epsilon,t} \in G$, depending on $\epsilon \in \mbb{R}$, by the ODE
\begin{equation}\label{eq:group-perturbation-ode}
    \frac{\diff e_{\epsilon,t}}{\diff t} = \epsilon\, T_eL_{e_{\epsilon,t}}\dot{\eta}_{t} \,,\quad\hbox{such that}\quad e_{\epsilon,0} = e \,.
\end{equation}
We consider perturbations of a path $g_t\in G$ by this process as
\begin{align}\label{eq:g-perturbation}
g_{\epsilon,t} := g_t\cdot e_{\epsilon,t},
\end{align}
giving us the infinitesimal variation
\begin{equation}\label{eq:delta-g}
    \delta g_t := \frac{\p}{\p\epsilon}\bigg|_{\epsilon=0}g_{\epsilon,t} = T_eL_{g_t}\eta_t \,.
\end{equation}
This follows from the fact that at $\epsilon=0$, the perturbation behaves like the exponential map, indeed
\begin{equation}\label{eq:variation-of-perturbation}
    \frac{\p}{\p\epsilon}\bigg|_{\epsilon=0}e_{\epsilon,t} = \eta_t\,,\quad\hbox{and}\quad \frac{\p}{\p\epsilon}\bigg|_{\epsilon=0}e_{\epsilon,t}^{-1} = -\eta_{t} \,.
\end{equation}
For a proof of this, see \cite[Lemma 3.1]{ACC2014}.
It is important to note that the procedure \eqref{eq:g-perturbation} to construct a perturbation of $g_t$ is fully deterministic, thus we have successfully created a variational family of paths $\{g_{\epsilon,t}\}_{\epsilon \in \mathbb{R}}$ such that each member is $\mathcal{F}_t$-adapted {\em and} satisfies the endpoint conditions $g_{\epsilon, t_0} = g_{t_0}, g_{\epsilon, t_1} = g_{t_1}$ for all $\epsilon \in \mathbb{R}$.
Further, the following Lemma reveals the dynamics of this perturbed process in the group.
\begin{lemma}\label{lemma:variation_g_evolution}
    Let $g_t\in G$ be a group valued process which is compatible with a driving semimartingale, $S_t$, so that there exists a family of $\mathfrak{g}$-valued semimartingales $\{w_t^i\}_{i \geq 0}$ satisfying
    \begin{equation}\label{eqn:compatibility_g_t}
        T_{g_t}L_{{g_t}^{-1}}(\circ \,\diff g_t) = \sum_{i\geq 0}w_t^i\circ \diff S_t^i \,
    \end{equation}
    (to see this, take $w_t^i := T_{g_t}L_{{g_t}^{-1}}(F_t^i)$ in \eqref{eq:define-g-inv-dg}).
    Then, for the group-valued perturbation $e_{\epsilon,t}$ introduced above, the process $g_{\epsilon,t} = g_t\cdot e_{\epsilon,t}$ is compatible with the same driving semimartingale and evolves according to
    \begin{equation}\label{eqn:evolution_g_epsilon_t}
        \diff g_{\epsilon,t} = T_eL_{g_{\epsilon,t}}\left(\Ad_{e_{\epsilon,t}^{-1}}w_t^0 + \epsilon \dot{\eta}_t \right) \diff t + \sum_{i\geq 1}T_eL_{g_{\epsilon,t}}\left(\Ad_{e_{\epsilon,t}^{-1}}w_t^i \right)\circ \diff S_t^i \,, 
    \end{equation}
    where $\Ad$ denotes the adjoint representation of $G$.
\end{lemma}
\begin{proof}
    The proof for this result can be found in Appendix \ref{app:variation_g_evolution}, and is related to a similar result discussed in \cite{ACC2014}.
\end{proof}

\begin{remark}
    We note that the above lemma is the only place where we must require the assumption $S_t^0 = t$ in Definition \ref{def:driving_semimartingale}. Otherwise, it suffices to have that $S_t^0$ is a strictly increasing c\`adl\`ag adapted finite-variation process in the unreduced Lagrangian setting, or any semimartingale in the Hamiltonian formulation. In general, most of our assumptions arise to make rigorous sense of the variational principles.
\end{remark}

This allows us to deduce a stochastic generalisation of the Lin constraint \eqref{eq:lin-constraint}, as we show in the following result.
\begin{corollary}[Stochastic Lin constraints]\label{cor:stoch-lin-constraints}
    Let $g_t \in G$ be a curve that is compatible with a driving semimartingale $S_t$ and consider a family of $\epsilon$-perturbed curves $g_{\epsilon,t} = g_t\cdot e_{\epsilon,t}$, where $e_{\epsilon,t}$ is a group perturbation as defined in \eqref{eq:group-perturbation-ode}. Then, taking variations of $g_t$ among the family $\{g_{\epsilon,t}\}_{\epsilon \in \mathbb{R}}$ induces constrained variations on the processes $\{w_t^i\}_{i \geq 0}$ in \eqref{eqn:compatibility_g_t}, satisfying the relations:
    \begin{align}
        \delta w_t^0 &= \ad_{w_t^0}\eta_t + \dot{\eta}_t \label{eq:stoch-lin-0}
        \,,\\
        \delta w_t^i &= \ad_{w_t^i}\eta_t
        \,, \label{eq:stoch-lin-i}
    \end{align}
    where $\ad$ denotes the adjoint representation of $\mathfrak{g}$. We call \eqref{eq:stoch-lin-0}--\eqref{eq:stoch-lin-i} the stochastic Lin constraints.
\end{corollary}
\begin{proof}
    The result follows almost immediately from Lemma \ref{lemma:variation_g_evolution}. Indeed, since $g_{\epsilon,t}$ is compatible with the driving semimartingale, it can be expressed uniquely as $\diff g_{\epsilon,t} = \sum_{i\geq 0} T_eL_{g_{\epsilon,t}} w_{\epsilon,t}^i\circ \diff S_t^i$. Then from Lemma \ref{lemma:variation_g_evolution}, we have
    \begin{equation*}
        \sum_{i\geq 0} \int_{t_0}^{t_1}T_eL_{g_{\epsilon,t}} w_{\epsilon,t}^i\circ \diff S_t^i = \int_{t_0}^{t_1}T_eL_{g_{\epsilon,t}}\left(\Ad_{e_{\epsilon,t}^{-1}}w_t^0 + \epsilon \dot{\eta}_t \right) \diff t + \sum_{i\geq 1} \int_{t_0}^{t_1} T_eL_{g_{\epsilon,t}}\left(\Ad_{e_{\epsilon,t}^{-1}}w_t^i \right)\circ \diff S_t^i
        \,,
    \end{equation*}
    Hence, matching terms and taking variations, we have
    \begin{align*}
        \delta w_t^0 &:= \frac{\p}{\p\epsilon}\bigg|_{\epsilon=0}w_{\epsilon,t}^0 = \frac{\p}{\p\epsilon}\bigg|_{\epsilon=0}\left(\Ad_{e_{\epsilon,t}^{-1}} w_t^0 + \epsilon \dot{\eta}_t \right) = \ad_{w_t^0}\eta_t + \dot{\eta}_t
        \,,\\
        \delta w_t^i &:= \frac{\p}{\p\epsilon}\bigg|_{\epsilon=0}w_{\epsilon,t}^i = \frac{\p}{\p\epsilon}\bigg|_{\epsilon=0}\Ad_{e_{\epsilon,t}^{-1}}w_t^i = \ad_{w_t^i}\eta_t
        \,,
    \end{align*}
    where we used the second relation in \eqref{eq:variation-of-perturbation} to deduce the last equality.
\end{proof}


\subsubsection{The Euler-Poincar\'e reduction.}
We are now ready to state and prove the following stochastic extension of Euler-Poincar\'e reduction.
\begin{theorem}[The stochastic Euler-Poincar\'e reduction theorem]\label{thm:EP}
	Let $S_t$ be a driving semimartingale and assume we have a left-invariant Lagrangian, $L:TG\rightarrow \mathbb{R}$ and a collection of vectors $\xi_i \in \mathfrak{g}$ for $i \geq 1$.
    Furthermore, let $\{\Xi_i\}_{i \geq 1}$ be a family of left-invariant vector fields on $G$, given by $\Xi_i(g) = T_eL_g\xi_i$. Then the following are equivalent:
	\begin{enumerate}
	 \item The curve $(g_t, V_t, p_t) \in TG \oplus T^*G$ is an extrema of the unreduced stochastic Hamilton-Pontryagin action functional \eqref{eq:HP-action} with $\Gamma_i \equiv 0$ and $\Xi_i(g) = T_eL_g\xi_i$.
	 \item The stochastic Euler-Lagrange equations \eqref{eq:stoch-euler-lagrange-mom}--\eqref{eq:stoch-euler-lagrange-pos} hold with $\Gamma_i \equiv 0$ and $\Xi_i(g) = T_eL_g\xi_i$.
	 \item The curve $(g_t, v_t, \mu_t) \in G\times \mathfrak{g} \times \mathfrak{g}^*$ is a solution to the reduced Hamilton-Pontryagin principle
    \begin{equation}\label{eqn:reduced_stochastic_HP_action}
	 	0 = \delta\int_{t_0}^{t_1}\ell(v_t)\, \diff t + \scp{\mu_t}{T_{g_t}L_{{g_t}^{-1}}(\circ \diff g_t) - v_t \circ \diff S_t^0 - \sum_{i\geq 1}\xi_i\circ \diff S_t^i} \,,
	 \end{equation}
     where the variations of $v_t$ and $\mu_t$ are taken arbitrarily and the variations of $g_t$ are taken according to Corollary \ref{cor:stoch-lin-constraints}, i.e. among the family of $\epsilon $-perturbations $g_{\epsilon, t} = g_t \cdot e_{\epsilon, t}$, with $e_{\epsilon, t}$ solving \eqref{eq:group-perturbation-ode}.
	 \item The following stochastic Euler-Poincar\'e equation holds
	 \begin{equation}\label{eqn:stochastic_EP}
	 	\diff \frac{\delta\ell}{\delta v_t} = \ad^*_{v_t}\frac{\delta\ell}{\delta v_t}\, \diff t + \sum_{i \geq 1} \ad^*_{\xi_i}\frac{\delta\ell}{\delta v_t}\circ \diff S_t^i \,,
	 \end{equation}
     where $\ad^*$ is the dual of $\ad$ with respect to the natural pairing between $\mathfrak{g}$ and its dual space $\mathfrak{g}^*$.
	\end{enumerate}
\end{theorem}

\begin{proof}
	The equivalence between \emph{1} and \emph{2} follows from Theorem \ref{thm:stochastic_HP}. We have also shown earlier that the actions in \emph{3} and \emph{1} are equivalent. Indeed, the Lagrangian is $G$-invariant and the constraint terms in \eqref{eq:HP-action} and \eqref{eqn:reduced_stochastic_HP_action} are equivalent, with $\mu = T^*_eL_{g}p$, $v = T_gL_{g^{-1}}V$ and $\xi_i = T_gL_{g^{-1}}\Xi_i$. It remains to demonstrate that the Euler-Poincar\'e equation \eqref{eqn:stochastic_EP} results from the reduced Hamilton-Pontryagin principle \eqref{eqn:reduced_stochastic_HP_action}. Taking variations of the action \eqref{eqn:reduced_stochastic_HP_action} and applying the result of Corollary \ref{cor:stoch-lin-constraints}, we have
\begin{align*}
    0 &= \delta\int_{t_0}^{t_1}\ell(v_t)\, \diff t + \int_{t_0}^{t_1} \scp{\mu_t}{T_{g_t}L_{{g_t}^{-1}}(\circ\diff g_t) - v_t\circ \diff S_t^0 - \sum_{i\geq 1}\xi_i\circ \diff S_t^i}
    \\
    &= \int_{t_0}^{t_1} \scp{\frac{\delta\ell}{\delta v_t}}{\delta v_t}\diff t + \int_{t_0}^{t_1} \scp{\mu}{(\delta w^0_t - \delta v_t)\, \diff t + \sum_{i\geq 1}\delta w^i_t\circ \diff S_t^i}
    \\
    &\qquad + \int_{t_0}^{t_1}\scp{\delta\mu_t}{(w^0_t - v_t) \, \diff t + \sum_{i\geq 1}(w^i_t -\xi_i)\circ \diff S_t^i}
    \\
    &= \int_{t_0}^{t_1} \scp{\frac{\delta\ell}{\delta v_t} - \mu_t}{\delta v_t} \diff t + \int_{t_0}^{t_1} \scp{\mu_t}{\left(\ad_{w_t^0}\eta_t + \dot{\eta}_t \right) \, \diff t + \sum_{i\geq 1} \ad_{w_t^i}\eta_t \circ \diff S_t^i}
    \\
    &\qquad + \int_{t_0}^{t_1}\scp{\delta\mu_t}{(w^0_t - v_t) \,\diff t + \sum_{i\geq 1}(w^i_t -\xi_i)\circ \diff S_t^i}
    \\
    &=\int_{t_0}^{t_1} \scp{\frac{\delta\ell}{\delta v_t} - \mu_t}{\delta v_t} \diff t - \int_{t_0}^{t_1} \scp{\circ\diff \mu_t - \ad^*_{w_t^0}\mu_t \, \diff t - \sum_{i\geq 1}\ad^*_{w_t^i}\mu_t \circ \diff S_t^i}{\eta_t}
    \\
    &\qquad + \int_{t_0}^{t_1}\scp{\delta\mu_t}{(w^0_t - v_t)\, \diff t + \sum_{i\geq 1}(w^i_t -\xi_i)\circ \diff S_t^i}
    \,,
\end{align*}
where we used integration-by-parts $\int^{t_1}_{t_0} \left<\mu_t, \dot{\eta}_t\right> \, \diff t = \left<\mu_t, \eta_t\right>\big|^{t_1}_{t_0} -\int^{t_1}_{t_0} \left<\circ \diff\mu_t, \eta_t\right> \, \diff t$ and the endpoint conditions $\eta_{t_0} = \eta_{t_1} = 0$ to arrive at the last line.
Thus, invoking the fundamental lemma of the stochastic calculus of variations (Lemma \ref{lemma:fundamental}) together with Corollary \ref{cor:fundamental_lemma} yields the following system
\begin{equation*}
    \diff \mu_t = \ad^*_{w_t^0}\mu_t \, \diff t + \sum_{i\geq 1}\ad^*_{w_t^i}\mu_t \circ \diff S_t^i 
    \,,\quad\hbox{where}\quad 
    \mu_t = \frac{\delta\ell}{\delta v_t} \,,\quad w^0_t = v_t \,,\quad\hbox{and}\quad w^i_t = \xi_i \,.
\end{equation*}
These are equivalent to the Euler-Poincar\'e equation \eqref{eqn:stochastic_EP}, and hence we have proven our claim.
\end{proof}

While we have only considered the case where the configuration space is given exactly by the symmetry group $Q = G$, we can also consider an extension of the theorem to the setting where the configuration space is given by a semidirect product $Q = G \ltimes V$, where $V$ is some vector space on which an action by $G$ is defined. This models systems with {\em broken symmetries}, as exemplified by the heavy top system, where the full $SO(3)$-symmetry of the rigid body is broken by gravity.
We discuss this further in Appendix \ref{app:semidirect}.

\section{Stochastic dissipative systems with symmetry}\label{sec:dissipation}

Hereafter, we restrict to the case $S_t^i = W_t^i$ for all $i \geq 1$, where $\{W_t^i\}_{i \geq 1}$ are i.i.d. Brownian motion. The goal of this section is to develop a principled framework for adding dissipation into the stochastic Hamiltonian systems, considered in the previous sections. Dissipation in physical systems typically occurs when it is coupled to an external environment (referred to as the {\em heat bath}), which extracts mechanical energy from the system and converts it to {\em heat} -- a measure of how much noise there is in the environment. 
We use this duality between noise and dissipation as our basis for introducing dissipation into our system.
To this end, assume that the temperature of the heat bath is fixed. Then a standard result in statistical mechanics states that the system is distributed according to the {\em Gibbs measure} at thermodynamic equilibrium \cite{souriau1970structure}, defined below.
\begin{definition}\label{def:symplectic-gibbs-measure}
Given a $2n$-dimensional symplectic manifold $(P, \omega)$ and a Hamiltonian $H_0 : P \rightarrow \mathbb R$, we define the {\em Gibbs measure} $\mathbb P_\infty$ on $P$ as a probability measure
\begin{align}\label{eq:gibbs-meas}
\mathbb P_\infty = Z^{-1} e^{-\beta H_0} |\,\omega^n|, \quad Z = \int_P e^{-\beta H_0} \mathrm{d} |\,\omega^n| \,,
\end{align}
where $\omega^n$ is the symplectic volume form, $|\,\omega^n|$ the corresponding measure (i.e., the Liouville measure) and $\beta > 0$ is the inverse temperature.
\end{definition}

This is characterised equivalently as the maximal entropy probability measure among all configurations with fixed average energy, in accordance with the second law of thermodynamics.

\begin{proposition}[Maximum entropy principle]\label{prop:max-entropy-principle}
    The Gibbs measure \eqref{eq:gibbs-meas} is a solution to the following constrained variational principle
    \begin{align}
        \mathbb{P}_\infty =
        \argmax_{\mu >\!\!> \lambda} \left\{-H(\mu | \lambda)\right\}, \, \text{ such that } 
        \, \int_P H_0 \,\diff \mu = c \, \text{ and }
        \, \int_P \diff \mu = 1 \,.
    \end{align}
    Here, $c < \infty$ is some real constant and $H(\mu | \lambda) := \int_P \frac{\diff \mu}{\diff \lambda} \log \frac{\diff \mu}{\diff \lambda} \diff \lambda$ is the relative entropy of the measure $\mu$ with respect to a reference measure $\lambda$. In particular, we take $\lambda = |\omega^n|$ as the reference measure on $P$. The first constraint fixes the average energy level of the system and the second constraint ensures that $\mu$ is a probability measure.
\end{proposition}
\begin{proof}
    See Appendix \ref{app:max-entropy-principle}.
\end{proof}

In order to seek for an appropriate dissipative mechanism that is compatible with the noise introduced in the previous section, we wish to identify a dissipative vector field $X \in \mathfrak{X}(P)$ (i.e., those that satisfy $\mathcal{L}_X H_0 < 0$) such that when added to the symplectic diffusion \eqref{eq:general-symplectic-diffusion}, the resulting process preserves the measure \eqref{eq:gibbs-meas}.
Such a vector field can be identified by the following result.

\begin{theorem}[\cite{arnaudon2019irreversible}, Theorem 1] \label{thm:gibbs-measure-preservation}
A dissipative extension to the stochastic Hamiltonian system \eqref{eq:general-symplectic-diffusion} of the form
\begin{align} \label{eq:symplectic-langevin-eq}
\diff Z_t =  X_{H_0}(Z_t) \,\diff t \underbrace{-\frac{\beta}{2}\sum_{i \geq 1}\{H_0,H_i\}X_{H_i}(Z_t) \,\diff t}_{\mathrm{Dissipation}}  \underbrace{+ \sum_{i\geq 1} X_{H_i}(Z_t)\circ \diff W
_t^i}_{\mathrm{Noise}}
\end{align}
preserves the Gibbs measure \eqref{eq:gibbs-meas} on $(P, \omega)$.
\end{theorem}


We call \eqref{eq:symplectic-langevin-eq} the {\em symplectic Langevin equation} as it can be viewed as a generalisation of the underdamped Langevin equation on Euclidean space to symplectic manifolds (to see this, on $P = T^*\mathbb{R}^n$, take $H_i(q, p) = q^i$ for $i = 1, \ldots, n$ to obtain the underdamped Langevin equation on $\mathbb{R}^n$).
One further hopes to prove that such a process is {\em ergodic}, that is, the invariant measure is unique and its statistics coincide with the statistics of the system in the infinite time limit $T \rightarrow \infty$. 
In general, proving the ergodicity of a stochastic process is challenging and case-dependent. However, there are known sufficient conditions on the drift and diffusion vector fields, such as the H\"ormander condition, that one can use to demonstrate ergodicity of the process (see \cite[Section 10]{barp2021unifying} for more details).

\begin{remark}
While it appears to be important, the fact that the noise term in \eqref{eq:symplectic-langevin-eq} is Hamiltonian is in fact not necessary for the process to admit an invariant measure of the form \eqref{eq:gibbs-meas}. However, the dissipative term that is added to preserve the Gibbs measure becomes more complicated in the general case \cite{barp2021unifying}. We also note that when the noise terms are chosen to be Hamiltonian, then the dissipative term in \eqref{eq:symplectic-langevin-eq} is unique up to a topological term related to the $(n-1)$-th de Rham cohomology group of $P$ \cite{barp2021unifying}. This is trivial if $P$ is simply-connected.
\end{remark}

\begin{remark}
A similar result to Theorem \ref{thm:gibbs-measure-preservation} in the case of more general driving semimartingales $\{S_t^i\}_{i \geq 1}$ is considered in \cite{diamantakis2023variational} when $P=T^*G$ for some Lie group $G$. In this case, the invariant measure has an extra correction term arising from the non-trivial L\'evy area of the processes $\{S_t^i\}_{i \geq 1}$. We defer the investigation of the case of general driving semimartingales to future work.
\end{remark}

\subsection{Symmetry reduction}
We now consider stochastic dissipative systems with symmetries and derive the corresponding symmetry reduced equations alongside their invariant measures. This yields systems whose energy is exchanged with the surrounding heat bath, but the other conserved quantities are fixed.
Following the previous section, we consider symplectic and Poisson reduction theory for the symplectic Langevin equation \eqref{eq:symplectic-langevin-eq}. In the special case $P=T^*G$, we will further show that the corresponding symmetry reduced system yields a dissipative term that is identical to the double-bracket dissipation \cite{bloch1994dissipation, bloch1996euler}. To the best of our knowledge, this provides the first such derivation of the double-bracket dissipation from purely mechanistic principles.

For convenience, let us re-express the symplectic Langevin equation \eqref{eq:symplectic-langevin-eq} in Poisson form, i.e.,
\begin{align} \label{eq:poisson-langevin-eq}
    \diff F(Z_t) = \{F, H_0\}(Z_t) \,\diff t -\frac{\beta}{2}\sum_{i \geq 1} \{H_0,H_i\}\{F, H_i\}(Z_t) \,\diff t + \sum_{i \geq 1} \{F, H_i\}(Z_t)\circ \diff W_t^i,
\end{align}
for any $F \in \Omega^0(P)$, where $\{\cdot, \cdot\}$ is the Poisson bracket induced by the symplectic form $\omega$.
In the presence of symmetries, the reduction theory for the Langevin system \eqref{eq:poisson-langevin-eq} follows in a similar manner to the purely stochastic case studied in the previous section.

\subsubsection{Symplectic reduction}
We first consider a symplectic reduction theory for stochastic-dissipative systems, which states that in the presence of symmetries, the symplectic Langevin diffusion \eqref{eq:poisson-langevin-eq} drops down to a Langevin system on the reduced symplectic manifold $P_\mu$.
To this end, we provide the following extension of Noether's theorem to the stochastic-dissipative case, which is essential to symplectic reduction.

\begin{lemma}[Noether's theorem in the stochastic-dissipative case] \label{lemma:noether-stoch-diss}
    Under the same setting as Proposition \ref{prop:stoch-noether}, the stochastic-dissipative system \eqref{eq:poisson-langevin-eq} preserves the momentum map $J : P \rightarrow \mathfrak g^*$.
\end{lemma}
\begin{proof}
    Echoing the proof of Theorem \ref{prop:stoch-noether}, we define the quantity $J^\xi(Z_t) := \left<\xi,J(Z_t)\right>$ for a fixed $\xi \in \mathfrak g$. Under the dynamics \eqref{eq:poisson-langevin-eq}, this quantity evolves as
    \begin{align}
        \diff J^\xi(Z_t) = \{J^\xi, H_0\}(Z_t) \,\diff t -\frac{\beta}{2}\sum_{i \geq 1} \{H_0,H_i\}\{J^\xi, H_i\}(Z_t) \,\diff t + \sum_{i \geq 1} \{J^\xi, H_i\}(Z_t)\circ \diff W_t^i = 0 \,,
    \end{align}
    which follows from the fact that $\{J^\xi, H_i\} = 0$ for all $i \geq 0$. Since $\xi \in \mathfrak g$ was chosen arbitrary, the momentum map $J$ is preserved.
\end{proof}

Following a similar argument to that presented in Section \ref{sec:symplectic-reduction}, we obtain the following result.

\begin{proposition} \label{prop:symplectic-reduction-stoch-dissipative}
Consider the symplectic Langevin diffusion \eqref{eq:symplectic-langevin-eq} on the symplectic manifold $(P, \omega)$ such that all $\{H_i\}_{i \geq 0}$ are $G$-invariant. Then for $Z_0 \in J^{-1}(\mu) \subseteq P$, if $Z_t$ solves the symplectic Langevin system \eqref{eq:symplectic-langevin-eq},  then $z_t^\mu := \pi_\mu(Z_t) \in P_\mu$ solves
\begin{align} \label{eq:reduced-symplectic-langevin}
\diff z_t^\mu = X_{h_0}^\mu(z_t^\mu) \,\diff t -\frac{\beta}{2}\sum_{i \geq 1} \{h_0, h_i\}_{\mu}(z_t^\mu) \, X_{h_i}^\mu(z_t^\mu) \,\diff t  + \sum_{i \geq 1} X_{h_i}^\mu(z_t^\mu)\circ \diff W_t^i \,.
\end{align}
\end{proposition}
\begin{proof}
    See Appendix \ref{app:proof-symplectic-reduction-stoch-dissipative}.
\end{proof}

We note that since the form of the reduced equation \eqref{eq:reduced-symplectic-langevin} is essentially the same as its unreduced counterpart \eqref{eq:symplectic-langevin-eq}, Theorem \ref{thm:gibbs-measure-preservation} implies that a Gibbs measure in the reduced space is preserved under the reduced dynamics, which we state more precisely in the following.

\begin{corollary}\label{cor:reduced-gibbs-preservation}
    The Gibbs measure 
    \begin{align}\label{eq:gibbs-meas-reduced}
    \mathbb P_\infty^\mu = Z_\mu^{-1} e^{-\beta h_0} |\,\omega^{n'}_\mu|, \quad Z_\mu = \int_{P_\mu} e^{-\beta h_0} \mathrm{d} |\,\omega^{n'}_\mu|,
    \end{align}
    defined on the $2n'$-dimensional space $P_\mu$ is preserved under the reduced Langevin system \eqref{eq:reduced-symplectic-langevin}.
\end{corollary}

\subsubsection{Poisson reduction} \label{sec:poisson-rediction-dissipation}
We also present a Poisson reduction result for the symplectic Langevin system \eqref{eq:symplectic-langevin-eq} by projecting its dynamics down to the space $P/G$. Defining the projection map $\pi : P \rightarrow P/G$ and assuming again that $\{H_i\}_{i \geq 0}$ are $G$-invariant so that the reduced Hamiltonians $h_i \circ \pi = H_i$ are well-defined, we can verify that (using the same argument as in Section \ref{sec:poisson-reduction}) for $Z_t$ solving \eqref{eq:symplectic-langevin-eq}, the projected flow $[Z_t] := \pi(Z_t)$ solves the system
\begin{align} \label{eq:reduced-poisson-langevin}
\begin{split}
    \diff f([Z_t]) = \{f, h_0\}_{P/G}([Z_t]) \diff t - \frac{\beta}{2} &\sum_{i \geq 1} \{h_0, h_i\}_{P/G} \,\{f, h_i\}_{P/G}([Z_t]) \diff t \\
    + &\sum_{i \geq 1} \{f, h_i\}_{P/G}([Z_t]) \circ \diff W_t^i.
\end{split}
\end{align}
From Remark \ref{rmk:symplectic-leaf-preservation}, the symplectic leaves of $P/G$ (i.e., the reduced spaces $P_\mu$) are preserved under \eqref{eq:reduced-poisson-langevin} and moreover its restricted dynamics is given by \eqref{eq:reduced-symplectic-langevin}. Therefore from Corollary \ref{cor:reduced-gibbs-preservation}, the reduced system \eqref{eq:reduced-poisson-langevin} preserves the Gibbs measure on each symplectic leaf.

\subsection{Double-bracket dissipation} \label{sec:double-bracket-dissipation}
Now let $P = T^*G$, where $G$ is an $n$-dimensional Lie group and the Hamiltonians $H_i : T^*G \rightarrow \mathbb{R}$ for all $i \geq 0$ are left-invariant under the cotangent-lifted action of $G$. The reduced space $P/G$ can thus be identified with $\mathfrak{g}^*$, which carries the $(-)$ Lie-Poisson bracket \eqref{eq:lie-poisson}. 
As before, letting $J_R$ be the right momentum map, we define the reduced Hamiltonians $h_i$ by $h_i \circ J_R = H_i$ for all $i$. In particular, we choose the noise Hamiltonians $\{H_i\}_{i \geq 1}$ to be of the particular form $H_i(\alpha_g) = \sigma \left< J_R(\alpha_g), \xi_i\right>$ for $i=1, \ldots, n$ and $H_i(\alpha_g) \equiv 0$ for $i > n$. Therefore, we have $h_i(\mu) = \sigma \left<\mu, \xi_i\right>$, where $\{\xi_i\}_{i=1}^n$ is an orthonormal basis of $\mathfrak{g}$ with respect to an inner product $\gamma : \mathfrak{g} \times \mathfrak{g} \rightarrow \mathbb{R}$ and $\sigma > 0$ is a constant.

\begin{remark}
    We note that the particular noise Hamiltonians chosen above is a natural choice. To see this, recall that $J_R(\alpha_g) = T_e^*L_g \cdot \alpha_g$, which gives us
    \begin{equation*}
        H_i(\alpha_g) = \sigma \left< J_R(\alpha_g), \xi_i\right> = \sigma \left< T_e^*L_g \cdot \alpha_g, \xi_i\right> = \sigma \left<\alpha_g, T_eL_g \cdot \xi_i\right> \,.
    \end{equation*}
    From the relation \eqref{eq:legendre-transform-Gamma}, we see that this choice of noise Hamiltonian corresponds to the ``trivial" case $\Gamma_i(g) \equiv 0$ and $\Xi_i(g) = \sigma T_eL_g \cdot \xi_i$, where the latter is precisely the left-invariant vector fields on $G$, scaled by a constant $\sigma > 0$. The corresponding noise process $\diff g_t = \sum_{i=1}^n \Xi_i(g_t) \circ \diff W_t^i$ is equivalent to a right-invariant Brownian motion on $G$ \cite[Theorem 1]{ito1950brownian}.
\end{remark}

For any $f \in C^\infty(\mathfrak{g}^*; \mathbb{R})$, the system \eqref{eq:reduced-poisson-langevin} reads
\begin{align}\label{eq:stoch-diss-lie-poisson}
    \diff f(\mu_t) = \{f, h_0\}_{\mathfrak{g}^*_-}(\mu_t) \diff t - \frac{\beta}{2}\sum_{i=1}^n \{h_0, h_i\}_{\mathfrak{g}^*_-} \,\{f, h_i\}_{\mathfrak{g}^*_-}(\mu_t) \diff t + \sum_{i=1}^n \{f, h_i\}_{\mathfrak{g}^*_-}(\mu_t) \circ \diff W_t^i,
\end{align}
for $\mu_t \in \mathfrak{g}$, where we recall the definition of the ($-$) Lie-Poisson bracket $\{\cdot, \cdot\}_{\mathfrak{g}^*_-}$ in \eqref{eq:lie-poisson}. Using the identity
\begin{align}\label{eq:bracket-f-hi}
    \{f, h_i\}_{\mathfrak{g}^*_-}(\mu_t) = -\left<\mu, \left[\frac{\delta f}{\delta \mu}, \frac{\delta h_i}{\delta \mu}\right]\right> = \left<\ad_{\delta h_i / \delta \mu}^*\mu, \frac{\delta f}{\delta \mu}\right>,
\end{align}
the dissipative term in \eqref{eq:stoch-diss-lie-poisson} can be expressed as
\begin{align*}
    - \frac{\beta}{2}\sum_{i=1}^n \{h_0, h_i\}_{\mathfrak{g}^*_-} \,\{f, h_i\}_{\mathfrak{g}^*_-}(\mu_t) &\stackrel{\eqref{eq:bracket-f-hi}}{=} \frac{\beta}{2}\sum_{i=1}^n \left<\mu_t, \left[\frac{\delta h_0}{\delta \mu}, \frac{\delta h_i}{\delta \mu}\right]\right> \left<\ad_{\delta h_i / \delta \mu}^*\mu_t, \frac{\delta f}{\delta \mu}\right> \\
    &= \frac{\beta}{2}\sum_{i=1}^n \left<\ad_{\delta h_0 / \delta \mu}^*\mu_t, \frac{\delta h_i}{\delta \mu}\right> \left<\ad_{\delta h_i / \delta \mu}^*\mu_t, \frac{\delta f}{\delta \mu}\right> \\
    &= \frac{\beta \sigma^2}{2}\sum_{i=1}^n \left<\ad_{\delta h_0 / \delta \mu}^*\mu_t, \,\xi_i\right> \left<\ad_{\xi_i}^*\mu_t, \frac{\delta f}{\delta \mu}\right> \\
    &= \frac{\beta \sigma^2}{2}\sum_{i=1}^n  \gamma\left((\ad_{\delta h_0 / \delta \mu_t}^*\mu)^\sharp, \,\xi_i\right) \left<\ad_{\xi_i}^*\mu_t, \frac{\delta f}{\delta \mu}\right> \\
    &= \frac{\beta \sigma^2}{2}\left<\sum_{i=1}^n \gamma\left((\ad_{\delta h_0 / \delta \mu_t}^*\mu_t)^\sharp, \,\xi_i\right) \ad_{\xi_i}^*\mu_t, \, \frac{\delta f}{\delta \mu}\right> \\
    &= \frac{\beta \sigma^2}{2}\left<\ad^*_{(\ad^*_{\delta h_0/\delta \mu} \mu_t)^\sharp} \mu_t, \, \frac{\delta f}{\delta \mu}\right>,
\end{align*}
where we have denoted by $\sharp : \mathfrak{g}^* \rightarrow \mathfrak{g}$ the musical isomorphism associated to the inner product $\gamma$ (that is, $\left<\alpha, \xi\right> = \gamma(\alpha^\sharp, \xi)$ for all $\alpha \in \mathfrak{g}^*$ and $\xi \in \mathfrak{g}$), and in the last line we used the linearity of $\ad^*$ in the first argument, in addition to the identity
\begin{align*}
    \sum_{i=1}^n \gamma\left((\ad_{\delta h_0 / \delta \mu_t}^*\mu_t)^\sharp, \,\xi_i\right) \xi_i = (\ad^*_{\delta h_0/\delta \mu} \mu_t)^\sharp \,.
\end{align*}
The latter follows from the fact that $\{\xi_i\}_{i=1}^n$ is an orthonormal basis of $\mathfrak{g}$ with respect to $\gamma$.
We can then show that \eqref{eq:stoch-diss-lie-poisson} corresponds to the following equation for $\mu_t \in \mathfrak{g}^*$
\begin{align}\label{eq:lie-poisson-langevin}
    \diff \mu_t = \ad_{\delta h_0 / \delta \mu}^*\mu_t \,\diff t + \theta \ad^*_{(\ad^*_{\delta h_0/\delta \mu} \mu_t)^\sharp} \mu_t \, \diff t + \sigma \sum_{i=1}^n \ad^*_{\xi_i} \mu_t \circ \diff W_t^i \,,
\end{align}
where $\theta := \beta \sigma^2 / 2$, which we refer to as the {\em Lie-Poisson-Langevin system}.
Importantly, the dissipative term $\ad^*_{(\ad^*_{\delta h_0/\delta \mu} \mu_t)^\sharp} \mu_t$ in \eqref{eq:lie-poisson-langevin} is precisely the {\em double-bracket dissipation}, first introduced in \cite{bloch1994dissipation, bloch1996euler} to model momentum-preserving dissipative systems, also closely related to the well-known Brockett isospectral flow \cite{brockett1991dynamical} used in linear programming.
Associated to this term is a symmetric {\em double-bracket}
\begin{align} \label{eq:double-symmetric-bracket}
    \{\{f, g\}\}_\gamma := \gamma^{-1}\left(\ad^*_{\delta f/\delta \mu} \mu_t, \,\ad^*_{\delta g/\delta \mu} \mu_t\right) \,,
\end{align}
where $\gamma^{-1}: \mathfrak{g}^* \times \mathfrak{g}^* \rightarrow \mathbb{R}$ is the co-metric associated to $\gamma$. One can then re-express the system \eqref{eq:stoch-diss-lie-poisson} in bracket form as
\begin{align} \label{eq:lie-poisson-langevin-bracket-form}
    \diff f(\mu_t) = \{f, h_0\}_{\mathfrak{g}^*_-}(\mu_t)\,\diff t - \theta \{\{f, h_0\}\}_\gamma (\mu_t)\,\diff t + \sigma \sum_{i=1}^n \{f, \left<\xi_i, \cdot \right>\}_{\mathfrak{g}^*_-}(\mu_t) \circ \diff W_t^i \,,
\end{align}
for any $f \in C^\infty(\mathfrak{g}^*; \mathbb{R})$. We note that in general, there is no canonical choice for the metric $\gamma$ and the basis $\{\xi_i\}_{i=1}^n$ -- these should be viewed as modelling choices. However, on compact semi-simple Lie algebras, the (minus) Killing form $(\xi, \eta) \mapsto -\mathrm{tr}\left(\ad_\xi \circ \ad_\eta\right)$ is non-degenerate and positive definite, making it a suitable candidate for the inner product. We also consider a semidirect product extension of the system \eqref{eq:lie-poisson-langevin} in Appendix \ref{app:semidirect-noise-dissipation}, which follows similarly, albeit with a different Poisson bracket.


Recall that the symplectic leaves of $\mathfrak{g}^*$ are characterised by the connected components $\mathcal{O}_{\mu}'$ of the coadjoint orbits $\mathcal{O}_{\mu} := \{\Ad^*_{g}\mu : g \in G\} \subset \mathfrak{g}^*$ (see \cite[Section 14.3]{marsden2013introduction}), which are equipped with the Kirillov-Kostant-Souriau (KKS) symplectic form $\omega_{\mathrm{KKS}}(\ad^*_\xi \mu, \,\ad^*_\eta \mu) = -\left<\mu, [\xi, \eta]\right>$ for all $\mu \in \mathcal{O}_{\mu}'$ and $\xi, \eta \in \mathfrak{g}$. The Lie-Poisson-Langevin system \eqref{eq:lie-poisson-langevin} thus preserves the coadjoint orbits, and moreover its dynamics restricted to the coadjoint orbits can be described by the Langevin system \eqref{eq:reduced-symplectic-langevin} on the reduced symplectic manifold $(\mathcal{O}_{\mu}', \omega_{\mathrm{KKS}})$. Thus, we obtain the following result, which is a slight generalisation of the result found in \cite[Theorem 3.3]{arnaudon2018noise}.

\begin{corollary} \label{cor:lie-poisson-gibbs-measure}
    For any $\mu \in \mathfrak{g}^*$, the Lie-Poisson-Langevin system \eqref{eq:lie-poisson-langevin} preserves the Gibbs measure
    \begin{align}\label{eq:lie-poisson-gibbs-measure}
    \mathbb P_\infty^\mu = Z_\mu^{-1} e^{-\beta h_0} |\omega^{n'}_{\mathrm{KKS}}|, \quad Z_\mu = \int_{\mathcal{O}_\mu'} e^{-\beta h_0} \mathrm{d}|\omega^{n'}_{\mathrm{KKS}}| \,,
    \end{align}
    defined over a connected component of the coadjoint orbit $\mathcal{O}_\mu' \subseteq \mathcal{O}_\mu$ with dimension $2n'$.
\end{corollary}

\begin{remark}
    In \cite{arnaudon2018noise}, the authors use the {\em selective-decay bracket} instead of the double-bracket to add dissipation to the stochastic Lie-Poisson system and yet recover the same result as Corollary \ref{cor:lie-poisson-gibbs-measure} concerning the preservation of the Gibbs measure. Selective-decay brackets are used to dissipate energy on a level set of a selected Casimir function, while double brackets preserve the symplectic leaves (and therefore {\em all} the Casimirs). Importantly, they {\em do not} necessarily coincide for arbitrary $\mathfrak{g}^*$. However, the fact that we obtain the same results as \cite{arnaudon2018noise} is purely coincidental, since in the special case of compact semi-simple groups, their choice of the Casimir $C(\mu) = \kappa(\mu, \mu)$, where $\kappa$ is the Killing form, leads to a selective-decay bracket that is equivalent to the double-bracket. We strongly argue for the use of double-bracket instead of the selective-decay bracket, as our result holds beyond the compact semi-simple case and moreover can be derived more naturally via symmetry reduction.
\end{remark}

Finally, we can also verify the ergodicity of the process \eqref{eq:lie-poisson-langevin} when restricted to the coadjoint orbits, therefore ensuring that the invariant measure \eqref{eq:lie-poisson-gibbs-measure} is unique and coincides with the measure of the long-time dynamics. We state this precisely as follows.

\begin{proposition}
    The measure $\mathbb{P}_\infty$ given by \eqref{eq:lie-poisson-gibbs-measure} is a unique invariant measure of the process $\{\mu_t\}_{t\in [0, T]}$ solving \eqref{eq:lie-poisson-langevin}, when restricted to a connected component of the coadjoint orbit $\mathcal{O}' \subseteq \mathcal{O}$. Moreover, let $\lambda := |\omega^{n'}_{\mathrm{KKS}}|$ be the Liouville measure on $\mathcal{O}'$, $\mathbb{P}_t(x, A) := \mathbb{P}(\mu_t \in A)$ the transition probability measure of the process $\{\mu_t\}_{t\in [0, T]}$ with $\mu_0 = x$, and $p_\infty$ the density of $\mathbb{P}_\infty$ with respect to $\lambda$. Then $\mathbb{P}_t(x, \cdot)$ admits a density $p_t(x, \cdot)$ with respect to $\lambda$, satisfying the following ergodic property
    \begin{align}\label{eq:ergodic-property-lie-poisson}
        \lim_{t \rightarrow \infty} \int_{\mathcal{O}'} |p_t(x, y) - p_\infty(y)| \,\mathrm{d}\lambda(y), \quad \forall x \in \mathcal{O}' \,,
    \end{align}
    i.e., we have $L^1_\lambda$ convergence $p_t(x, \cdot) \rightarrow p_\infty(\cdot)$ as $t \rightarrow \infty$ for any $x \in \mathcal{O}'$.
\end{proposition}
\begin{proof}
    We first note that for any $\mu \in \mathcal{O}'$, the tangent space of $\mathcal{O}'$ is characterised by
    $T_\mu \mathcal{O}' = \{\ad^*_\xi \mu | \xi \in \mathfrak{g}\}$ \cite[Section 14.2]{marsden2013introduction}. Hence, we have $T_\mu \mathcal{O}' = \mathrm{Span}\{\ad_{\xi_1} \mu, \ldots, \ad_{\xi_n} \mu\}$ since we assumed that $\{\xi_i\}_{i=1}^n$ span $\mathfrak{g}$ (i.e., the noise vector fields in \eqref{eq:lie-poisson-langevin} span $T_\mu \mathcal{O}'$). Then, by \cite[Proposition 6.1]{ichihara1974classification}, the process \eqref{eq:lie-poisson-langevin} is elliptic (i.e. the corresponding generator $\mathcal{A}$ is an elliptic differential operator) and admits at most one invariant measure, which proves the uniqueness. The ergodic property \eqref{eq:ergodic-property-lie-poisson} then follows directly from the result \cite[Proposition 5.1]{ichihara1974classification}, owing to the ellipticity of the process and the fact that $p_\infty$ has full support on $\mathcal{O}'$.
\end{proof}

To summarise, we have presented an explicit derivation of the double-bracket dissipation from the point of view of statistical mechanics. In particular, we have shown that it arises naturally as the symmetry-reduced counterpart of the dissipative term in \eqref{eq:symplectic-langevin-eq}, introduced to preserve the Gibbs measure \eqref{eq:gibbs-meas} on the phase space $P=T^*G$. Moreover, the resulting symmetry-reduced process \eqref{eq:lie-poisson-langevin} preserves the Gibbs measure on the connected components of the coadjoint orbits and is furthermore ergodic on this space.

\section{Examples}\label{sec:examples}
In this section, we provide concrete examples of stochastic-dissipative systems deduced from symmetry reduction that are of physical interest.


\subsection{Stochastic dissipative rigid body dynamics}\label{sec:stoch-diss-rigid-body}

The reduced Hamiltonian of the free rigid body system on $SO(3)$ is \cite[Chapter 15]{marsden2013introduction}

\begin{equation}
h_0(\boldsymbol{\Pi}) = \frac12 \boldsymbol{\Pi} \cdot \mathbb{I}^{-1} \boldsymbol{\Pi} \,,
\end{equation}
for angular momentum $\boldsymbol{\Pi} \in \mathfrak{so}^*(3) \cong \mathbb{R}^3$ and moment of inertia $\mathbb{I}: \mathfrak{so}(3) \rightarrow \mathfrak{so}^*(3)$, which is a symmetric tensor. The inverse operator $\mathbb{I}^{-1} \in \mathfrak{so}(3) \otimes \mathfrak{so}(3) \cong \mathbb{R}^{3 \times 3}$ naturally defines a metric on $\mathfrak{so}^*(3)$. We can take $\{\boldsymbol{\xi}_1, \boldsymbol{\xi}_2, \boldsymbol{\xi}_3\}$ to be the eigenvectors of $\mathbb{I} \in \mathbb{R}^{3 \times 3}$. Taking the noise Hamiltonians $h_i(\vec{\Pi}) = \sigma \vec{\xi}_i \cdot \vec{\Pi}$, the correspoding Lie-Poisson-Langevin system \eqref{eq:lie-poisson-langevin} becomes
\begin{align}\label{eq:stoch-diss-rb}
    \diff \boldsymbol{\Pi}_t + \boldsymbol{\Pi}_t \times \mathbb{I}^{-1} \boldsymbol{\Pi}_t \,\diff t + 
    \theta
    \boldsymbol{\Pi}_t \times \mathbb{I}^{-1}(\boldsymbol{\Pi}_t \times \mathbb{I}^{-1} \boldsymbol{\Pi}_t) \,\diff t + \sigma \sum_{i=1}^3  \boldsymbol{\Pi}_t\times \boldsymbol{\xi}_i \circ \diff W_t^i = 0 \,,
\end{align}
where we have used that the coadjoint representation of the Lie alebra $\mathfrak{so}(3)$ is given by $\ad^*_{\vec{\xi}} \vec{\Pi} = \vec{\xi} \times \vec{\Pi}$.
This gives the momentum-preserving stochastic dissipative rigid body system, which is a stochastic extension of the dissipative rigid body system in \cite{bloch1996euler} and has been studied in detail in \cite{arnaudon2018noise}. The system \eqref{eq:stoch-diss-rb} is also equivalent to the stochastic Landau-Lifshitz-Ginsburg equation \cite{banas2013stochastic} modelling the stochastic non-equilibrium dynamics of ferromagnetism. Their ensemble behaviour has been investigated for example in \cite{banas2013stochastic, arnaudon2019networks}, in particular they are known to exhibit phase transition with varying temperature.

The coadjoint orbits of $\mathfrak{so}^*(3)$ are the $2$-spheres $\|\boldsymbol{\Pi}\|^2 = const$.
The KKS symplectic form on $\mathcal{O}_{\boldsymbol{\mu}} \subset \mathfrak{so}^*(3)$ for some $\boldsymbol{\mu} \in \mathfrak{so}^*(3) \cong \mathbb{R}^3$ is
\begin{align}
    \omega_{\mathcal{O}_{\boldsymbol{\mu}}}(\boldsymbol{\mu} \times \boldsymbol{\xi}, \boldsymbol{\mu} \times \boldsymbol{\eta})  = - \boldsymbol{\mu} \cdot \boldsymbol{\xi} \times \boldsymbol{\eta}, \qquad \vec{\xi}, \vec{\eta} \in \mathbb{R}^3,
\end{align}
which is equivalent to the natural volume on the 2-sphere $S^2_\mu = \{\boldsymbol{\Pi} \in \mathbb{R}^3 : \|\boldsymbol{\Pi}\|^2 = \|\boldsymbol{\mu}\|^2\}$ (that is, if $\iota_\mu : S^2_\mu \xhookrightarrow{} \mathbb{R}^3$ is the natural inclusion, then $\omega_{\mathcal{O}_{\boldsymbol{\mu}}} = \iota_\mu^* \diff^3 x$, where $\diff^3 x$ is the Euclidean volume form on $\mathbb{R}^3$). Thus, from Corollary \ref{cor:lie-poisson-gibbs-measure}, the invariant measure on $S_\mu^2$ is given by
\begin{align}
    \mathbb{P}^\mu_\infty \,\,\propto\,\, \iota_{\mu}^* \left(e^{-\beta h}\diff^3 x\right),
\end{align}
which is the restriction (marginalisation) of the Euclidean Gaussian measure $\mathcal{N}(0, \beta^{-1}\mathbb{I})$ on $\mathbb{R}^3$ onto $S_\mu^2$. Its density on $S_\mu^2$ has two modes corresponding to the two minimal energy configurations of rotations around the shortest principle axis (clockwise and anticlockwise).

\subsection{Charged Brownian particle in a magnetic field}\label{sec:brownian-particle}
We now give an example of a stochastic dissipative system obtained by symplectic reduction, describing the stochastic motion of charged Brownian particles in a magnetic field. We take the configuration space to be the Kaluza-Klein configuration space $Q = \mathbb{R}^3 \times U(1)$, equipped with a connection form $\mathcal{A} : TQ \rightarrow \mathfrak{u}(1)$.

One starts with the Kaluza-Klein Lagrangian with potential energy (adapted from \cite{gay2013geometric})
\begin{align}
    L(\vec{q}, \dot{\vec{q}}, \theta, \dot{\theta}) = \frac{m}{2}\|\dot{q}\|^2 + \frac12 |\mathcal{A}(\vec{q}, \dot{\vec{q}}, \theta, \dot{\theta})|^2 - V(\vec{q}),
\end{align}
where $(\vec{q}, \dot{\vec{q}}, \theta, \dot{\theta}) \in TQ$ and $V : \mathbb{R}^3 \rightarrow \mathbb{R}$ is an arbitrary potential energy.
Without loss of generality, we can write (since the bundle $Q$ is trivial and the group is $U(1)$ Abelian)
\begin{align}
    \mathcal{A}(\vec{q}, \dot{\vec{q}}, \theta, \dot{\theta}) = A_{\vec{q}} \dot{\vec{q}} + \dot{\theta},
\end{align}
for some $\mathfrak{u}(1)$-valued one-form $A : T\mathbb{R}^3 \rightarrow \mathfrak{u}(1)$, which is simply a one-form on $\mathbb{R}^3$ upon identifying $\mathfrak{u}(1) \cong \mathbb{R}$.
Taking the Legendre transform, this yields a Hamiltonian on $T^*Q$ of the form
\begin{align}\label{eq:particle-in-magnet-hamiltonian}
    H_0(\vec{q}, \vec{p}, \theta, \mu) = \frac{1}{2m}\|\vec{p} - \mu A_{\vec{q}}\|^2 + \frac12 \mu^2 + V(\vec{q}),
\end{align}
where $\vec{p} = m \dot{\vec{q}} + \mu A_{\vec{q}}$ and $\mu = A_{\vec{q}} \dot{\vec{q}} + \dot{\theta}$ are the conjugate momenta on the $\mathbb{R}^3$ and $U(1)$ component of $Q$ respectively. We equip $T^*Q$ with the canonical symplectic form $\omega = \sum_{i=1}^3 \dd q^i \wedge \dd p_i + \dd \theta \wedge \dd \mu$. Since the Hamiltonian \eqref{eq:particle-in-magnet-hamiltonian} does not depend on the $\theta$-variable (i.e., it is a cyclic variable), it is trivially invariant under the lifted action of $U(1)$ on $T^*Q$. Thus, the momentum map
\begin{align}
    J(\vec{q}, \vec{p}, \theta, \mu) = \mu,
\end{align}
is preserved under the flow of $X_{H_0}$ and $J^{-1}(\mu) = \{(\vec{q}', \vec{p}', \theta', \mu') \in T^*Q | \mu' = \mu\}$ is an invariant manifold. We consider the shifting map 
$\pi_\mu: J^{-1}(\mu) \rightarrow J^{-1}(0)$, defined by 
\begin{align} \label{eq:shifting-map}
(\pi_\mu)_{(\vec{q}, \theta)}(\vec{p}, \mu) := (\vec{p}, \mu) - \mu \mathcal{A}_{(\vec{q}, \theta)} = (\vec{p} - \mu A_{\vec{q}}, 0),
\end{align}
and noting that $J^{-1}(0) \cong T^*(Q / U(1)) \cong T^*\mathbb{R}^3$ holds\footnote{This identification follows from the fact that any $p_q \in J^{-1}(0)$ satisfies $\left<p_q, \xi_Q(q)\right> = 0$ for all $\xi \in \mathfrak{u}(1)$ by definition, so $p$ is a horizontal one-form, which can be identified with a one-form on $Q/U(1)$.}, we have the relation
\begin{align}
    T^*\mathbb{R}^3 \stackrel{\pi_\mu}{\leftarrow} J^{-1}(\mu) \xhookrightarrow{\iota_\mu} T^*Q.
\end{align}
This induces a symplectic form $\omega_\mu$ on $T^*\mathbb{R}^3$ by $\pi_\mu^* \omega_\mu = \iota^* \omega$, which is given explicitly by
\begin{align}
    \omega_\mu = \sum_{i=1}^3 \dd q^i \wedge \dd p_i - \mu \,\dd A.
\end{align}
The space $(T^*\mathbb{R}^3, \omega_\mu)$ is symplectomorphic to $J^{-1}(\mu) / U(1)$ and therefore may be understood as the reduced space $P_\mu$. The corresponding reduced Hamiltonian on $T^*\mathbb{R}^3$ then reads
\begin{align}
    h_0(\vec{q}, \tilde{\vec{p}}) = \frac{1}{2m} \|\tilde{\vec{p}}\|^2 + V(\vec{q}),
\end{align}
for $(\vec{q}, \tilde{\vec{p}}) \in T^*\mathbb{R}^3$ and taking the noise Hamiltonians $h_i(\vec{q}, \tilde{\vec{p}}) = -\sigma q^i$ for $i=1,2,3$, \eqref{eq:reduced-symplectic-langevin} becomes
\begin{align}
    \diff \vec{q}_t = (\tilde{\vec{p}}_t / m) \,\diff t, \qquad \diff \tilde{\vec{p}}_t = \left(-\nabla_{\vec{q}}V(\vec{q}_t) + \frac{\mu}{m}\tilde{\vec{p}}_t \times \vec{B}(\vec{q}_t) - \frac{\theta}{m} \tilde{\vec{p}}_t \right)\diff t + \sigma \,\diff \vec{W}_t,
\end{align}
where $\theta = {\beta \sigma^2}/{2}$ and $(\vec{B})_i = (\epsilon_{ijk} / 2) \left(\partial A_j / \partial q^k - \partial A_k / \partial q^j \right)$ is the so-called magnetic term. This is precisely the equation for a charged Brownian particle in a magnetic field \cite{czopnik2001brownian}, which has seen recent applications in accelerating convergence in Hamiltonian Monte-Carlo (HMC) methods for sampling distributions \cite{tripuraneni2017magnetic, mongwe2021magnetic, mongwe2021quantum} (since HMC can be understood as a discretisation of the underdamped Langevin system \cite{rousset2010free}).
Noting that the symplectic volume form on $(T^*\mathbb{R}^3, \omega_\mu)$ is just $\omega_\mu \wedge \omega_\mu \wedge \omega_\mu = \diff^3 q \, \diff^3 \tilde{p}$, the invariant measure \eqref{eq:gibbs-meas-reduced} on the reduced space is
\begin{align}
    \mathbb{P}_\infty^{\mu} \,\propto\, e^{-\beta h_0(\vec{q}, \tilde{\vec{p}})} \diff^3 q \, \diff^3 \tilde{p} = e^{{- \frac{\beta}{2} \|\tilde{\vec{p}}}\|^2 - \beta V(\vec{q})} \diff^3 q \, \diff^3 \tilde{p} \,,
\end{align}
which is equivalent to the invariant measure of the standard underdamped Langevin system. We note that this reduction method using the shifting map \eqref{eq:shifting-map} can be further generalised to arbitrary principal bundles arising in Yang-Mills theory \cite{sternberg1977minimal, montgomery1984canonical, marsden1992lectures, ortega2005cotangent}, allowing us to obtain stochastic-dissipative analogues of such systems within our framework.

\subsection{Dissipative SALT-Euler model}\label{sec:salt-with-dissipation}
Here, we formally consider an infinite dimensional example arising in ideal fluid dynamics, where the configuration manifold is the group of volume-preserving diffeomorphisms $\mathfrak{Diff}_{\mathrm{vol}}(M)$ over a compact, simply-connected, orientable volume manifold $(M, \mathrm{vol})$. The Lie algebra corresponding to this space is $\mathfrak{X}_{\mathrm{vol}}(M)$, the space of divergence-free vector fields, equipped with the vector field commutator $[u, v] = uv - vu$, for all $u,v \in \mathfrak{X}_{\mathrm{vol}}(M)$. Its dual is the space $\mathfrak{X}^*_{\mathrm{vol}}(M) = \Omega^1(M) / \dd \mathcal{F}(M)$ of all one-forms modulo exact one-forms \cite{marsden1983coadjoint}. On simply-connected base manifold $M$, this can further be identified with the space of closed two-forms $\mathfrak{X}^*_{\mathrm{vol}}(M) \cong \dd \Omega^1(M) \subset \Omega^2(M)$, with the duality pairing given by
\begin{align}\label{eq:X-vol-duality-pairing}
    \left<\omega, u\right>_{\mathfrak{X}^*_{\mathrm{vol}} \times \mathfrak{X}_{\mathrm{vol}}} := \int_M \left<\alpha, u\right> \mathrm{vol},
\end{align}
for any $\omega \in \dd \Omega^1(M)$, $\alpha \in \Omega^1(M)$ is any one-form such that $\dd \alpha = \omega$, and any $u \in \mathfrak{X}_{\mathrm{vol}}(M)$.
The adjoint and coadjoint representations of $\mathfrak{X}_{\mathrm{vol}}(M)$ are given by $\ad_uv = -[u, v]$ and $\ad^*_u \omega = \mathcal{L}_u \omega$ respectively, for all $u, v \in \mathfrak{X}_{\mathrm{vol}}(M)$ and $\omega \in \mathfrak{X}^*_{\mathrm{vol}}(M)$.

Given a Riemannian metric tensor $\gamma : T M \times T M \rightarrow \mathbb{R}$, we consider its lift to a
metric pairing $\tilde{\gamma}(\cdot, \cdot) : \bigwedge^2T^* M \times \bigwedge^2T^* M \rightarrow \mathbb{R}$ (see Appendix \ref{app:proof-of-orthonormal-condition} for details). Now define the reduced Hamiltonians
\begin{align} \label{eq:euler-hamiltonians}
    h_0(\omega) = \frac12 \int_M \tilde{\gamma}(\omega, \psi) \mathrm{vol}, \qquad h_i(\omega) = \sigma \int_M \tilde{\gamma}(\omega, \psi_i) \mathrm{vol}, \quad i = 1, 2, \ldots,
\end{align}
on $\mathfrak{X}^*_{\mathrm{vol}}(M)$, where $\psi := \Delta^{-1} \omega \in \Omega^2(M)$ is the streamfunction corresponding to $\omega$ (interpreted as a two-form) and $\Delta : \Omega^2(M) \rightarrow \Omega^2(M)$ denotes the Laplace-de Rham operator on two-forms.
Given an exact two-form $\psi$, we can associate with it a divergence-free vector field $u$ via a linear operator $\Upsilon \psi = u$\footnote{More precisely, this is defined by $\Upsilon \psi := (\vec{\delta} \psi)^\sharp$, where $\vec{\delta}$ is the codifferential operator \cite{marsden1983coadjoint} and $\sharp : \Omega^1(M) \rightarrow \mathfrak{X}(M)$ is the musical isomorphism defined respect to the Riemannian metric $\gamma$. In two dimensions, the operator $\psi \mapsto \Upsilon\psi$ is the skew-gradient when two-forms are identified as zero-forms.}. One can show that this is precisely the functional derivative $\frac{\delta h_0}{\delta \omega} = \Upsilon \psi \in \mathfrak{X}_{\text{vol}}(M)$ defined with respect to the pairing \eqref{eq:X-vol-duality-pairing}.

For the streamfunctions $\{\psi_i\}_{i \in \mathbb{Z}_+}$ characterising the noise Hamiltonians, we consider the following particular choice. Fix $n \in \mathbb{Z}_+$ and consider the inner product $\tilde{\gamma}^n(\alpha, \beta) := \int_{M} \tilde{\gamma}(\alpha, \Delta^{n} \beta) \text{vol}$ for any $\alpha, \beta \in \Omega^2(M)$. We choose $\{\psi_i\}_{i \in \mathbb{Z}_+}$ to be a set of exact two-forms that are orthonormal with respect to this inner-product. Then, we have the following result.
\begin{proposition}\label{prop:vector-field-orthonormality}
    Let $\{\psi_i\}_{i \in \mathbb{Z}_+}$ be a set of exact two-forms that is orthonormal with respect to the inner-product $\tilde{\gamma}^n(\alpha, \beta) := \int_{M} \tilde{\gamma}(\alpha, \Delta^{n} \beta) \text{vol}$. The vector fields $\{\Upsilon \psi_i\}_{i \in \mathbb{Z}_+}$ in the Lie algebra $\mathfrak{X}_{\mathrm{vol}}(M)$ are then orthonormal with respect to the inner product $\gamma^{n-1} : \mathfrak{X}_{\mathrm{vol}}(M) \times \mathfrak{X}_{\mathrm{vol}}(M) \rightarrow \mathbb{R}$, defined by
    \begin{align}\label{eq:X-vol-metric}
        \gamma^{n-1}(u, v) := \int_M \gamma(u, (\Delta^\sharp)^{n-1} v) \mathrm{vol},
    \end{align}
    where $\Delta^\sharp : \mathfrak{X}_{\mathrm{vol}}(M) \rightarrow \mathfrak{X}_{\mathrm{vol}}(M)$ is defined by $\left<\Delta \omega, u\right>_{T^*M \times TM} = \left<\omega, \Delta^\sharp u\right>_{T^*M \times TM}$.
\end{proposition}
\begin{proof}
    See Appendix \ref{app:proof-of-orthonormal-condition}.
\end{proof}

This result allows us to apply the argument in Section \ref{sec:double-bracket-dissipation} to formally deduce the Lie-Poisson-Langevin system \eqref{eq:lie-poisson-langevin} on the group $\mathfrak{Diff}_{\mathrm{vol}}(M)$, which gives us
\begin{align} \label{eq:salt-euler-with-dissipation}
    \diff \omega + \mathcal{L}_u \omega \,\diff t + \theta \mathcal{L}_{(\mathcal{L}_u \omega)^\sharp} \omega \,\diff t + \sigma \sum_{i=1}^\infty \mathcal{L}_{\xi_i} \omega \circ \diff W_t^i = 0,
\end{align}
where $\omega \in \dd \Omega^1(M)$ is the vorticity two-form, $\xi_i = \frac{\delta h_i}{\delta \omega} = \Upsilon \psi_i$ is the velocity vector field corresponding to the streamfunction $\psi_i$, and $u := \frac{\delta h_0}{\delta \omega} = \Upsilon \psi \in \mathfrak{X}_{\mathrm{vol}}(M)$ is the velocity vector field corresponding to the streamfunction $\psi$. Here, the isomorphism $\sharp : \mathfrak{X}^*_{\mathrm{vol}}(M) \rightarrow \mathfrak{X}_{\mathrm{vol}}(M)$ is defined with respect to the metric \eqref{eq:X-vol-metric} for some $k \in \mathbb{Z}_+$, that is, $\gamma^{n-1}(\omega^\sharp, v) = \left<\omega, v\right>_{\mathfrak{X}^*_{\mathrm{vol}} \times \mathfrak{X}_{\mathrm{vol}}}$.
We note that the system \eqref{eq:salt-euler-with-dissipation} can be seen as a dissipative extension of the SALT-Euler system introduced in \cite{holm2015variational}.
In the special case $M = \mathbb{R}^2$ (or $\mathbb{R}^3$)\footnote{On $\mathbb{R}^n$, one needs to impose additional decay conditions on the diffeomorphisms to define a Lie group \cite{michor2013zoo}.} and taking $\sigma=0$, our model \eqref{eq:salt-euler-with-dissipation} also recovers the Vallis-Carnevale-Young system introduced in \cite[Section 4.2]{vallis1989extremal}, originally considered as an extension to ideal fluid models for computing the stable solutions to 2D fluids, as well as parameterising the {\em energy cascade} effect in 3D fluids (see also the work \cite{shepherd1990general}, which considers an extension to the Vallis-Carnevale-Young model on general Hamiltonian systems, not dissimilar to our system \eqref{eq:lie-poisson-langevin} with $\sigma = 0$).

Now considering the SDE
\begin{align}
    \diff Z_t = \left(u(t, Z_t) + \theta \left(\mathcal{L}_u\omega\right)^\sharp(t, Z_t)\right) \diff t + \sigma \sum_{i=1}^\infty \xi_i(Z_t) \circ \diff W_t^i,
\end{align}
and denoting by $\{\phi_t\}_{t \in [0, T]}$ its stochastic flow of diffeomorphism (assuming it exists), we can deduce from It\^o's pushforward formula \cite[Theorem 4.9.2]{kunita1984stochastic} that a solution to \eqref{eq:salt-euler-with-dissipation} is given by $\omega_t = (\phi_t)_* \omega_0$.
Indeed, this lies on the coadjoint orbits of $\mathfrak{X}^*_{\mathrm{vol}}(M)$, given by 
\begin{align}\label{eq:diffeo-coadjoint-orbits}
\mathcal{O}_{\omega_0} = \{\eta_* \omega_0: \eta \in \mathfrak{Diff}_{\mathrm{vol}}(M)\},
\end{align}
where $\eta_*\omega_0$ denotes the pushforward of the two-form $\omega_0$ with respect to the diffeomorphism $\eta$ \cite{marsden1983coadjoint}. This is equipped with the KKS symplectic form
\begin{align}\label{eq:diffeo-kks-form}
\Omega_\omega(\mathcal{L}_{u_1} \omega, \mathcal{L}_{u_2} \omega) = \int_M \omega(u_1, u_2) \mathrm{vol} \,,
\end{align}
for any $\omega \in \mathcal{O}_{\omega_0}$. However unlike the previous examples, it is not clear what the invariant measures corresponding to \eqref{eq:salt-euler-with-dissipation} are or whether one exists in this setting, as our result (Corollary \ref{cor:lie-poisson-gibbs-measure}) is restricted to finite dimensions. In the following example, we consider a finite-dimensional approximation to the system \eqref{eq:salt-euler-with-dissipation} using point vortices, in order to gain further insights.

\subsection{Stochastic dissipative point vortices on the $2$-sphere}\label{sec:stoch-diss-point-vortices}
Let us now restrict to the setting $M = S^2$. On a two-dimensional simply-connected manifold, every two-form is closed and exact. Hence, we can set $\mathfrak{X}^*_{\mathrm{vol}}(S^2) \cong  \Omega^2(S^2)$. Further, we have the isomorphism $q: \Omega^2(S^2) \stackrel{\sim}{\rightarrow} \Omega^0(S^2)$ given by $q : \omega \mapsto \omega / \mathrm{vol}$, so we can identify $\mathfrak{X}^*_{\mathrm{vol}}(S^2)$ with the space $\Omega^0(S^2)$ of smooth functions over $S^2$. Now for any $f \in \Omega^0(S^2)$, we can show that
\begin{align}
    \Delta (f \,\mathrm{vol}) = \mathrm{div}(\mathrm{grad} f) \, \mathrm{vol} \,,
\end{align}
where $\mathrm{grad} f := (\dd f)^\sharp$, for $\sharp : T^* S^2 \rightarrow T S^2$ defined with respect to the Riemannian metric $\gamma$ on $S^2$ induced from the embedding $S^2 \xhookrightarrow{} \mathbb{R}^3$. The operator $\mathrm{div}(\mathrm{grad}(\cdot))$ is the {\em Laplace-Beltrami operator}, which we will also denote by $\Delta$. We can also show that $\Upsilon (f \,\mathrm{vol}) = X_{f}$ for any $f \in \Omega^0(S^2)$, where $X_f$ is the Hamiltonian vector field on $(S^2, \mathrm{vol})$, viewed as a symplectic manifold \cite{marsden1983coadjoint} (for convenience, we will sometimes use the notation $\nabla^\perp f := X_f$ for any $f \in \Omega^0(S^2)$). Hence, under the identification $\Omega^2(S^2) \cong \Omega^0(S^2)$, the Laplace-de Rham operator on two-forms becomes the Laplace-Beltrami operator and the velocity vector field $\Upsilon \psi$ becomes a Hamiltonian vector field $X_\psi$. Now, choosing $\mathrm{vol} = \mathrm{vol}_{S^2}$, the area form on $S^2$ induced by the embedding $S^2 \xhookrightarrow{} \mathbb{R}^3$, we have $\tilde{\gamma}(\mathrm{vol}_{S^2}, \mathrm{vol}_{S^2}) = 1$\footnote{This follows from the fact that $\mathrm{vol}_{S^2} = \sqrt{|\gamma|} \diff x \wedge \diff y$ and $\tilde{\gamma} = |\gamma^{-1}| \left(\frac{\partial}{\partial x} \wedge \frac{\partial}{\partial y}\right) \otimes \left(\frac{\partial}{\partial x} \wedge \frac{\partial}{\partial y}\right)$ in local coordinates.}. Hence, under the identification $\mathfrak{X}^*_{\mathrm{vol}}(S^2) = \Omega^0(S^2)$, the reduced Hamiltonians \eqref{eq:euler-hamiltonians} become 
\begin{align} \label{eq:euler-hamiltonians-2d}
    h_0(\omega) = \frac12 \int_{S^2} \omega \psi \,\mathrm{vol}_{S^2} \,, \qquad h_i(\omega) = \sigma \int_{S^2} \omega \psi_i \,\mathrm{vol}_{S^2} \,, \quad i = 1, 2, \ldots \,,
\end{align}
for any $\omega, \psi, \psi_i \in \Omega^0(S^2)$ and $\sigma > 0$. We can now take limits and consider distribution-valued vorticities $\omega$, owing to the density of $\Omega^0(S^2)$ in the space of distributions. Due to the ellipticity of the Laplacian, $\psi = \Delta^{-1}\omega$ is smooth, so the Hamiltonians \eqref{eq:euler-hamiltonians-2d} will be well-defined in this limit.
In particular,  consider a singular vorticity ansatz
\begin{align} \label{eq:point-vortex-ansatz}
    \omega(t, x) = \sum_{i=1}^N \Gamma_i \delta(x; x_i(t)) \,,
\end{align}
where $\delta(x; x')$ denotes the Dirac delta function on the $2$-sphere, i.e., $\int_{S^2} \delta(x;x') f(x') \mathrm{vol}_{S^2} = f(x)$ for any function $f$ over $S^2$. Solutions of this form are referred to as the {\em point vortex solutions} of \eqref{eq:salt-euler-with-dissipation}, which may be viewed as a discretisation to the continuous solutions via the vortex method \cite{mimeau2021review}.

As discussed in \cite{marsden1983coadjoint}, the geometry of point vortices can be deduced from the geometry of the coadjoint orbits on $\mathfrak{Diff}_{\mathrm{vol}}(S^2)$. In particular, substituting the point vortex ansatz \eqref{eq:point-vortex-ansatz} into \eqref{eq:diffeo-coadjoint-orbits}, the coadjoint orbits of the singular system can be identified with the space
\begin{align}\label{eq:singular-coadjoint-orbit}
\mathcal{O}_{\omega_0} = \{\{x_i(t)\}_{i=1}^N : \exists \eta \in \mathfrak{Diff}_{\mathrm{vol}}(S^2) \text{ s.t. } x_i(t) = \eta(x_i(0)) \text{ for all } i=1, \ldots, N\} \,.
\end{align}
Further, since $\mathfrak{Diff}_{\mathrm{vol}}(S^2)$ is $N$-transitive on $S^2$ for any $N \in \mathbb{Z}_+$ \cite[Theorem 8]{michor1994n}, we have $\mathcal{O}_\omega \equiv (S^2)^{\times N} \backslash \mathcal{C}$, where $\mathcal{C} = \bigcup_{i \neq j} \{(x_1, \ldots, x_N) \in (S^2)^{\times N}: x_i = x_j\}$ is the collision set. The corresponding KKS symplectic form \eqref{eq:diffeo-kks-form} on this singular coadjoint orbit takes the form 
\begin{align}
\Omega_\omega(\mathcal{L}_{u_1} \omega, \mathcal{L}_{u_2} \omega) = \sum_{i=1}^N \Gamma_i \,{\mathrm{vol}}_{S^2}^i(u_1(x_i), u_2(x_i)) \,,
\end{align}
for any $u_1, u_2 \in \mathfrak{X}_{\mathrm{vol}}(S^2)$.
Hence, by symplectic reduction theory, one can deduce that the Hamiltonian dynamics of point vortices occur on the finite-dimensional symplectic manifold $((S^2)^{\times N} \backslash \mathcal{C}, \,\sum_{i=1}^N \Gamma_i \,{\mathrm{vol}}_{S^2}^i)$, which is a singular coadjoint orbit of $\mathfrak{Diff}_{\mathrm{vol}}(S^2)$.
Further, substituting the ansatz \eqref{eq:point-vortex-ansatz} into \eqref{eq:euler-hamiltonians-2d}, the reduced Hamiltonians on this space can be written as
\begin{align} \label{eq:point-vortex-hamiltonian}
h_0(\{x_i(t)\}_{i=1}^N) = -\sum_{i,j=1}^N \frac{\Gamma_i \Gamma_j}{4\pi} \log \|x_i - x_j\|\,, \qquad h_i(\{x_i(t)\}_{i=1}^N) = \sigma \sum_{i=1}^N \Gamma_i \psi_i(x_i) \,,
\end{align}
where $\| \cdot \|$
denotes the chordal distance on $S^2$ (see derivation in Appendix \ref{app:point-vortex-derivation}).

Now let $\{\{(\lambda_{\ell, m}, Y^m_\ell)\}_{|m|=0}^\ell\}_{\ell = 0}^\infty$ be the eigenvalue-function pair of the Laplace-Beltrami operator $\Delta$ on $S^2$, the eigenfunctions being precisely the spherical harmonics. We consider the functions $\psi_{\ell, m}(x) := \lambda_{\ell, m}^{-1/2} Y^m_\ell(x) \in \Omega^0(S^2)$ as the noise potentials, which are orthonormal with respect to the inner product $\tilde{\gamma}_1(f, g) := \int_{S^2} f \,\Delta g \,\text{vol}_{S^2}$ and therefore by Proposition \ref{prop:vector-field-orthonormality}, the vector fields $\nabla^\perp \psi_{\ell, m}$ are orthonormal with respect to the inner product $\gamma_0(u, v) := \int_{S^2} \gamma(u, v)\text{vol}_{S^2}$.
With this choice of noise vector field, we can show that the reduced dynamics \eqref{eq:reduced-symplectic-langevin} on the coadjoint orbit satisfies the following stochastic-dissipative point vortex system
\begin{align}\label{eq:stoch-diss-point-vortex}
\begin{split}
\diff \boldsymbol{x}_i &= \frac{1}{4\pi R} \sum_{\substack{j = 1 \\ j \neq i}}^N \frac{\Gamma_j \,\boldsymbol{x}_j  \times \boldsymbol{x}_i}{R^2 - \boldsymbol{x}_i  \cdot \boldsymbol{x}_j}\diff t + \sigma \sum_{\ell = 1}^\infty \sum_{|m|=0}^\ell \lambda_{\ell, m}^{-1/2} \,\nabla^\perp Y_\ell^m(\boldsymbol{x}_i) \circ \diff W_t^{\ell, m} \\
&\quad - \theta \sum_{\substack{k=1 \\ k \neq i}}^N \Bigg[\Bigg(\sum_{\substack{j = 1 \\ j \neq k}}^N \frac{\Gamma_j \,\boldsymbol{x}_j  \times \boldsymbol{x}_k}{4\pi R(R^2 - \boldsymbol{x}_j  \cdot \boldsymbol{x}_k)}\Bigg) \times \Bigg(\frac{\Gamma_k \vec{x}_i}{4\pi R(R^2 - \boldsymbol{x}_i  \cdot \boldsymbol{x}_k)} \Bigg) \\
&\qquad \qquad + \Bigg(\sum_{\substack{j = 1 \\ j \neq k}}^N \frac{\Gamma_j \,\vec{x}_i \cdot \boldsymbol{x}_j  \times \boldsymbol{x}_k}{4\pi R(R^2 - \boldsymbol{x}_j  \cdot \boldsymbol{x}_k)}\Bigg) \Bigg(\frac{\Gamma_k \vec{x}_k \times \vec{x}_i}{4\pi R(R^2 - \boldsymbol{x}_i  \cdot \boldsymbol{x}_k)^2}\Bigg)\Bigg] \,\diff t \,,
\end{split}
\end{align}
where $\boldsymbol{x}_i \in \mathbb{R}^3$, with $\|\boldsymbol{x}_i\| = R$ (the radius of the sphere) for all $i=1, \ldots, N$, is the extrinsic representation of $x_i \in S^2$ under the embedding $S^2 \xhookrightarrow{} \mathbb{R}^3$ (see the full derivation in Appendix \ref{app:point-vortex-derivation}). One can check that this system can also be deduced by directly substituting the ansatz \eqref{eq:point-vortex-ansatz} into \eqref{eq:salt-euler-with-dissipation}. To the best of our knowledge, the dissipation in system \eqref{eq:stoch-diss-point-vortex} has not been considered in the literature before. However, we believe that it is quite natural, arising as the point vortex approximation to the dissipative SALT system \eqref{eq:salt-euler-with-dissipation}, and the point vortex approximation to the Vallis-Carnevale-Young-Shepherd system \cite{vallis1989extremal, shepherd1990general} in the case $\sigma=0$ (no noise). We display in Figure \ref{fig:point_vortex_with_dissipation} a simulation of \eqref{eq:stoch-diss-point-vortex} in the zero-noise limit, with $N=6$ and taking unit vortex strengths $\Gamma_i = 1$ for all $i = 1, \ldots, 6$. In particular, Figure \ref{fig:hamiltonian-evolution} shows the evolution of the point vortex Hamiltonian $h_0$ (in \eqref{eq:point-vortex-hamiltonian}) under this evolution, verifying that the system is indeed dissipative, which is not immediately obvious by inspection of the new term in \eqref{eq:stoch-diss-point-vortex}.

\begin{figure}[t]
    \centering
    \begin{subfigure}[t]{.4\textwidth}
        \centering  \includegraphics[width=.9\linewidth]{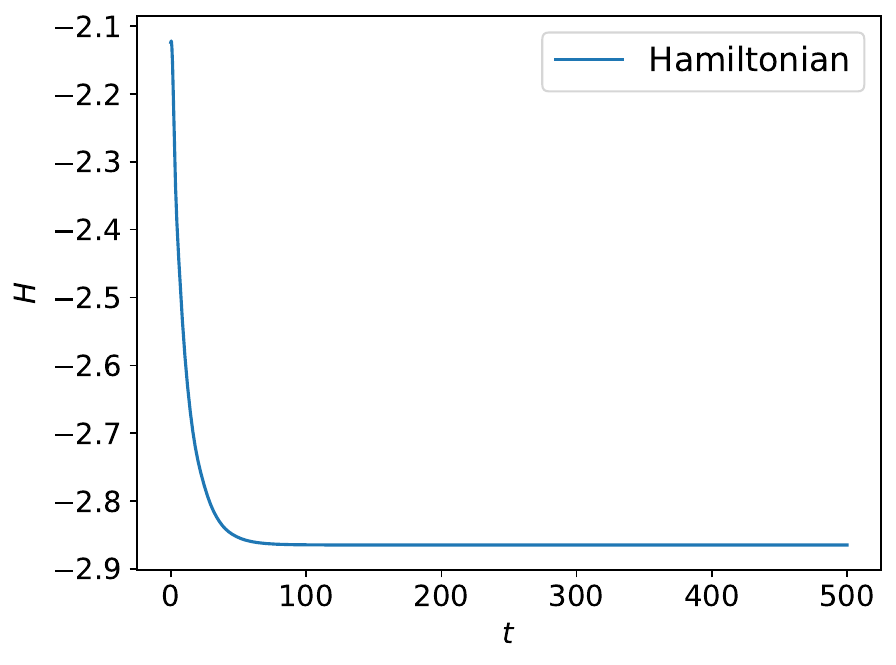}
        \caption[Network2]%
        {{Evolution of Hamiltonian $h_0$ in \eqref{eq:point-vortex-hamiltonian}}}    
    \label{fig:hamiltonian-evolution}
    \end{subfigure}%
    \hspace{40pt}
    \begin{subfigure}[t]{.4\textwidth} 
        \centering 
\includegraphics[width=.9\linewidth]{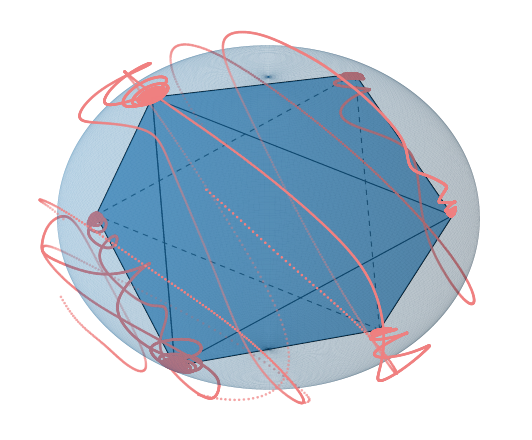}
        \caption[]%
        {{Dissipative point vortex trajectory}}    
        \label{fig:sample-trajectory}
    \end{subfigure}
    \caption{Numerical simulation of the dissipative point vortex system \eqref{eq:stoch-diss-point-vortex} with the noise coeffient $\sigma$ taken to zero. Here, we consider six vortices with unit strengths $\Gamma_n = 1$ for $n = 1, \ldots, 6$. The initial vortex locations are sampled randomly from i.i.d. uniform distribution on the sphere. We see in (a) that in the presence of the new dissipative term, the point vortex Hamiltonian \eqref{eq:point-vortex-hamiltonian} decreases under the dynamics \eqref{eq:stoch-diss-point-vortex}, as expected. The corresponding trajectory is plotted in (b), which shows the asymptotic convergence of the point vortex locations to equally spaced points on the sphere, i.e., the six vertices of the octahedron.}
    \label{fig:point_vortex_with_dissipation}
\end{figure}

Finally, we wish to understand the invariant measure of system \eqref{eq:stoch-diss-point-vortex}, which we believe will help us to understand the invariant measure of the infinite dimensional system \eqref{eq:euler-hamiltonians-2d}. From Corollary \ref{cor:lie-poisson-gibbs-measure}, the Gibbs measure on the singular coadjoint orbit \eqref{eq:singular-coadjoint-orbit} reads
\begin{align}
    \mathbb{P}_\infty^{\omega_0} \,&\propto\, e^{-\frac{\beta}{2}h_0(\{x_n\}_{n=1}^N)} \prod_{n=1}^N |\Gamma_n| \,\mathrm{vol}_{S^2}^n \nonumber \\
    &\stackrel{\eqref{eq:point-vortex-hamiltonian}}{\propto} \prod_{m,n=1}^N \|x_m - x_n\|^{\frac{\beta \Gamma_m \Gamma_n}{8\pi}}\prod_{n=1}^N |\Gamma_n| \,\mathrm{vol}_{S^2}^n \,. \label{eq:gibbs-point-vortex}
\end{align}
We see that when $\Gamma_n > 0$ or $\Gamma_n < 0$ for all $n=1, \ldots, N$, it assigns high probability to vortex configurations that are spaced out across the sphere, over those that are clustered. Indeed, Figure \ref{fig:sample-trajectory} shows that in the purely dissipative setting with unit vortex strengths, the vortices converge to a low-energy configuration that is equally spaced out on the sphere. In the continuous limit, we can imagine the invariant measure in the infinite-dimensional space assigning high probabilities to purely positive or negative vorticity fields that have minimal fluctuations.

When the $\Gamma_n$'s are allowed to be either positive or negative, we see that it prefers a highly mixed configuration where opposite-signed vortices are attracted to one another and like-signed vortices are spread apart. When we allow $\beta < 0$ (this is the {\em negative temperature} regime, introduced in \cite{onsager1949statistical}), then an opposite phenomenon occurs where it becomes highly likely for like-signed vortices to attract one another, forming large clusters of vorticity.
However, note that \eqref{eq:gibbs-point-vortex} may not actually define a probability measure, as it becomes non-integrable near the collision set $\mathcal{C}$ when $|\beta \Gamma_m \Gamma_n|$ is not large enough. Thus, in the continuous limit when the vorticity field $\omega_0$ is smooth and can take either positive or negative values, an invariant probability measure to system \eqref{eq:salt-euler-with-dissipation} may not exist.

\section{Summary and discussion}

In this article, we establish a mathematically sound variational principle which permits the derivation of structure-preserving stochastic dynamical equations. Specifically, we introduce the stochastic Hamilton-Pontryagin principle to derive a stochastic extension to the Euler-Lagrange equations, which, upon taking the Legendre transform, can be shown to be equivalent to Bismut's symplectic diffusions \cite{bismut1981mechanique}. A key assumption that enables us to make sense of this variational principle is the semimartingale compatibility assumption (Definition \ref{def:compatibility}), which, in simple terms, states that the paths in the variational family admit a ``derivative" with respect to the driving semimartingale. On $\mathbb{R}^n$, this is a fairly natural assumption by the martingale representation theorem, stating that compatibility in our sense holds provided the processes are adapted to the driving semimartingale. Further, adaptedness is necessary regardless in order to define the stochastic integrals in the action functional. Should an analogous result hold true for manifold-valued processes, then it will be sufficient to just assume adaptedness, making the compatibility assumption extraneous. While we are not aware of such a result, it is not hard to imagine that this is true, given many classical results in stochastic analysis can be extended to the manifold setting \cite{norris1992complete}.

The compatibility condition again plays an important role when considering a stochastic extension of the Euler-Poincaré reduction theorem. Here, it is essential as the reduced action functional cannot be defined otherwise. When considering the stochastic Euler-Poincaré reduction, we also use a particular variational family introduced in \cite{ACC2014}, which makes the variational processes satisfy the endpoint conditions, while also retaining adaptedness. These ingredients enable one to make proper sense of Euler-Poincaré reduction in the stochastic setting.

We then consider a framework for introducing dissipation to our stochastic system that balances the noise term, so that the Gibbs measure is an invariant measure of the system. The preservation of Gibbs measure is fundamental in statistical mechanics as it models a system (at thermodynamic equilibrium) coupled to an external heat bath, exchanging energy while keeping the temperature fixed. On a symplectic manifold, a dissipative extension of the symplectic diffusion considered in \cite{arnaudon2019irreversible} is such a process, preserving the Gibbs measure defined on the symplectic manifold \cite{souriau1970structure}. By applying symmetry reduction to this system, we can derive non-canonical versions of the stochastic-dissipative system, which preserve the Gibbs measure on the symplectic leaves. In particular, we show that applying Lie-Poisson reduction to the dissipative term results in the double-bracket dissipation \cite{bloch1994dissipation}. Thus, this gives us a new derivation of the double-bracket dissipation from the point of view of symmetry reduction.

There are some outstanding issues we have not addressed in this paper that may be of interest to consider in the future. As stated earlier, it would be interesting to establish an extension of the martingale representation theorem to the manifold setting, and with respect to the general driving semimartingale given in Definition \ref{def:driving_semimartingale} instead of the usual Brownian setting (if this hasn't yet been established). Specifically, the question is whether any $\mathcal{F}_t$-adapted $Q$-valued semimartingale $q_t$ is compatible (in the sense of Definition \ref{def:compatibility}) with the driving semimartingale $S_t$ generating the filtration. Or if not, it would be interesting to know the conditions on $Q$ such that this result is true. Having such a result will strengthen our statement of the variational principle, as the compatibility assumption simply becomes a consequence of adaptedness, making the statement more natural. 

Another interesting question to address is whether the stochastic dissipative system can be derived from a variational principle. While we have only considered dissipation from the Hamiltonian side, having a derivation from the Lagrangian side may be of interest as it will allow us to consider structure-preserving discretisations by means of variational integrators \cite{bou2009stochastic, holm2018stochastic, kraus2021variational}. Dissipation is traditionally incoorporated into variational principles via the somewhat unsatisfying Lagrange-d'Alembert approach. However, recent advances in numerical optimisation has lead to a variational perspective on gradient dynamics by considering Lagrangians scaled by a time-varying factor \cite{wibisono2016variational, tao2020variational}. It would be interesting to see if the dissipative mechanism we consider here admits a variational structure of such a form, and if so, whether it can be combined with the stochastic action principle to yield stochastic-dissipative equations.

Finally, while our work only considers the finite dimensional setting, the extent to which this holds in infinite dimensions is worth further discussion. In particular, we would be interested in seeing what the invariant measures in the infinite dimensional setting correspond to, for example, for system \eqref{eq:salt-euler-with-dissipation}. For the Euler and Navier-Stokes equations with additive noise, its invariant measures have been studied, for example in \cite{albeverio1990global, flandoli1994dissipativity}, and their ergodic property has been studied, for example in \cite{hairer2006ergodicity}. In the structure-preserving system \eqref{eq:salt-euler-with-dissipation}, the problem may be more challenging as we expect the Gibbsian measure to be supported on an infinite-dimensional manifold (namely the coadjoint orbits \eqref{eq:diffeo-coadjoint-orbits}), rather than a topological vector space, which at first glance seems hopelessly difficult to tackle. However, there may be an easier way to approach this, for example by considering a finite-dimensional approximation, as we do in the point-vortex example, and considering the limiting measure, or characterise the measure in an ambient topological vector space and restrict its support on the submanifold of interest. The latter idea is analogous to how we can characterise the invariant measure on the coadjoint orbits of the stochastic-dissipative rigid body as the restriction of a Gaussian measure on $\mathbb{R}^3$ onto the $2$-sphere (see Example \ref{sec:stoch-diss-rigid-body}). By a parallel argument, it may be that the invariant measure of its infinite dimensional extension \eqref{eq:diffeo-coadjoint-orbits} can be given by a restriction of some Gaussian measure on a Banach space of the vorticity fields, on the coadjoint orbits \eqref{eq:diffeo-coadjoint-orbits}. We see that this motivates several interesting mathematical questions that will be important in understanding the statistical behaviour of our stochastic-dissipative system in the infinite-dimensional setting.





\paragraph{Acknowledgements.} We thank Darryl D. Holm, James-Michael Leahy, and Ruiao Hu for the many useful discussions that have helped form the techniques we employed in the stochastic variational principle. Oliver D. Street was supported during the present work by European Research Council (ERC) Synergy grant `Stochastic Transport in Upper Ocean Dynamics' (STUOD) – DLV-856408.

\bibliographystyle{alpha}
\bibliography{biblio}

\newcommand{\etalchar}[1]{$^{#1}$}
\begin{thebibliography}{CCH{\etalchar{+}}20b}

\bibitem[ABT19]{arnaudon2019irreversible}
Alexis Arnaudon, Alessandro Barp, and So~Takao.
\newblock Irreversible {L}angevin {MCMC} on {L}ie groups.
\newblock In {\em Geometric Science of Information: 4th International
  Conference, GSI 2019, Toulouse, France, August 27--29, 2019, Proceedings 4},
  pages 171--179. Springer, 2019.

\bibitem[AC90]{albeverio1990global}
Sergio Albeverio and Ana-Bela Cruzeiro.
\newblock Global flows with invariant ({G}ibbs) measures for {E}uler and
  {N}avier-{S}tokes two dimensional fluids.
\newblock {\em Communications in mathematical physics}, 129:431--444, 1990.

\bibitem[ACC14]{ACC2014}
Marc Arnaudon, Xin Chen, and Ana~Bela Cruzeiro.
\newblock {Stochastic Euler-Poincaré reduction}.
\newblock {\em Journal of Mathematical Physics}, 55(8), 08 2014.
\newblock 081507.

\bibitem[ADCH18]{arnaudon2018noise}
Alexis Arnaudon, Alex~L De~Castro, and Darryl~D Holm.
\newblock Noise and dissipation on coadjoint orbits.
\newblock {\em Journal of nonlinear science}, 28:91--145, 2018.

\bibitem[AHS19]{arnaudon2019geometric}
Alexis Arnaudon, Darryl~D Holm, and Stefan Sommer.
\newblock A geometric framework for stochastic shape analysis.
\newblock {\em Foundations of Computational Mathematics}, 19:653--701, 2019.

\bibitem[Arn66]{arnold1966sur}
Vladimir Arnol'd.
\newblock {Sur la geometrie differentielle des groupes de Lie de dimension
  infinite}.
\newblock {\em Annales de l'Institut Fourier}, 16:319--361, 1966.

\bibitem[Arn69]{arnold1969hamiltonian}
Vladimir Arnol'd.
\newblock {The Hamiltonian nature of the Euler equations in the dynamics of a
  rigid body and of an ideal fluid}.
\newblock {\em Uspekhi Mat. Nauk}, 24(3):147, 1969.

\bibitem[AT19]{arnaudon2019networks}
Alexis Arnaudon and So~Takao.
\newblock Networks of coadjoint orbits: From geometric to statistical
  mechanics.
\newblock {\em Journal of Geometric Mechanics}, 11(4), 2019.

\bibitem[BAB{\etalchar{+}}17]{berner2017stochastic}
Judith Berner, Ulrich Achatz, Lauriane Batte, Lisa Bengtsson, Alvaro
  De~La~Camara, Hannah~M Christensen, Matteo Colangeli, Danielle~RB Coleman,
  Daan Crommelin, Stamen~I Dolaptchiev, et~al.
\newblock Stochastic parameterization: Toward a new view of weather and climate
  models.
\newblock {\em Bulletin of the American Meteorological Society},
  98(3):565--588, 2017.

\bibitem[BBNP13]{banas2013stochastic}
Lubomir Banas, Zdzislaw Brzezniak, Mikhail Neklyudov, and Andreas Prohl.
\newblock {\em Stochastic ferromagnetism: analysis and numerics}, volume~58.
\newblock Walter de Gruyter, 2013.

\bibitem[BdLT22]{bethencourt2022transport}
Aythami Bethencourt-de Le{\'o}n and So~Takao.
\newblock Transport noise restores uniqueness and prevents blow-up in geometric
  transport equations.
\newblock {\em arXiv preprint arXiv:2211.14695}, 2022.

\bibitem[Bis81]{bismut1981mechanique}
Jean-Michel Bismut.
\newblock Mechanique al{\'e}atoire lecture notes in mathematics 866, 1981.

\bibitem[BKMR94]{bloch1994dissipation}
Anthony~M Bloch, PS~Krishnaprasad, Jerrold~E Marsden, and Tudor~S Ratiu.
\newblock Dissipation induced instabilities.
\newblock In {\em Annales de l'Institut Henri Poincar{\'e} C, Analyse non
  lin{\'e}aire}, volume~11, pages 37--90. Elsevier, 1994.

\bibitem[BKMR96]{bloch1996euler}
Anthony Bloch, PS~Krishnaprasad, Jerrold~E Marsden, and Tudor~S Ratiu.
\newblock The {E}uler-{P}oincar{\'e} equations and double bracket dissipation.
\newblock {\em Communications in mathematical physics}, 175(1):1--42, 1996.

\bibitem[Bro91]{brockett1991dynamical}
Roger~W Brockett.
\newblock Dynamical systems that sort lists, diagonalize matrices, and solve
  linear programming problems.
\newblock {\em Linear Algebra and its applications}, 146:79--91, 1991.

\bibitem[BRO09]{bou2009stochastic}
Nawaf Bou-Rabee and Houman Owhadi.
\newblock Stochastic variational integrators.
\newblock {\em IMA Journal of Numerical Analysis}, 29(2):421--443, 2009.

\bibitem[BTB{\etalchar{+}}21]{barp2021unifying}
Alessandro Barp, So~Takao, Michael Betancourt, Alexis Arnaudon, and Mark
  Girolami.
\newblock A unifying and canonical description of measure-preserving
  diffusions.
\newblock {\em arXiv preprint arXiv:2105.02845}, 2021.

\bibitem[CCH{\etalchar{+}}20a]{cotter2020data}
Colin Cotter, Dan Crisan, Darryl Holm, Wei Pan, and Igor Shevchenko.
\newblock Data assimilation for a quasi-geostrophic model with
  circulation-preserving stochastic transport noise.
\newblock {\em Journal of Statistical Physics}, 179(5-6):1186--1221, 2020.

\bibitem[CCH{\etalchar{+}}20b]{cotter2020particle}
Colin Cotter, Dan Crisan, Darryl~D Holm, Wei Pan, and Igor Shevchenko.
\newblock A particle filter for stochastic advection by {L}ie transport: A case
  study for the damped and forced incompressible two-dimensional {E}uler
  equation.
\newblock {\em SIAM/ASA Journal on Uncertainty Quantification},
  8(4):1446--1492, 2020.

\bibitem[CG01]{czopnik2001brownian}
Rados{\l}aw Czopnik and Piotr Garbaczewski.
\newblock Brownian motion in a magnetic field.
\newblock {\em Physical Review E}, 63(2):021105, 2001.

\bibitem[CHLN22]{crisan2022variational}
Dan Crisan, Darryl~D Holm, James-Michael Leahy, and Torstein Nilssen.
\newblock Variational principles for fluid dynamics on rough paths.
\newblock {\em Advances in Mathematics}, 404:108409, 2022.

\bibitem[CM87]{cendra1987lin}
Hernan Cendra and Jerrold~E Marsden.
\newblock Lin constraints, {C}lebsch potentials and variational principles.
\newblock {\em Physica D: Nonlinear Phenomena}, 27(1-2):63--89, 1987.

\bibitem[DHP23]{diamantakis2023variational}
Theo Diamantakis, Darryl~D Holm, and Grigorios~A Pavliotis.
\newblock Variational principles on geometric rough paths and the {L}{\'e}vy
  area correction.
\newblock {\em SIAM Journal on Applied Dynamical Systems}, 22(2):1182--1218,
  2023.

\bibitem[FF13]{fedrizzi2013noise}
Ennio Fedrizzi and Franco Flandoli.
\newblock Noise prevents singularities in linear transport equations.
\newblock {\em Journal of Functional Analysis}, 264(6):1329--1354, 2013.

\bibitem[FGP10]{flandoli2010well}
Franco Flandoli, Massimiliano Gubinelli, and Enrico Priola.
\newblock Well-posedness of the transport equation by stochastic perturbation.
\newblock {\em Inventiones mathematicae}, 180(1):1--53, 2010.

\bibitem[Fla94]{flandoli1994dissipativity}
Franco Flandoli.
\newblock {Dissipativity and invariant measures for stochastic Navier-Stokes
  equations}.
\newblock {\em Nonlinear Differential Equations and Applications NoDEA},
  1(4):403--423, 1994.

\bibitem[GBH13]{gay2013selective}
Fran{\c{c}}ois Gay-Balmaz and Darryl~D Holm.
\newblock {Selective decay by Casimir dissipation in inviscid fluids}.
\newblock {\em Nonlinearity}, 26(2):495, 2013.

\bibitem[GBH18]{gay2018stochastic}
Fran{\c{c}}ois Gay-Balmaz and Darryl~D Holm.
\newblock {Stochastic geometric models with non-stationary spatial correlations
  in Lagrangian fluid flows}.
\newblock {\em Journal of nonlinear science}, 28:873--904, 2018.

\bibitem[GBR09]{GBR2009}
Fran{\c c}ois Gay-Balmaz and Tudor~S. Ratiu.
\newblock The geometric structure of complex fluids.
\newblock {\em Advances in Applied Mathematics}, 42(2):176--275, 2009.

\bibitem[GBTV13]{gay2013geometric}
Fran{\c{c}}ois Gay-Balmaz, C~Tronci, and C~Vizman.
\newblock Geometric dynamics on the automorphism group of principal bundles:
  Geodesic flows, dual pairs and chromomorphism groups.
\newblock {\em Journal of Geometric Mechanics}, 5(1):39--84, 2013.

\bibitem[GF00]{gelfand2000calculus}
Izrail~Moiseevitch Gelfand and S.~V. Fomin.
\newblock {\em Calculus of variations}.
\newblock Courier Corporation, 2000.

\bibitem[HL21]{holm-luesink2021}
Darryl~D. Holm and Erwin Luesink.
\newblock Stochastic wave--current interaction in thermal shallow water
  dynamics.
\newblock {\em Journal of Nonlinear Science}, 31(2):29, 2021.

\bibitem[HM06]{hairer2006ergodicity}
Martin Hairer and Jonathan~C Mattingly.
\newblock {Ergodicity of the 2D Navier-Stokes equations with degenerate
  stochastic forcing}.
\newblock {\em Annals of Mathematics}, pages 993--1032, 2006.

\bibitem[HMR98]{HMR1998}
Darryl~D Holm, Jerrold~E Marsden, and Tudor~S Ratiu.
\newblock The {E}uler--{P}oincar{\'e} equations and semidirect products with
  applications to continuum theories.
\newblock {\em Advances in Mathematics}, 137(1):1--81, 1998.

\bibitem[Hol15]{holm2015variational}
Darryl~D Holm.
\newblock Variational principles for stochastic fluid dynamics.
\newblock {\em Proceedings of the Royal Society A: Mathematical, Physical and
  Engineering Sciences}, 471(2176):20140963, 2015.

\bibitem[HT18]{holm2018stochastic}
Darryl~D Holm and Tomasz~M Tyranowski.
\newblock {Stochastic discrete Hamiltonian variational integrators}.
\newblock {\em BIT Numerical Mathematics}, 58:1009--1048, 2018.

\bibitem[IK74]{ichihara1974classification}
Kanji Ichihara and Hiroshi Kunita.
\newblock A classification of the second order degenerate elliptic operators
  and its probabilistic characterization.
\newblock {\em Zeitschrift f{\"u}r Wahrscheinlichkeitstheorie und Verwandte
  Gebiete}, 30(3):235--254, 1974.

\bibitem[It{\^o}50]{ito1950brownian}
Kiyosi It{\^o}.
\newblock {Brownian motions in a Lie group}.
\newblock {\em Proceedings of the Japan Academy}, 26(8):4--10, 1950.

\bibitem[Kal03]{kalnay2003atmospheric}
Eugenia Kalnay.
\newblock {\em Atmospheric modeling, data assimilation and predictability}.
\newblock Cambridge university press, 2003.

\bibitem[KT21]{kraus2021variational}
Michael Kraus and Tomasz~M Tyranowski.
\newblock {Variational integrators for stochastic dissipative Hamiltonian
  systems}.
\newblock {\em IMA Journal of Numerical Analysis}, 41(2):1318--1367, 2021.

\bibitem[Kun84]{kunita1984stochastic}
H~Kunita.
\newblock Stochastic differential equations and stochastic flows of
  diffeomorphisms.
\newblock In {\em Ecole d'{\'e}t{\'e} de probabilit{\'e}s de Saint-Flour
  XII-1982}, pages 143--303. Springer, 1984.

\bibitem[LCO07]{lazaro2007stochastic}
Joan-Andreu L{\'a}zaro-Cam{\'\i} and Juan-Pablo Ortega.
\newblock {Stochastic Hamiltonian dynamical systems}.
\newblock {\em arXiv preprint math/0702787}, 2007.

\bibitem[Mar92]{marsden1992lectures}
Jerrold~E. Marsden.
\newblock {\em Lectures on Mechanics}, volume 174 of {\em London Mathematical
  Society Lecture Note Series}.
\newblock Cambridge University Press, 1992.

\bibitem[MM13]{michor2013zoo}
Peter~W Michor and David Mumford.
\newblock A zoo of diffeomorphism groups on $\mathbb{R}^{n}$.
\newblock {\em Annals of Global Analysis and Geometry}, 44:529--540, 2013.

\bibitem[MM21]{mimeau2021review}
Chlo{\'e} Mimeau and Iraj Mortazavi.
\newblock A review of vortex methods and their applications: From creation to
  recent advances.
\newblock {\em Fluids}, 6(2):68, 2021.

\bibitem[MMM21a]{mongwe2021magnetic}
Wilson~Tsakane Mongwe, Rendani Mbuvha, and Tshilidzi Marwala.
\newblock {Magnetic Hamiltonian Monte Carlo with partial momentum refreshment}.
\newblock {\em IEEE Access}, 9:108009--108016, 2021.

\bibitem[MMM21b]{mongwe2021quantum}
Wilson~Tsakane Mongwe, Rendani Mbuvha, and Tshilidzi Marwala.
\newblock {Quantum-inspired magnetic Hamiltonian Monte Carlo}.
\newblock {\em Plos one}, 16(10):e0258277, 2021.

\bibitem[Mon84]{montgomery1984canonical}
Richard Montgomery.
\newblock {Canonical formulations of a classical particle in a Yang-Mills field
  and Wong's equations}.
\newblock {\em letters in mathematical physics}, 8:59--67, 1984.

\bibitem[MR86]{marsden1986reduction}
Jerrold~E Marsden and Tudor Ratiu.
\newblock Reduction of {P}oisson manifolds.
\newblock {\em Letters in mathematical Physics}, 11(2):161--169, 1986.

\bibitem[MR13]{marsden2013introduction}
Jerrold~E Marsden and Tudor~S Ratiu.
\newblock {\em Introduction to mechanics and symmetry: a basic exposition of
  classical mechanical systems}, volume~17.
\newblock Springer Science \& Business Media, 2013.

\bibitem[MRW84]{marsden1984semidirect}
Jerrold~E Marsden, Tudor Ra{\c{t}}iu, and Alan Weinstein.
\newblock Semidirect products and reduction in mechanics.
\newblock {\em Transactions of the american mathematical society},
  281(1):147--177, 1984.

\bibitem[MV94]{michor1994n}
Peter~W Michor and Cornelia Vizman.
\newblock $ n $-transitivity of certain diffeomorphism groups.
\newblock {\em arXiv preprint dg-ga/9406005}, 1994.

\bibitem[MW74]{marsden1974reduction}
Jerrold Marsden and Alan Weinstein.
\newblock Reduction of symplectic manifolds with symmetry.
\newblock {\em Reports on mathematical physics}, 5(1):121--130, 1974.

\bibitem[MW83]{marsden1983coadjoint}
Jerrold Marsden and Alan Weinstein.
\newblock Coadjoint orbits, vortices, and clebsch variables for incompressible
  fluids.
\newblock {\em Physica D: Nonlinear Phenomena}, 7(1-3):305--323, 1983.

\bibitem[Nel66]{nelson1966derivation}
Edward Nelson.
\newblock Derivation of the {S}chr{\"o}dinger equation from {N}ewtonian
  mechanics.
\newblock {\em Physical review}, 150(4):1079, 1966.

\bibitem[Nel85]{nelson2020quantum}
Edward Nelson.
\newblock {\em Quantum fluctuations}, volume 108.
\newblock Princeton University Press, 1985.

\bibitem[New02]{newton2002n}
Paul~K Newton.
\newblock N-vortex problem: Analytical techniques.
\newblock {\em Appl. Mech. Rev.}, 55(1):B15--B16, 2002.

\bibitem[Nor92]{norris1992complete}
James~R Norris.
\newblock A complete differential formalism for stochastic calculus in
  manifolds.
\newblock {\em S{\'e}minaire de probabilit{\'e}s de Strasbourg}, 26:189--209,
  1992.

\bibitem[Oks13]{oksendal2013stochastic}
Bernt Oksendal.
\newblock {\em Stochastic differential equations: an introduction with
  applications}.
\newblock Springer Science \& Business Media, 2013.

\bibitem[Ons49]{onsager1949statistical}
Lars Onsager.
\newblock Statistical hydrodynamics.
\newblock {\em Il Nuovo Cimento (1943-1954)}, 6(Suppl 2):279--287, 1949.

\bibitem[OR05]{ortega2005cotangent}
Juan-Pablo Ortega and Tudor~S Ratiu.
\newblock Cotangent bundle reduction.
\newblock {\em arXiv preprint math/0508636}, 2005.

\bibitem[Poi01]{poincare1901}
Henri Poincar{\'e}.
\newblock Sur une forme nouvelle des {\'e}quations de la m{\'e}chanique.
\newblock {\em C. R. Acad. Sci. Paris}, CXXXII(7):369–371, 1901.

\bibitem[RSL10]{rousset2010free}
Mathias Rousset, Gabriel Stoltz, and Tony Lelievre.
\newblock {\em Free energy computations: a mathematical perspective}.
\newblock World Scientific, 2010.

\bibitem[Sal88]{salmon1988hamiltonian}
Rick Salmon.
\newblock Hamiltonian fluid mechanics.
\newblock {\em Annual review of fluid mechanics}, 20(1):225--256, 1988.

\bibitem[SC21]{street2021semi}
Oliver~D Street and Dan Crisan.
\newblock Semi-martingale driven variational principles.
\newblock {\em Proceedings of the Royal Society A}, 477(2247):20200957, 2021.

\bibitem[She90]{shepherd1990general}
Theodore~G Shepherd.
\newblock A general method for finding extremal states of {H}amiltonian
  dynamical systems, with applications to perfect fluids.
\newblock {\em Journal of Fluid Mechanics}, 213:573--587, 1990.

\bibitem[Sou70]{souriau1970structure}
J-M Souriau.
\newblock {\em Structure des syst{\`e}mes dynamiques: ma{\^\i}trises de
  math{\'e}matiques}.
\newblock {\'E}ditions Dunod, 1970.

\bibitem[Ste77]{sternberg1977minimal}
Shlomo Sternberg.
\newblock {Minimal coupling and the symplectic mechanics of a classical
  particle in the presence of a Yang-Mills field}.
\newblock {\em Proceedings of the National Academy of Sciences},
  74(12):5253--5254, 1977.

\bibitem[Str23]{street2023waterwaves}
Oliver~D. Street.
\newblock A structure preserving stochastic perturbation of classical water
  wave theory.
\newblock {\em Physica D: Nonlinear Phenomena}, 447:133689, 2023.

\bibitem[TO20]{tao2020variational}
Molei Tao and Tomoki Ohsawa.
\newblock {Variational optimization on Lie groups, with examples of leading
  (generalized) eigenvalue problems}.
\newblock In {\em International Conference on Artificial Intelligence and
  Statistics}, pages 4269--4280. PMLR, 2020.

\bibitem[TRGT17]{tripuraneni2017magnetic}
Nilesh Tripuraneni, Mark Rowland, Zoubin Ghahramani, and Richard Turner.
\newblock {Magnetic Hamiltonian Monte Carlo}.
\newblock In {\em International Conference on Machine Learning}, pages
  3453--3461. PMLR, 2017.

\bibitem[VCY89]{vallis1989extremal}
GK~Vallis, GF~Carnevale, and WR~Young.
\newblock Extremal energy properties and construction of stable solutions of
  the {E}uler equations.
\newblock {\em Journal of Fluid Mechanics}, 207:133--152, 1989.

\bibitem[WWJ16]{wibisono2016variational}
Andre Wibisono, Ashia~C Wilson, and Michael~I Jordan.
\newblock A variational perspective on accelerated methods in optimization.
\newblock {\em proceedings of the National Academy of Sciences},
  113(47):E7351--E7358, 2016.

\bibitem[YM06a]{yoshimura2006diracI}
Hiroaki Yoshimura and Jerrold~E Marsden.
\newblock {D}irac structures in {L}agrangian mechanics {P}art {I}: implicit
  {L}agrangian systems.
\newblock {\em Journal of Geometry and Physics}, 57(1):133--156, 2006.

\bibitem[YM06b]{yoshimura2006diracII}
Hiroaki Yoshimura and Jerrold~E Marsden.
\newblock {D}irac structures in {L}agrangian mechanics {P}art {II}:
  {V}ariational structures.
\newblock {\em Journal of Geometry and Physics}, 57(1):209--250, 2006.

\bibitem[YM07]{YOSHIMURA2007381}
Hiroaki Yoshimura and Jerrold~E Marsden.
\newblock Reduction of {D}irac structures and the {H}amilton-{P}ontryagin
  principle.
\newblock {\em Reports on Mathematical Physics}, 60(3):381--426, 2007.

\bibitem[ZY82]{zambrini1982semi}
J-C Zambrini and K~Yasue.
\newblock Semi-classical quantum mechanics and stochastic calculus of
  variations.
\newblock {\em Annals of Physics}, 143(1):54--83, 1982.

\end{thebibliography}

\begin{appendices}
\setcounter{lemma}{0}
    \renewcommand{\thelemma}{\Alph{section}.\arabic{lemma}}
\setcounter{proposition}{0}
    \renewcommand{\theproposition}{\Alph{section}.\arabic{proposition}}

\section{Intrinsic Stochastic Euler-Lagrange Equations}\label{app:intrinsic-stoch-EL}

Here, we consider the intrinsic formulation of the stochastic Euler-Lagrange equations \eqref{eq:stoch-euler-lagrange-mom} -- \eqref{eq:stoch-euler-lagrange-pos}, without reliance on local coordinates. We start by introducing the geometric structure of the Pontryagin bundle $TQ \oplus T^*Q$ and higher-order bundles of $Q$, as discussed in \cite{yoshimura2006diracI, yoshimura2006diracII}.
The Pontryagin bundle is canonically equipped with three projections
\begin{align}
    \mathrm{pr}_{TQ} &: TQ \oplus T^*Q \rightarrow TQ, \label{eq:proj-V} \\
    \mathrm{pr}_{T^*Q} &: TQ \oplus T^*Q \rightarrow T^*Q, \label{eq:proj-p} \\
    \mathrm{pr}_{Q} &: TQ \oplus T^*Q \rightarrow Q, \label{eq:proj-q}
\end{align}
where the first two are the projections arising from the direct sum structure in the tangent spaces and the last follows from the projection arising from $TQ \oplus T^*Q$, viewed as a vector bundle over $Q$. We also have a canonical map
\begin{align*}
    G : TQ \oplus T^*Q \rightarrow \mathbb{R},
\end{align*}
defined by $G : (q, v_q, \alpha_q) \mapsto \left<\alpha_q, v_q\right>_{T_q^*Q \times T_qQ}$, referred to as the momentum function.
For an arbitrary vector bundle $E$ over $Q$, we denote the canonical projections on its tangent and cotangent bundle by
\begin{align*}
    \tau_{E} : TE \rightarrow E \quad \text{and} \quad \pi_{E} : T^*E \rightarrow E,
\end{align*}
respectively. Using this, we can show that the second order bundles $TT^*Q$ and $T^*TQ$ are embedded in $TQ \oplus T^*Q$ via the maps
\begin{align*}
    \rho_{T^*TQ} &:= \pi_{TQ} \oplus \tau_{T^*Q} \circ \kappa_Q^{-1} : T^*TQ \rightarrow TQ \oplus T^*Q, \\
    \rho_{TT^*Q} &:= \pi_{TQ} \circ \kappa_Q \oplus \tau_{T^*Q} : TT^*Q \rightarrow TQ \oplus T^*Q,
\end{align*}
where $\kappa_Q : TT^*Q \stackrel{\sim}{\rightarrow} T^*TQ$ is the isomorphism between the two bundles (see \cite[Proposition 4.1]{yoshimura2006diracI} for the existence of such a map).

Now let $\Omega$ be the canonical symplectic form on $T^*Q$. This induces a natural isomorphism
\begin{align*}
    \Omega^\flat : TT^*Q \rightarrow T^*T^*Q,
\end{align*}
defined by
\begin{align*}
    \left<V, \Omega^\flat(W)\right>_{TT^*Q \times T^*T^*Q} = \Omega(V, W),
\end{align*}
for any $V, W \in TT^*Q$. We also denote by $\Theta_{T^*T^*Q}$ the tautological one-form on $T^*T^*Q$. That is,
\begin{align*}
    \Theta_{T^*T^*Q}|_\alpha(V_\alpha) = \left<\alpha, T\pi_{T^*Q} \cdot V_\alpha \right>_{T^*T^*Q \times TT^*Q},
\end{align*}
for any $\alpha \in T^*T^*Q$ and $V_\alpha \in T_\alpha(T^*T^*Q)$. Combining this, we get a one-form $\chi$ on $TT^*Q$ by
\begin{align*}
    \chi := (\Omega^\flat)^* \Theta_{T^*T^*Q}.
\end{align*}

We now have all the objects set up to express the Hamilton-Pontryagin action functional in intrinsic form. Let us define the {\em generalised energy} $E_i : TQ \oplus T^*Q \rightarrow \mathbb{R}$ for $i \geq 0$ by
\begin{align}
    E_0(Z) &= G(Z) - L(\mathrm{pr}_{TQ}(Z)) \,, \label{eq:E0} \\
    E_i(Z) &= \left<\mathrm{pr}_{T^*Q}(Z), \,\Xi_i(\mathrm{pr}_{Q}(Z))\right>_{T^*Q \times TQ} - \Gamma_i(\mathrm{pr}_{Q}(Z)) \,, \quad i \geq 1 \,,\label{eq:Ei}
\end{align}
for any $Z \in TQ \oplus T^*Q$, where we recall that $\Xi_i \in \mathfrak{X}(TQ)$ for $i \geq 1$ are the noise vector fields and $\Gamma_i \in \Omega^0(Q)$ for $i \geq 1$ are the stochastic potentials. Furthermore, for a process $Z_t \in TQ \oplus T^*Q$ that is compatible with $\{S_t^i\}_{i \geq 0}$, i.e., there exists $T(TQ \oplus T^*Q)$-valued semimartingales $\{F_t^i\}_{i \geq 0}$ such that
\begin{align} \label{eq:Z-decomposition-pontryagin}
    \diff Z_t = \sum_{i \geq 0} F_t^i \circ \diff S_t^i \,,
\end{align}
we define a $\mathbb{R}$-valued semimartingale
\begin{align}
    \int G((\rho_{TT^*Q} \circ T\mathrm{pr}_{T^*Q}) \circ \diff Z_t) := \sum_{i \geq 0} \int G(\rho_{TT^*Q} \circ T\mathrm{pr}_{T^*Q}F_t^i) \circ \diff S_t^i \,.
\end{align}
Note that this is well-defined by the uniqueness of the decomposition \eqref{eq:Z-decomposition-pontryagin} (see Corollary \ref{cor:uniqueness-of-decomposition}).
Then, we define the Hamilton-Pontryagin action functional intrinsically by
\begin{align}\label{eq:HP-action-intrinsic}
    \mathcal{S}[Z_t] := \int^{t_1}_{t_0} G((\rho_{TT^*Q} \circ T\mathrm{pr}_{T^*Q}) \circ \diff Z_t) - \sum_{i \geq 0} E_i(Z_t) \circ \diff S_t^i \,.
\end{align}
In local coordinates, this functional has the desired form as we show in the following.
\begin{proposition}
    In local coordinates, the action functional \eqref{eq:HP-action-intrinsic} can be written as
    \begin{align}\label{eq:HP-action-coords}
        \mathcal{S} = \int_{t_0}^{t_1} L(q_t, V_t) \,\diff t + \sum_{i \geq 1} \Gamma_i(q_t)\circ \diff S_t^i + \scp{p_t}{\circ\,\diff q_t - V_t \,\diff t - \sum_{i \geq 1} \Xi_i(q_t) \circ \diff S_t^i} \,.
    \end{align}
\end{proposition}
\begin{proof}
    On a local chart $U \subset Q$, set $Z = (q, V, p) \in TU \oplus T^*U \cong U \times \mathbb{R}^n \times \mathbb{R}^n$, where $n := \mathrm{dim}(Q)$. In this chart, the projections \eqref{eq:proj-V}--\eqref{eq:proj-q} read
    \begin{align*}
        \mathrm{pr}_{TQ} : (q, V, p) \mapsto (q, V)\,, \quad \mathrm{pr}_{T^*Q} : (q, V, p) \mapsto (q, p)\,, \quad \mathrm{pr}_{Q} : (q, V, p) \mapsto q \,,
    \end{align*}
    and noting that $G(q, V, p) = \left<p, V\right>$, the generalised energies \eqref{eq:E0}--\eqref{eq:Ei} can be expressed as
    \begin{align}\label{eq:generalised-energies-coords}
        E_0(q, V, p) = \left<p, V\right> - L(q, V), \qquad E_i(q, V, p) = \left<p, \Xi_i(q)\right> - \Gamma_i(q) \,.
    \end{align}
    Next, let $(q, V, p, q', V', p') \in T(TU \oplus T^*U) \cong (U \times \mathbb{R}^n \times \mathbb{R}^n) \times (\mathbb{R}^n \times \mathbb{R}^n \times \mathbb{R}^n)$ be the local coordinates on the tangent bundle of $TQ \oplus T^*Q$. Then, we have the following local expressions for the maps $T\mathrm{pr}_{T^*Q} : T(TQ \oplus T^*Q) \rightarrow TT^*Q$ and $\rho_{TT^*Q} : TT^*Q \rightarrow TQ \oplus T^*Q$:
    \begin{align*}
        T\mathrm{pr}_{T^*Q} &: (q, V, p, q', V', p') \mapsto (q, p, q', p')\,, \\
        \rho_{TT^*Q} &: (q, p, q', p') \mapsto (q, q', p)\,.
    \end{align*}
    Further, by the compatibility assumption, the process $Z_t = (q_t, V_t, p_t)$ can be expressed as
    \begin{align*}
    \diff (q_t, V_t, p_t)^\top = \sum_{i \geq 0} (F_i^q(t), F_i^V(t), F_i^p(t))^\top \circ \diff S_t^i \,,
    \end{align*}
    where $(Z_t, F_i(t)) = (q_t, V_t, p_t, F_i^q(t), F_i^V(t), F_i^p(t)) \in T(TU \oplus T^*U)$. Then, we can show that
    \begin{align*}
        G(\rho_{TT^*Q} \circ T\mathrm{pr}_{T^*Q} (Z_t, F_i(t))) = \left<p_t, F_i^q(t)\right> \,,
    \end{align*}
    which, again by compatibility, is equivalent to writing
    \begin{align}\label{eq:momentum-term-coords}
        G((\rho_{TT^*Q} \circ T\mathrm{pr}_{T^*Q}) \circ \diff Z_t) = \left<p_t, \circ \,\diff q_t\right> \,.
    \end{align}
    Finally, combining the local expressions \eqref{eq:generalised-energies-coords}--\eqref{eq:momentum-term-coords}, we see that \eqref{eq:HP-action-intrinsic} is equivalent to \eqref{eq:HP-action-coords}.
\end{proof}

We now state the intrinsic version of the stochastic Hamilton-Pontryagin principle, which is a stochastic extension of the result found in \cite[Proposition 3.3]{yoshimura2006diracII}.
\begin{theorem}[Stochastic Hamilton-Pontryagin Principle: Intrinsic version]\label{thm:stochastic_HP_intrinsic}
    Taking the extrema of $\mathcal{S}$ given in \eqref{eq:HP-action-intrinsic} among all $TQ \oplus T^*Q$-valued continuous semimartingales $Z_t$ that is compatible with $\{S_t^i\}_{i \geq 0}$ and such that the endpoints of the base process are fixed, i.e., $\mathrm{pr}_Q(Z(t_0)) = a$ and $\mathrm{pr}_Q(Z(t_1)) = b$ for some $a, b \in Q$, we obtain the stochastic Euler-Lagrange equations:
    \begin{align}
        &\diff Z_t = \sum_{i\geq 1} F_i(t) \circ \diff S_t^i \,, \quad \text{where} \label{eq:intrinisic-stoch-EL-1}\\
        &(T\mathrm{pr}_{T^*Q})^*\chi(F_i(t)) = (T\tau_{TQ \oplus T^*Q})^* \dd E_i(Z_t, F_i(t)) \,, \quad \forall i \geq 0 \,. \label{eq:intrinisic-stoch-EL-2}
    \end{align}
\end{theorem}

\begin{proof}
    First, consider a partition $t_0 = t_{\epsilon_0} < t_{\epsilon_1} < \cdots < t_{\epsilon_{N-1}} < t_{\epsilon_N} = t_1$ such that for all $k = 0, \ldots, N-1$, there exists a local chart $U_k \subset Q$ such that $Z_t \in U_k$ for all $t \in [t_{\epsilon_k}, t_{\epsilon_{k+1}}]$. By the compatibility assumption, there exists a family of $T(TQ \oplus T^*Q)$-valued semimartingales $\{F_i(t)\}_{i \geq 0}$ such that \eqref{eq:intrinisic-stoch-EL-1} holds. Furthermore, taking variations on $Z_t$ induces variations on the processes $F_i$: that is, there exists $T(T(TQ \oplus T^*Q))$-valued semimartingales $\{\delta F_i(t)\}_{i \geq 0}$ such that $\diff (\delta Z_t) = \sum_{i\geq 0} \delta F_i(t) \circ \diff S_t^i$ holds. We claim that the following holds
    \begin{align}
        0 &= \delta \mathcal{S}[Z_t] \nonumber \\
        &= \sum_{i \geq 0} \int^{t_1}_{t_0} \left<\delta F_i(t), \, (T\mathrm{pr}_{T^*Q})^*\chi(F_i(t)) - (T\tau_{TQ \oplus T^*Q})^* \dd E_i(Z_t, F_i(t))\right>_{E^* \times E} \circ \diff S_t^i \nonumber \\
        &\quad + \left<\Theta_{T^*Q} \left(\mathrm{pr}_{T^*Q}(Z_t)\right), \,T\mathrm{pr}_{T^*Q} \cdot \delta Z_t\right>_{T^*T^*Q \times TT^*Q} \Big|^{t_1}_{t_0} \,, \label{eq:intrinsic-variations-HP}
    \end{align}
    where $E := T^*T(TQ \oplus T^*Q)$.
    To show this, we first decompose the integral above as $\int^{t_1}_{t_0} (\cdots) = \sum_{k=0}^{N-1} \int^{t_{\epsilon_{k+1}}}_{t_{\epsilon_k}} (\cdots) =: \sum_{k=0}^{N-1} I_k$ so that on each sub-integral, we can operate on local charts.
    Now, on a local chart, write $F_i(t) = (q_t, V_t, p_t, F^q_i(t), F^V_i(t), F^p_i(t))$. Then, we have the local expressions
    \begin{align}
        (T\mathrm{pr}_{T^*Q})^*\chi(F_i(t)) &= \sum_{j=1}^n \left(-[F^p_i(t)]_j \,\dd q_t^j + [F^q_i(t)]_j \,\dd p_t^j \right) \\
        (T\tau_{TQ \oplus T^*Q})^* \dd E_i(Z_t, F_i(t)) &= \sum_{j=1}^n \left(\frac{\partial E_i}{\partial q_t^j} \dd q_t^j + \frac{\partial E_i}{\partial V_t^j} \dd V_t^j + \frac{\partial E_i}{\partial p_t^j} \dd p_t^j\right) \,.
    \end{align}
    Using this, we have the following local expression for a sub-integral:
    \begin{align*}
        I_k &= \sum_{i \geq 0} \int^{t_{\epsilon_{k+1}}}_{t_{\epsilon_{k+1}}} \left(- \left<\delta q_t, F_i^p(t) + \frac{\partial E_i}{\partial q_t}\right> + \left<\delta p_t, F_i^q(t) - \frac{\partial E_i}{\partial p_t}\right> - \left<\delta V_t, \frac{\partial E_i}{\partial V_t}\right> \right)\circ \diff S_t^i \\
        &\stackrel{\eqref{eq:generalised-energies-coords}}{=} \int^{t_{\epsilon_{k+1}}}_{t_{\epsilon_{k+1}}} \left(\left<\delta q_t, \frac{\partial L}{\partial q_t} - F_0^p(t) \right> + \left<\delta p_t, F_0^q(t) - V_t\right> - \left<\delta V_t, p_t - \frac{\partial L}{\partial V_t}\right> \right)\circ \diff S_t^0 \\
        &\quad + \sum_{i \geq 1} \int^{t_{\epsilon_{k+1}}}_{t_{\epsilon_{k+1}}} \left(\left<\delta q_t, \frac{\partial \Gamma_i}{\partial q_t} - \frac{\partial}{\partial q_t} \left<p_t, \Xi_i(q_t)\right> - F_i^p(t)\right> + \left<\delta p_t, F_i^q(t) - \Xi_i(q_t)\right>\right)\circ \diff S_t^i \\
        &= \int^{t_{\epsilon_{k+1}}}_{t_{\epsilon_{k+1}}} \left<\delta q_t, \frac{\partial L}{\partial q_t} \circ \diff S_t^0 + \sum_{i\geq 1} \left(\frac{\partial \Gamma_i}{\partial q_t} - \frac{\partial}{\partial q_t} \left<p_t, \Xi_i(q_t)\right>\right) \circ \diff S_t^i - \circ \diff p_t\right> \\
        &\qquad + \left<\delta p_t, \circ \diff q_t - V_t \circ \diff S_t^0 - \sum_{i \geq 1} \Xi_i(q_t) \circ \diff S_t^i\right> + \left<\delta V_t, p_t - \frac{\partial L}{\partial V_t}\right> \circ \diff S_t^0 \,.
    \end{align*}
    Furthermore, noting that
    \begin{align*}
        \Theta_{T^*Q} \left(\mathrm{pr}_{T^*Q}(Z_t)\right) = p_t \,\dd q_t\,, \qquad T\mathrm{pr}_{T^*Q} \cdot \delta Z_t = \delta q_t \frac{\partial}{\partial q_t} + \delta p_t \frac{\partial}{\partial p_t} \,,
    \end{align*}
    we have
    \begin{align*}
        J_k := \left<\Theta_{T^*Q} \left(\mathrm{pr}_{T^*Q}(Z_t)\right), \,T\mathrm{pr}_{T^*Q} \cdot \delta Z_t\right>_{T^*T^*Q \times TT^*Q} \Big|^{t_{\epsilon_{k+1}}}_{t_{\epsilon_{k}}} = p_{t_{\epsilon_{k+1}}} \,\delta q_{t_{\epsilon_{k+1}}} - p_{t_{\epsilon_{k}}} \,\delta q_{t_{\epsilon_{k}}} \,.
    \end{align*}
    Hence, we see that the local expression for $I_k + J_k$ agrees with the expression \eqref{eq:local-chart-variations-HP} for the variation of the action on a local chart \eqref{eq:HP-action-coords}, computed in the proof of Theorem \ref{thm:stochastic_HP}. Now, we can sum these up to get $\eqref{eq:intrinsic-variations-HP} = \sum_{k=0}^{N-1} (I_k + J_k)$, which verifies the expression \eqref{eq:intrinsic-variations-HP}. Finally, using the endpoint conditions $\delta q_{t_{0}} = \delta q_{t_{1}} = 0$ and the fundamental lemma of the stochastic calculus of variations (Lemma \ref{lemma:fundamental}), we can deduce the relation \eqref{eq:intrinisic-stoch-EL-2}.
\end{proof}
    
\section{Auxiliary results and proofs} \label{app:auxiliary}

In this appendix, we provide proofs for some of the auxiliary results stated in the main body. We restate the statement of each result for completeness.

\subsection{Proof of Lemma \ref{lemma:variation_g_evolution}}
\label{app:variation_g_evolution}

\begin{namedthm*}{Lemma \ref{lemma:variation_g_evolution}}
    Let $g_t\in G$ be a group valued process which is compatible with a driving semimartingale, $S_t$, so that there exists a family of $\mathfrak{g}$-valued semimartingales $\{w_t^i\}_{i \geq 0}$ satisfying
    \begin{equation}
        T_{g_t}L_{{g_t}^{-1}}(\circ \diff g_t) = \sum_{i\geq 0}w_t^i\circ \diff S_t^i \,
    \end{equation}
    (to see this, take $w_t^i := T_{g_t}L_{{g_t}^{-1}}(F_t^i)$ in \eqref{eq:define-g-inv-dg}).
    Then, for the group-valued perturbation $e_{\epsilon,t}$ introduced above, the process $g_{\epsilon,t} = g_t\cdot e_{\epsilon,t}$ is compatible with the same driving semimartingale and evolves according to
    \begin{equation}
        \diff g_{\epsilon,t} = T_eL_{g_{\epsilon,t}}\left(\Ad_{e_{\epsilon,t}^{-1}}w_t^0 + \epsilon \dot{\eta}_t \right) \, \diff t + \sum_{i\geq 1}T_eL_{g_{\epsilon,t}}\left(\Ad_{e_{\epsilon,t}^{-1}}w_t^i \right)\circ \diff S_t^i \,, 
    \end{equation}
    where $\Ad$ denotes the adjoint representation of $G$. 
\end{namedthm*}

\begin{proof}
    By applying the Stratonovich product rule and the compatibility of the process $g_t$ with the driving semimartingale $S_t$, we have
    \begin{align*}
        \diff g_{\epsilon,t} &= T_{g_t}R_{e_{\epsilon,t}} \circ \diff g_t + T_{e_{\epsilon,t}}L_{g_t} \dot{e}_{\epsilon,t} \,\diff t
        \\
        &= \left(T_{g_t}R_{e_{\epsilon,t}}T_eL_{g_t} w_t^0 + \epsilon T_{e_{\epsilon,t}}L_{g_t}T_eL_{e_{\epsilon,t}}\dot{\eta}_t \right) \, \diff t + \sum_{i\geq 1} T_{g_t}R_{e_{\epsilon,t}}T_eL_{g_t} w_t^i \circ \diff S_t^i
        \\
        &=T_eL_{g_{\epsilon,t}}\left(T_{e_{\epsilon,t}^{-1}}R_{e_{\epsilon,t}} T_eL_{e_{\epsilon,t}^{-1}}w_t^0 + \epsilon\dot{\eta}_t \right) \, \diff t + \sum_{i\geq 1} T_eL_{g_{\epsilon,t}}\left( T_{e_{\epsilon,t}^{-1}}R_{e_{\epsilon,t}} T_eL_{e_{\epsilon,t}^{-1}}w_t^i \right) \circ \diff S_t^i
        \\
        &= T_eL_{g_{\epsilon,t}}\left(\Ad_{e_{\epsilon,t}^{-1}}w_t^0 + \epsilon \dot{\eta}_t \right) \, \diff t + \sum_{i\geq 1}T_eL_{g_{\epsilon,t}}\left(\Ad_{e_{\epsilon,t}^{-1}}w_t^i \right)\circ \diff S_t^i
        \,,
    \end{align*}
    as required.
\end{proof}

\subsection{Proof of Theorem \ref{prop:stoch-symplectic-reduction}} \label{app:stoch-symplectic-reduction}

We recall that for a symplectic manifold $(P, \omega)$ and a group $G$ acting on it, $J : P \rightarrow \mathfrak{g}^*$ denotes the corresponding momentum map, $\pi_\mu : J^{-1}(\mu) \rightarrow J^{-1}(\mu)/G_\mu$ denotes the projection with respect to the isotropy subgroup $G_\mu \leq G$, and $\iota_\mu: J^{-1}(\mu) \xhookrightarrow{} P$ denotes the natural inclusion. We also define the induced symplectic form on the reduced space $P_\mu$ by $\iota_\mu^* \omega = \pi_\mu^*\omega_\mu$. With these notations, we recall the statement of Theorem \ref{prop:stoch-symplectic-reduction}.

\begin{namedthm*}{Theorem \ref{prop:stoch-symplectic-reduction}}
    Consider the stochastic Hamiltonian system \eqref{eq:general-symplectic-diffusion} on the symplectic manifold $(P, \omega)$ such that all $\{H_i\}_{i=0}^N$ are $G$-invariant.
    Further, for some $T>0$, let $\{\Phi_t\}_{t \in [0, T]}$ be the stochastic flow of \eqref{eq:general-symplectic-diffusion} on $P$ and $\{\phi_t^\mu\}_{t \in [0, T]}$ be the stochastic flow of \eqref{eq:reduced-symplectic-diffusion} on $P_\mu$. Then, there exists a flow $\psi_t$ on $J^{-1}(\mu)$ such that $\Phi_t \circ \iota_\mu = \iota_\mu \circ \psi_t$ and $\phi_t^\mu \circ \pi_\mu = \pi_\mu \circ \psi_t$ for all $t \in [0, T]$.
\end{namedthm*}

To show this, we first prove the following identity.

\begin{lemma}\label{lemma:reduced-vector-field}
    Consider the setting in Proposition \ref{prop:stoch-symplectic-reduction} and define the reduced Hamiltonians $\{h_i\}_{i=0}^N$ by $\pi_\mu^* h_i = \iota_\mu^* H_i$. Then the following identity holds:
    \begin{align}
        (\pi_\mu)_* X_{H_i}|_{J^{-1}(\mu)} = X^\mu_{h_i} \,,
    \end{align}
    where we denoted by $X_{H_i} \in \mathfrak{X}_{ham}(P)$ the Hamiltonian vector field on $P$ with respect to the Hamiltonian $H_i$, and $X^\mu_{h_i} \in \mathfrak{X}_{ham}(P_\mu)$ denotes the Hamiltonian vector field on $P_\mu$ with respect to the reduced Hamiltonian $h_i$.
    We used the notation $X_{H_i}|_{J^{-1}(\mu)}$ to denote the restriction of $X_{H_i}$ on $J^{-1}(\mu) \subseteq P$.
\end{lemma}
\begin{proof}[Proof of Lemma \ref{lemma:reduced-vector-field}]
Since $J^{-1}(\mu)$ is an invariant manifold under the flow of $X_{H_i}$ for any $i = 0, \ldots, N$, the restriction  $X_{H_i}|_{J^{-1}(\mu)}$ of $X_{H_i}$ on $J^{-1}(\mu)$ is well-defined, and is given by
\begin{align} \label{eq:vector-field-restriction}
(\iota_\mu)_* X_{H_i}|_{J^{-1}(\mu)}(\iota_\mu(p)) = X_{H_i}(\iota_\mu(p)) \,,
\end{align}
for all $p \in J^{-1}(\mu)$.
By the definition of Hamiltonian vector fields, we have
\begin{align}
    &\quad X_{H_i} \intprod \omega = \dd H_i \nonumber \\
    &\Rightarrow \iota_\mu^*(X_{H_i} \intprod \omega) = \iota_\mu^* \dd H_i \nonumber \\
    &\stackrel{\eqref{eq:vector-field-restriction}}{\Rightarrow} \iota_\mu^*((\iota_\mu)_* X_{H_i}|_{J^{-1}(\mu)} \intprod \omega) = \iota_\mu^* \dd H_i \nonumber \\
    &\Rightarrow X_{H_i}|_{J^{-1}(\mu)} \intprod \iota_\mu^* \omega = \dd (\iota_\mu^* H_i) \nonumber \\
    &\stackrel{\eqref{eq:reduced-symplectic-form}}{\Rightarrow} X_{H_i}|_{J^{-1}(\mu)} \intprod \pi_\mu^* \omega_\mu = \dd (\pi_\mu^* h_i) \nonumber \\
    &\Rightarrow \pi_\mu^* ((\pi_\mu)_* X_{H_i}|_{J^{-1}(\mu)} \intprod \omega_\mu) = \pi_\mu^* (\dd h_i) \nonumber \\
    &\Rightarrow (\pi_\mu)_* X_{H_i}|_{J^{-1}(\mu)} \intprod \omega_\mu = \dd h_i \,, \label{eq:reduced-vector-field}
\end{align}
where we used the commutativity of the pull-back with the exterior derivative in the fourth line, and in the last line we used that $\pi_\mu^*$ is injective since $\pi_\mu$ is a submersion. In particular, \eqref{eq:reduced-vector-field} implies that $(\pi_\mu)_* X_{H_i}|_{J^{-1}(\mu)}$ is equivalent to the Hamiltonian vector field $X^\mu_{h_i}$, by definition.
\end{proof}

We now prove our main result.

\begin{proof}[Proof of Theorem \ref{prop:stoch-symplectic-reduction}]
Let $Z_0 \in P$ be such that $J(Z_0) = \mu$ and let $Z_0^\mu$ be the corresponding element in the submanifold $J^{-1}(\mu)$ such that $\iota_\mu(Z_0^\mu) = Z_0$. By Theorem \ref{prop:stoch-noether}, we have $J(\Phi_t(Z_0)) = \mu$ for all $t \in [0, T]$, hence there exists $Z_t^\mu \in J^{-1}(\mu)$ such that $\Phi_t(Z_0) = \iota_\mu(Z_t^\mu)$. Thus, defining $\psi_t(Z_0^\mu) := Z_t^\mu$ for all $t \in [0, T]$, we trivially have $\Phi_t \circ \iota_\mu = \iota_\mu \circ \psi_t$.

Next, we show that $\phi_t^\mu \circ \pi_\mu = \pi_\mu \circ \psi_t$. Restricted to $J^{-1}(\mu) \subseteq P$, the system \eqref{eq:general-symplectic-diffusion} can be written equivalently as
\begin{align*}
    \diff Z_t^\mu = \sum_{i=0}^N X_{H_i}|_{J^{-1}(\mu)}(Z_t^\mu) \circ \diff S_t^i \,.
\end{align*}
Now, by the Stratonovich chain rule \eqref{eq:strat-chain-rule}, we have
\begin{align*}
    \diff \pi_\mu(Z_t^\mu) &= \sum_{i=0}^N (\pi_\mu)_* X_{H_0}|_{J^{-1}(\mu)}(\pi_\mu(Z_t^\mu)) \circ \diff S_t^i = \sum_{i=0}^N X_{h_i}^\mu(\pi_\mu(Z_t^\mu)) \circ \diff S_t^i \,,
\end{align*}
where we used that $(\pi_\mu)_* X_{H_i}|_{J^{-1}(\mu)} = X^\mu_{h_i}$ (Lemma \ref{lemma:reduced-vector-field}). Hence, we see from \eqref{eq:reduced-symplectic-diffusion} that $\phi^\mu_t(\pi_\mu(Z_0^\mu)) = \pi_\mu(Z_t^\mu)$, and therefore $\phi_t^\mu \circ \pi_\mu = \pi_\mu \circ \psi_t$ since $Z_0^\mu$ was chosen arbitrarily. 
\end{proof}

\subsection{Proof of Proposition \ref{prop:max-entropy-principle}} \label{app:max-entropy-principle}

Recall that the Gibbs measure on a symplectic manifold $(P, \omega)$ is defined as the measure
\begin{align}\label{eq:gibbs-meas-app}
\mathbb P_\infty = Z^{-1} e^{-\beta H_0} |\,\omega^n|\,, \quad Z = \int_P e^{-\beta H_0} \mathrm{d} |\,\omega^n| \,.
\end{align}
This can be derived as the measure with the largest entropy among all probability measures that are absolutely continuous with respect to $|\omega^n|$ on $P$ with fixed average energy.

\begin{namedthm*}{Proposition \ref{prop:max-entropy-principle}}
    The Gibbs measure \eqref{eq:gibbs-meas-app} is a solution to the following constrained variational principle
    \begin{align}\label{eq:max-entropy-principle-app}
        \mathbb{P}_\infty =
        \argmax_{\mu >\!\!> \lambda} \left\{-H(\mu | \lambda)\right\}\,, \, \text{ such that } 
        \, \int_P H_0 \,\diff \mu = c \, \text{ and }
        \, \int_P \diff \mu = 1 \,.
    \end{align}
    Here, $c < \infty$ is some finite constant and $H(\mu | \lambda) := \int_P \frac{\diff \mu}{\diff \lambda} \log \frac{\diff \mu}{\diff \lambda} \diff \lambda$ is the relative entropy of the measure $\mu$ with respect to a reference measure $\lambda$. In particular, we take $\lambda = |\omega^n|$ as the reference measure on $P$. The first constraint fixes the average energy of the system and the second constraint ensures that $\mu$ is a probability measure.
\end{namedthm*}

\begin{proof}
    Since we are optimising over measures $\mu$ with $\mu >\!\!> \lambda$, we can set $\mu = \phi(x) \lambda$ for some positive function $\phi : P \rightarrow \mathbb{R}$ without loss of generality. Thus, introducing Lagrange multipliers $\beta, \gamma \in \mathbb{R}$, the constrained optimisation problem \eqref{eq:max-entropy-principle-app} can be re-formulated as finding the minima of the action functional
    \begin{align}
        S[\phi, \beta, \gamma] := \int_P \phi(x) \log \phi(x) \diff \lambda(x) + \beta \left(\int_P H_0(x) \phi(x) \,\diff \lambda(x) - c\right) + \gamma \left(\int_P \phi(x) \diff \lambda(x) - 1\right) \,.
    \end{align}
    Taking variations $\delta S$ with respect to $(\phi, \beta, \gamma)$ and setting $\delta S = 0$, we obtain the relations
    \begin{align}
        \delta \phi &: \quad \log \phi(x) + 1 + \beta H_0(x) + \gamma = 0 \label{eq:delta-phi}\\
        \delta \beta &: \quad \int_P H_0(x) \phi(x) \,\diff \lambda(x) = c \label{eq:delta-beta}\\
        \delta \gamma &: \quad \int_P \phi(x) \,\diff \lambda(x) = 1\,.\label{eq:delta-gamma}
    \end{align}
    The first equation \eqref{eq:delta-phi} yields
    \begin{align} \label{eq:phi-sol-app}
        \phi(x) = Z^{-1}\exp(- \beta H_0(x)), \quad \text{where} \quad Z := e^{1 + \gamma} \,.
    \end{align}
    From \eqref{eq:delta-gamma}, we can further deduce that
    \begin{align} \label{eq:norm-const-app}
        Z = \int_P e^{-\beta H_0(x)} \mathrm{d} \lambda(x) \,.
    \end{align}
    Finally, we show that the distribution \eqref{eq:phi-sol-app}--\eqref{eq:norm-const-app} is indeed the maximal entropy measure. Let $\mu_*$ be the probability measure corresponding to the density \eqref{eq:phi-sol-app}--\eqref{eq:norm-const-app}. Then for any $\mu >\!\!> \lambda$ satisfying $\int_P H_0 \,\diff \mu = c$ and $\int_P \diff \mu = 1$, we have
    \begin{align*}
        -H(\mu_* | \lambda) + H(\mu | \lambda) &= -\int_P \frac{\diff \mu_*}{\diff \lambda}(x) \log \frac{\diff \mu_*}{\diff \lambda}(x) \diff \lambda(x) + \int_P \frac{\diff \mu}{\diff \lambda}(x) \log \frac{\diff \mu}{\diff \lambda}(x) \diff \lambda(x) \\
        &\stackrel{\eqref{eq:phi-sol-app}}{=} \int_P (\log Z + \beta H_0(x)) \diff \mu_*(x) + \int_P \log \frac{\diff \mu}{\diff \lambda}(x) \diff \mu(x) \\
        &= \log Z + \beta c + \int_P \log \frac{\diff \mu}{\diff \lambda}(x) \diff \mu(x) \\
        &= \int_P (\log Z + \beta H_0(x)) \diff \mu(x) + \int_P \log \frac{\diff \mu}{\diff \lambda}(x) \diff \mu(x) \\
        &=  -\int_P \log \frac{\diff \mu_*}{\diff \lambda}(x) \diff \mu(x) + \int_P \log \frac{\diff \mu}{\diff \lambda}(x) \diff \mu(x) \\
        &= \int_P \log \frac{\diff \mu}{\diff \mu_*}(x) \diff \mu(x) \\
        &=: H(\mu | \mu_*) \,.
    \end{align*}
    Since the relative entropy is always positive, we have $-H(\mu_* | \lambda) + H(\mu | \lambda) = H(\mu | \mu_*) > 0$, as expected.
\end{proof}

\subsection{Proof of Proposition \ref{prop:symplectic-reduction-stoch-dissipative}} \label{app:proof-symplectic-reduction-stoch-dissipative}

The following gives an analogous statement to Proposition \ref{prop:stoch-symplectic-reduction} in the stochastic-dissipative case.

\begin{namedthm*}{Proposition \ref{prop:symplectic-reduction-stoch-dissipative}}
Consider the symplectic Langevin diffusion \eqref{eq:symplectic-langevin-eq} on the symplectic manifold $(P, \omega)$ such that all $\{H_i\}_{i=0}^N$ are $G$-invariant. Then for $Z_0 \in J^{-1}(\mu) \subseteq P$, if $Z_t$ solves the symplectic Langevin system \eqref{eq:symplectic-langevin-eq},  then $z_t^\mu := \pi_\mu(Z_t) \in P_\mu$ solves
\begin{align}\label{eq:reduced-symplectic-langevin-app}
\diff z_t^\mu = X_{h_0}^\mu(z_t^\mu) \,\diff t -\frac{\beta}{2}\sum_{i=1}^N\{h_0, h_i\}_{\mu}(z_t^\mu) \, X_{h_i}^\mu(z_t^\mu) \,\diff t  + \sum_{i=1}^N X_{h_i}^\mu(z_t^\mu)\circ \diff W_t^i \,.
\end{align}
\end{namedthm*}
\begin{proof}
    Following the proof in Proposition \ref{prop:stoch-symplectic-reduction}, we first show that $Z_t \in J^{-1}(\mu)$ for all $t$, which follows from Noether's theorem (Lemma \ref{lemma:noether-stoch-diss}), and then use the Stratonovich chain rule \eqref{eq:strat-chain-rule}
    \begin{align}
    \begin{split}
        \diff \pi_\mu(Z_t) &= (\pi_\mu)_*X_{H_0}|_{J^{-1}(\mu)}(\pi_\mu(Z_t)) \,\diff t + \sum_{i=1}^N (\pi_\mu)_*X_{H_i}|_{J^{-1}(\mu)}(\pi_\mu(Z_t)) \circ \diff W_t^i \\
        &\quad -\frac{\beta}{2}\sum_{i=1}^N (\pi_\mu)_* \big(\{H_0, H_i\}X_{H_i}\big)|_{J^{-1}(\mu)}(\pi_\mu(Z_t))\,\diff t \,,
    \end{split}\label{eq:reduced-langevin-computation}
    \end{align}
    to show that this is equivalent to \eqref{eq:reduced-symplectic-langevin}. To verify the last statement, we know from Lemma \ref{lemma:reduced-vector-field} that $(\pi_\mu)_*X_{H_i}|_{J^{-1}(\mu)}  = X^\mu_{h_i}$ for all $i=0, \ldots, N$. We also have
    \begin{align*}
        \{H_0, H_i\}|_{J^{-1}(\mu)}(Z_t) &= (\iota_\mu^* \omega)(X_{H_0}|_{J^{-1}(\mu)}, X_{H_i}|_{J^{-1}(\mu)})(Z_t) \\
        &= (\pi_\mu^* \omega_\mu) (X_{H_0}|_{J^{-1}(\mu)}, X_{H_i}|_{J^{-1}(\mu)})(Z_t) \\
        &= \omega_\mu \big((\pi_\mu)_* X_{H_0}|_{J^{-1}(\mu)}, (\pi_\mu)_* X_{H_i}|_{J^{-1}(\mu)}\big)(\pi_\mu(Z_t)) \\
        &= \omega_\mu \big(X_{h_0}^\mu, X_{h_i}^\mu\big)(\pi_\mu(Z_t)) \\
        &= \{h_0, h_i\}_{\mu}(\pi_\mu(Z_t))\,,
    \end{align*}
    giving us
    \begin{align*}
        (\pi_\mu)_* \big(\{H_0, H_i\}X_{H_i}\big)|_{J^{-1}(\mu)}(\pi_\mu(Z_t)) = \{h_0, h_i\}_{\mu}(\pi_\mu(Z_t)) X^\mu_{h_i}(\pi_\mu(Z_t)) \,.
    \end{align*}
    Putting this together, we see that \eqref{eq:reduced-langevin-computation} is equivalent to \eqref{eq:reduced-symplectic-langevin-app}.
\end{proof}

\subsection{Proof of Proposition \ref{prop:vector-field-orthonormality}}\label{app:proof-of-orthonormal-condition}
Denote by $\gamma : TM \times TM \rightarrow \mathbb{R}$ the Riemannian metric on $M$ and define the isomorphism $\sharp : T^*M \rightarrow TM$ by $\gamma(\alpha^\sharp, v) = \left<\alpha, v\right>_{T^*M \times TM}$ for any $\alpha \in \Omega^1(M)$ and $v \in \mathfrak{X}(M)$. We define the induced metric $\tilde{\gamma} : T^*M \times T^*M \rightarrow \mathbb{R}$ on the covectors by $\tilde{\gamma}(\alpha, \beta) = \gamma(\alpha^\sharp, \beta^\sharp)$.
We can also naturally lift the metric $\tilde{\gamma}$ to higher-order covectors, which we will denote by $\tilde{\gamma}^{(k)} : \bigwedge^kT^*M \times \bigwedge^kT^*M \rightarrow \mathbb{R}$. For example, on $k=2$, this has the coordinate expression
\begin{align}
    \tilde{\gamma}^{(2)}(\alpha, \beta) = \gamma^{ik} \gamma^{jl} \alpha_{ij} \beta_{kl}, \qquad \alpha, \beta \in T^*M \wedge T^*M \,,
\end{align}
where $\tilde{\gamma} = \gamma^{ij} \frac{\partial}{\partial x^i} \otimes \frac{\partial}{\partial x^j}$ is the coordinate expression for the cometric $\tilde{\gamma}$. When it is clear from context, we will omit the superscript from $\tilde{\gamma}^{(k)}$ and simply denote it by $\tilde{\gamma}$.

We now prove Proposition \ref{prop:vector-field-orthonormality}, which we restate below for convenience.

\begin{namedthm*}{Proposition \ref{prop:vector-field-orthonormality}}
    Let $\{\psi_i\}_{i \in \mathbb{Z}_+}$ be a set of exact two-forms that is orthonormal with respect to the inner-product $\tilde{\gamma}^n(\alpha, \beta) := \int_{M} \tilde{\gamma}(\alpha, \Delta^{n} \beta) \text{vol}$. The vector fields $\{\nabla^\perp \psi_i\}_{i \in \mathbb{Z}_+}$ in the Lie algebra $\mathfrak{X}_{\mathrm{vol}}(M)$ are then orthonormal with respect to the inner product $\gamma^{n-1} : \mathfrak{X}_{\mathrm{vol}}(M) \times \mathfrak{X}_{\mathrm{vol}}(M) \rightarrow \mathbb{R}$, defined by
    \begin{align}
        \gamma^{n-1}(u, v) := \int_M \gamma(u, (\Delta^\sharp)^{n-1} v) \mathrm{vol}\,,
    \end{align}
    where $\Delta^\sharp : \mathfrak{X}_{\mathrm{vol}}(M) \rightarrow \mathfrak{X}_{\mathrm{vol}}(M)$ is defined by $\left<\Delta \omega, u\right>_{T^*M \times TM} = \left<\omega, \Delta^\sharp u\right>_{T^*M \times TM}$.
\end{namedthm*}
\begin{proof}
    First we note that the Laplace-deRham operator is defined by
    \begin{align}
        \Delta = \dd \vec{\delta} + \vec{\delta}\dd\,,
    \end{align}
    where $\vec{\delta} : \Omega^2(M) \rightarrow \Omega^1(M)$ is the codifferential operator, defined by $\tilde{\gamma}(\dd \alpha, \beta) = \tilde{\gamma}(\alpha, \vec{\delta} \beta )$ for any $\alpha \in \Omega^1(M)$ and $\beta \in \Omega^2(M)$.
    For any $\phi, \psi \in \dd \Omega^1(M)$, since $\dd \phi = \dd \psi = 0$, we have
    \begin{align*}
        \Delta \phi &= (\dd \vec{\delta} + \vec{\delta}\dd) \phi = \dd \vec{\delta} \phi, \\
        \Delta^2 \phi &= (\dd \vec{\delta} + \vec{\delta}\dd) \Delta \phi = (\dd \vec{\delta} + \vec{\delta}\dd) \dd \vec{\delta} \phi = (\dd \vec{\delta})^2 \phi \\
        &\vdots \\
        \Delta^n \phi &= (\dd \vec{\delta} + \vec{\delta}\dd) \Delta^{n-1} \phi = \cdots = (\dd \vec{\delta})^n \phi \,,
    \end{align*}
    and likewise, $\Delta^n \psi = (\dd \vec{\delta})^n \psi$. Similarly, we can show that $\Delta^{n-1} \vec{\delta}\psi = (\vec{\delta}\dd)^{n-1} \vec{\delta}\psi$. Thus, we have
    \begin{align}
        \tilde{\gamma}^n(\phi, \psi) &= \int_{M} \tilde{\gamma}(\phi, \Delta^{n} \psi) \text{vol} = \int_{M} \tilde{\gamma}(\phi, (\dd \vec{\delta})^{n} \psi) \text{vol} \nonumber \\
        &= \int_{M} \tilde{\gamma}(\phi, \dd (\vec{\delta} \dd)^{n-1}\vec{\delta} \psi) \text{vol} = \int_{M} \tilde{\gamma}(\vec{\delta} \phi, (\vec{\delta} \dd)^{n-1}\vec{\delta} \psi) \text{vol} \nonumber \\
        &= \int_{M} \tilde{\gamma}(\vec{\delta} \phi, \Delta^{n-1}\vec{\delta} \psi) \text{vol} \,. \label{eq:gamma-phi-psi-intermediate}
    \end{align}
    Next, given an exact two-form $\psi$, we can define the corresponding vector field by $\nabla^\perp \psi := (\vec{\delta}\psi)^\sharp \in \mathfrak{X}_{\mathrm{vol}}(M)$, where the isomorphism $\sharp : T^*M \rightarrow TM$ is defined with respect to the Riemannian metric $\gamma$. This implies
    \begin{align*}
        \eqref{eq:gamma-phi-psi-intermediate} &= \int_{M} \tilde{\gamma}(\vec{\delta} \phi, \Delta^{n-1}\vec{\delta} \psi) \text{vol} = \int_{M} \left<\Delta^{n-1}\vec{\delta} \psi, (\vec{\delta} \phi)^\sharp\right>_{T^*M \times TM} \text{vol} \nonumber \\
        &= \int_{M} \left<\vec{\delta} \psi, (\Delta^\sharp)^{n-1}(\vec{\delta} \phi)^\sharp\right>_{T^*M \times TM} \text{vol} = \int_{M} \gamma((\vec{\delta} \psi)^\sharp, (\Delta^\sharp)^{n-1}(\vec{\delta} \phi)^\sharp) \text{vol} \nonumber \\
        &= \int_{M} \gamma(\nabla^\perp \psi, (\Delta^\sharp)^{n-1}\nabla^\perp \phi) \text{vol} = \gamma^{n-1}(\nabla^\perp \psi, \nabla^\perp \phi) \,.
    \end{align*}
    Hence, we have
    \begin{align*}
        \gamma^{n-1}(\nabla^\perp \psi_i, \nabla^\perp \psi_j) = \tilde{\gamma}^n(\psi_i, \psi_j) = \delta_{ij} \,,
    \end{align*}
    which proves our claim.
\end{proof}

\section{Symmetry reduction for semidirect product systems with noise and dissipation}\label{app:semidirect}
The Euler-Poincar\'e theorem (Theorem \ref{thm:EP}) can be extended to the case where the symmetry of the Lagrangian is \emph{broken} by its dependence on additional parameters. In addition to our group, $G$, suppose we also have a vector space, $V$, and a left\footnote{When instead considering right actions, the modifications to the theory are analogous to the deterministic case \cite{HMR1998}.} representation of $G$ on $V$. We will denote this representation by $\Phi_g:V\rightarrow V$ for each $g\in G$, and the corresponding dual representation by $\Phi_g^*:V^*\rightarrow V^*$. Note that $\Phi_g^*$ is in fact a \emph{right} action here, and the corresponding \emph{left} dual representation is $\Phi_{g^{-1}}^*$. In what follows, we have a semidirect product structure, $G \ltimes V$, and the equations are equivalent to Lie-Poisson equations on the semidirect product co-algebra, $\mathfrak{g}^* \ltimes V^*$. It should be noted that, the symmetry broken equations are \emph{not} the usual Euler-Poincar\'e equations applied to the group $G \ltimes V$, and are expressed instead on the space $\mathfrak{g}\ltimes V^*$. To go from the Hamiltonian picture to the Lagrangian picture in this case, we perform a Legendre transformation in $\mathfrak{g}^*$ only, and not in the representation space $V^*$. For further discussion of the relationship between symmetry breaking and semidirect product structures, see \cite{HMR1998} or \cite{GBR2009}.

\subsection{Lagrangian formulation}
We first consider a Lagrangian derivation of the system. In this case, we consider a Lagrangian\footnote{Note that the convention for the dependence of the Lagrangian to be on the dual space, $V^*$, is a consequence of the prior development of the deterministic theory in the Hamiltonian description before the Lagrangian.}, $L:TG\times V^*\rightarrow \mathbb{R}$, which is left $G$-invariant. In this section, we will consider left invariant Lagrangians and left representations. As in the deterministic case, it is possible to consider any combination of left/right invariant Lagrangians with a left/right representation of $G$ on $V$. We define a collection of Lagrangians, $L_{a_0}:TG\rightarrow \mathbb{R}$, smoothly parameterised by $a_0 \in V^*$, by $L_{a_0}(\cdot) = L(\cdot,a_0)$. Using the procedure introduced in \cite{HMR1998}, each $L_{a_0}$ is invariant under the lift to $TG$ of the left action of the isotropy group
\begin{equation}
    G_{a_0} = \{ g \in G : \Phi_g^*a_0 = a_0 \} \,,
\end{equation}
rather than the full group. Using the $G$-invariance of the Lagrangian, we may define the reduced Lagrangian by acting from the left with $g^{-1}$ as follows
\begin{equation}
    L(g_t,V_t,a_0) = L(e,T_gL_{g^{-1}}V_t, \Phi_{g}^*a_0) =:\ell(v_t,a_t) \,,
\end{equation}
where $v_t = T_gL_{g^{-1}}V_t$, and $a_t = \Phi_{g}^*a_0$. Note that the dual action of $g\in G$ on $V^*$ is a right action $\Phi^*_g:V^*\rightarrow V^*$, and therefore the required left action by the inverse is $\Phi^*_{(g^{-1})^{-1}} = \Phi^*_g V^*\rightarrow V^*$.

As in the setup of Theorem \ref{thm:stochastic_HP}, we can consider stochastic potentials, $\Gamma_i(g,a_0):G\times V^*\rightarrow \mathbb{R}$, which can also depend on the parameter $a_0 \in V^*$. In Theorem \ref{thm:EP}, these were not included since their lack of dependence on an element of the tangent space would not allow for meaningful reduction. In the semidirect product case, we take each stochastic potential $\Gamma$ to be $G$-invariant and ${\Gamma}_{a_0}(\cdot) = \Gamma(\cdot,a_0)$ to have a broken symmetry analogous to the Lagrangian. This allows us to define a collection of reduced stochastic potentials, $\gamma_i:V^*\rightarrow \mathbb{R}$, as
\begin{equation}
    \Gamma_i(g,a_0) = \Gamma_i(e,\Phi_{g}^*a_0) =: \gamma_i(a_t) \,.
\end{equation}
When taking variations with respect to $a_t$, we obtain variational derivatives defined through the natural pairing between $V$ and $V^*$. In order to vary the action with respect to both $v_t$ and $a_t$ and derive an equation of motion, we need to define a way to transform between the pairing between the vector space and its dual, $V\times V^*$, and that between the Lie algebra and co-algebra, $\mathfrak{g}\times \mathfrak{g}^*$. To do so, first note that our dual representation of $G$ on $V^*$ allows us to define the action of the Lie algebra on $V^*$ as the infinitesimal action of $G$ on $V^*$. That is, the representation $\Phi^*$ can be interpreted as a map $G\times V^* \rightarrow V^*$, and we may consider the tangent at the identity of this map in its Lie group valued argument. Thus the action of $\zeta\in\mathfrak{g}$, on elements of $V^*$, is given by $T_e\Phi^*_{\zeta}:V^* \rightarrow V^*$ \footnote{Note that many authors (see e.g. \cite{HMR1998}) denote this action by minus concatenation $ -\zeta a = T_e\Phi^*_{\zeta}a $ by convention.}.
\begin{definition}
    Given the representation $\Phi^*_g$ of $G$ on $V^*$ and the action of the Lie algebra on $V^*$ as defined above, we define the diamond operator, $\diamond: V \times V^* \rightarrow \mathfrak{g}^*$, by
    \begin{equation}
        \scp{b}{T_e\Phi_{\zeta}^*a}_{V\times V^*} = \scp{\zeta}{b \diamond a}_{\mathfrak{g}\times\mathfrak{g}^*}\,,
    \end{equation}
    where $b \in V$, $a \in V^*$, and $\zeta \in \mathfrak{g}$.
\end{definition}

\begin{lemma}[Constrained variation of $a_t$]\label{lemma:variation_a}
    For $a_t$ defined by $a_t = \Phi_{g}^*a_0$, we have
    \begin{equation}
        \delta a_t = T_e\Phi^*_{\eta}a_t \,,
    \end{equation}
    where $\eta = T_gL_{g^{-1}}\delta g$ is defined by its relationship to the arbitrary variation in $g$.
\end{lemma}
\begin{proof}
    By the definition of the curve $a_t$ in $V^*$, we have
    \begin{equation}
        \delta a_t = \delta (\Phi_g^*) a_0 = \delta \Phi_g^* (\Phi_g^*)^{-1}a_t = (T_g\Phi^*_{\delta g})(\Phi_{g^{-1}}^*) a_t \,,
    \end{equation}
    where we have used the identity $\Phi_{g^{-1}}^* = (\Phi_g^*)^{-1}$.
    Since the dual representation of $G$ on $V^*$ is a \emph{right} action, and the definition of right action is $\Phi^*_g\Phi^*_h = \Phi^*_{hg}$, we make take the following derivative of this identity, $T_g\Phi^*_{V}\Phi^*_h = T_{hg}\Phi^*_{T_gL_hV}$, for $V\in T_gG$. This identity can be applied with $h = g^{-1}$ to give
    \begin{equation}
        \delta a_t = T_e\Phi^*_{T_gL_{g^{-1}}\delta g}\,a_t =: T_e\Phi^*_{\eta}a_t \,,
    \end{equation}
    as required.
\end{proof}
This allows us to state an extension of Theorem \ref{thm:EP} to semidirect product spaces.
\begin{theorem}[The stochastic Euler-Poincar\'e theorem for semidirect product Lie algebras]\label{thm:EP_sdp}
	Let $S_t$ be a driving semimartingale and assume we have a left-invariant Lagrangian, $L:TG\times V^*\rightarrow \mathbb{R}$, a collection of Lie algebra-valued objects, $\sigma_i$, which do not depend on time, and stochastic potentials $\Gamma_i:G\times V^*\rightarrow \mathbb{R}$. Assume also that we have functions $L_{a_0}$, $({\Gamma_i})_{a_0}$, $\ell$, and $\gamma_i$ as in the above preamble. Furthermore, we relate each Lie algebra-valued stochastic perturbation term, $\xi_i$, to the corresponding term in Theorem \ref{thm:stochastic_HP} via group action, as $\Xi_i(g) = T_eL_g\xi_i$. Then the following are equivalent:
	\begin{enumerate}
	 \item The unreduced Hamilton-Pontryagin principle, as stated in Theorem \ref{thm:stochastic_HP}, holds for each $L_{a_0}$ and $({\Gamma_i})_{a_0}$.
	 \item The curve, $(g,V)\in TG$, satisfies the stochastic Euler-Lagrange equations \eqref{eq:stoch-euler-lagrange-mom}--\eqref{eq:stoch-euler-lagrange-pos}, for each $L_{a_0}$ and $({\Gamma_i})_{a_0}$.
	 \item The reduced Hamilton-Pontryagin principle
	 \begin{equation}\label{eqn:reduced_stochastic_HP_action_sdp}
	 	0 = \delta\int_{t_0}^{t_1}\ell(v,a) \,\diff t + \sum_{i\geq 1}\gamma_i(a)\circ \diff S_t^i + \scp{\mu}{T_{g}L_{{g}^{-1}}(\circ\diff g) - v\circ \diff S_t^0 - \sum_{i\geq 1}\xi_i\circ \diff S_t^i} \,,
	 \end{equation}
	 holds for $(g,v,\mu,a) \in G\times \mathfrak{g} \times \mathfrak{g}^* \times V^*$.
	 \item The following stochastic Euler-Poincar\'e equations hold
	 \begin{align}\label{eqn:stochastic_EP_sdp}
	 	\diff \frac{\delta\ell}{\delta v} &= 
        \left(\ad^*_{v}\frac{\delta\ell}{\delta v} + \frac{\delta\ell}{\delta a}\diamond a\right) \diff t + \sum_{i \geq 1} \left( \ad^*_{\xi_i}\frac{\delta\ell}{\delta v} + \frac{\delta\gamma_i}{\delta a}\diamond a\right)\circ \diff S_t^i
        \,,\\
        \diff a &= T_e\Phi^*_{v}\,a \,\diff t + \sum_{i\geq 1} T_e\Phi^*_{\xi_i}\,a\circ \diff S_t^i \label{eqn:stochastic_EP_sdp2}
        \,,
	 \end{align}
     where $\ad^*$ is the dual of $\ad$ with respect to the natural pairing between $\mathfrak{g}$ and its dual space $\mathfrak{g}^*$.
 	\end{enumerate}
\end{theorem}
The proof of this theorem closely follows that of Theorem \ref{thm:EP}. The only additional technicality lies in the form of the equation for the parameter $a$, which follows from Lemma \ref{lemma:variation_a}.

\begin{corollary}
    The stochastic Euler-Poincar\'e equations corresponding to Theorem \ref{thm:EP_sdp} can be expressed as the Lie-Poisson equation
    \begin{equation}
        \diff f(\mu,a) = \sum_i \{ f,h_i \}_{-} \circ \diff S_t^i \,,
    \end{equation}
    where $f$ is a function on the Lie co-algebra and $\{ \cdot , \cdot \}_{-}$ is the $(-)$ Lie-Poisson bracket in the case where there exists a \emph{left} representation of $G$ on $V$.
\end{corollary}
\begin{proof}
We begin by performing a Legendre transform
\begin{align}
	\ell(v,a) &= \scp{\mu}{v} - h(\mu,a)
	\,,\label{semidirect-legendre-transform}\\
	\gamma_i(a) &= \scp{\mu}{\xi_i} - h_i(\mu,a)
	\,,
\end{align}
and deriving the following Lie-Poisson equations
\begin{equation}
    \diff \begin{bmatrix} \mu \\ a \end{bmatrix} = \begin{bmatrix} \ad^*_{\Box}\mu & -\Box \diamond a \\ T_e\Phi^*_{\Box}\,a & 0 \end{bmatrix}\begin{bmatrix} {\delta h}/{\delta \mu} \\ {\delta  h}/{\delta a} \end{bmatrix} \diff t + \sum_{i\geq 1}\begin{bmatrix} \ad^*_{\Box}\mu & -\Box \diamond a \\ T_e\Phi^*_{\Box}\,a & 0 \end{bmatrix}\begin{bmatrix} {\delta h_i}/{\delta \mu} \\ {\delta  h_i}/{\delta a} \end{bmatrix} \circ \diff S_t^i \,,
\end{equation}
which are equivalent to equations \eqref{eqn:stochastic_EP_sdp} and \eqref{eqn:stochastic_EP_sdp2}. This equivalence follows directly from the relationships between the variables induced by the Legendre transformation, that is
\begin{equation}
    v = \frac{\delta h}{\delta \mu} \,, \quad \xi_i = \frac{\delta h_i}{\delta \mu} \,, \quad \mu = \frac{\delta\ell}{\delta v} \,,\quad \frac{\delta\ell}{\delta a} = -\frac{\delta h}{\delta a} \,,\quad\hbox{and}\quad \frac{\delta\gamma_i}{\delta a} = -\frac{\delta h_i}{\delta a} \,.
\end{equation}
Denoting $h=h_0$ for convenience, a function defined on the semidirect product co-algebra evolves according to
\begin{align}
    \diff f(\mu,a) &= -\scp{\left( \frac{\delta f}{\delta \mu} , \frac{\delta f}{\delta a} \right)}{\left( -\sum_i\ad^*_{\delta h_i / \delta \mu}\mu\circ \diff S_t^i + \sum_i \frac{\delta h_i}{\delta a}\diamond a\circ \diff S_t^i \,,\, -\sum_i T_e\Phi^*_{\delta h_i / \delta \mu} a \circ \diff S_t^i \right)}
    \nonumber \\
    &= \scp{\frac{\delta f}{\delta \mu}}{\sum_i\ad^*_{\delta h_i / \delta \mu}\mu\circ \diff S_t^i} - \scp{\frac{\delta f}{\delta \mu}}{\sum_i \frac{\delta h_i}{\delta a}\diamond a\circ \diff S_t^i} + \scp{\frac{\delta f}{\delta a}}{\sum_i T_e\Phi^*_{\delta h_i / \delta \mu} a \circ \diff S_t^i}
    \nonumber \\
    &= \sum_i \left[ \scp{ \ad_{\delta h_i / \delta \mu} \frac{\delta f}{\delta \mu} }{\mu} - \scp{ \frac{\delta h_i}{\delta a} }{T_e\Phi^*_{{\delta f}/{\delta \mu}}a} + \scp{\frac{\delta f}{\delta a}}{ T_e\Phi^*_{\delta h_i / \delta \mu} a} \right] \circ \diff S_t^i
    \nonumber \\
    &= -\sum_i \scp{(\mu,a)}{\left( \left[ \frac{\delta f}{\delta \mu} , \frac{\delta h_i}{\delta \mu} \right] \,,\, T_e\Phi_{{\delta f}/{\delta \mu}}\frac{\delta h_i}{\delta a} - T_e\Phi_{\delta h_i / \delta \mu} \frac{\delta f}{\delta a} \right)} \circ \diff S_t^i \nonumber \\
    &=: \sum_i \{ f , h_i \}_{-} \circ \diff S_t^i \,, \label{eq:semidirect-equation}
\end{align}
where $\{ \cdot, \cdot\}_{-}$ is defined by the final line of the above calculation. 
\end{proof}

We have thus demonstrated that the stochastic Euler-Poincar\'e theorem, formulated through the Hamilton-Pontryagin approach, yields equations which are a natural extension of the Lie-Poisson system.

\subsection{Hamiltonian formulation}

From a Hamiltonian perspective, the derivation of the system simply follows from semi-direct product reduction theory \cite{marsden1984semidirect}. This states that if a Hamiltonian $H_{a_0} : T^*G \rightarrow \mathbb{R}$ for some $a_0 \in V^*$ is invariant under the action of an isotropy subgroup $G_{a_0} = \{g \in G : \Phi_g^* a_0 = a_0\}$ and the {\em extended Hamiltonian} $H : T^*G \times V^* \rightarrow \mathbb{R}$ defined by $H(T_e^*L_g \alpha_g, \Phi^*_g a_0) = H_{a_0}(\alpha_g)$ is invariant under the induced left $G$-action $\widetilde{L}_h : (\alpha_g, a) \mapsto (T_g^*L_{h^{-1}} \alpha_g, \Phi_{h^{-1}}^* a)$ for any $h \in G$, then there is a reduced Hamiltonian system on the dual Lie algebra $\mathfrak{s}^* = \mathfrak{g}^* \ltimes V^*$ of the semi-direct product group $S := G \ltimes V$. Denoting by $J_R : T^*S \rightarrow \mathfrak{s}^*$ the right momentum map $J_R(\alpha_g, a_0) = (T_e^*L_g \alpha_g, \Phi^*_g a_0)$ corresponding to the right action of $S$ on  $T^*S$, the reduced Hamiltonian $h : \mathfrak{s}^*_- \rightarrow \mathbb{R}$ reads $h \circ J_R = H$ and the canonical bracket on $T^*G$ induces a {\em semi-direct product bracket} on $\mathfrak{s}^*_-$, given by
\begin{align} \label{eq:semidirect-bracket}
    \{f, h\}_{\mathfrak{s}^*_-}(\mu, a) = -\left<\mu, \left[\frac{\delta f}{\delta \mu}, \frac{\delta h}{\delta \mu}\right]\right>_{\mathfrak{g}^* \times \mathfrak{g}} - \left<a, T_e \Phi_{{\delta f}/{\delta \mu}} \frac{\delta h}{\delta a}  - T_e \Phi_{{\delta h}/{\delta \mu}} \frac{\delta f}{\delta a} \right>_{V^* \times V},
\end{align}
where $T_e \Phi : \mathfrak{g} \rightarrow \mathfrak{gl}(V)$ is the tangent map of the group representation $\Phi : G \rightarrow GL(V)$ at the identity (i.e., the Lie algebra representation).

We can easily construct a stochastic extension of the system preserving the structure of the deterministic system by considering noise Hamiltonians of the form
\begin{align}
    H_i(\alpha_g, a_0) = \sigma \left<J(\alpha_g, a_0), \left(\xi_i, -\frac{\delta \gamma_i}{\delta a_0}\right)\right>_{\mathfrak{s}^* \times \mathfrak{s}},
\end{align}
for some $\xi_i \in \mathfrak{g}$ and $\gamma_i : V^* \rightarrow \mathbb{R}$, which therefore yields the reduced Hamiltonian
\begin{align}\label{eq:semidirect-noise-hamiltonians}
    h_i(\mu, a) = \sigma \left<\left(\mu, a\right), \left(\xi_i, -\frac{\delta \gamma_i}{\delta a_0}\right)\right>_{\mathfrak{s}^* \times \mathfrak{s}}.
\end{align}
Plugging these expressions into \eqref{eq:stochastic-poisson-system} yields \eqref{eq:semidirect-equation} and in particular, the dynamics on the variables $(\mu, a)$ read
\begin{align}
    \diff \mu &= \left(\ad^*_{\delta h / \delta \mu} \mu - \frac{\delta h}{\delta a} \diamond a \right) \diff t + \sigma \sum_{i \geq 1} \left(\ad^*_{\xi_i} \mu + \frac{\delta \gamma_i}{\delta a} \diamond a \right) \circ \diff S_t^i, \\
    \diff a &= T_e\Phi^*_{\delta h / \delta \mu} a \,\diff t + \sigma \sum_{i 
    \geq 1} T_e \Phi_{\xi_i}^* a \circ \diff S_t^i,
\end{align}
which we see is equivalent to \eqref{eqn:stochastic_EP_sdp}--\eqref{eqn:stochastic_EP_sdp2} by identifying
\begin{align*}
    \frac{\delta h}{\delta \mu} = v, \quad \frac{\delta \ell}{\delta v} = \mu, \quad \frac{\delta h}{\delta a} = - \frac{\delta \ell}{\delta a},
\end{align*}
which follows from the Legendre transform \eqref{semidirect-legendre-transform}.

\subsection{Semi-direct product systems with noise and dissipation}\label{app:semidirect-noise-dissipation}
Likewise, we can derive a system with structure-preserving dissipation via the framework in Section \ref{sec:double-bracket-dissipation}, since semi-direct product reduction can simply viewed as Lie-Poisson reduction on the extended space $T^*S$ with symmetry group $S$. Plugging our expressions \eqref{eq:semidirect-bracket}--\eqref{eq:semidirect-noise-hamiltonians} into \eqref{eq:double-symmetric-bracket}--\eqref{eq:lie-poisson-langevin-bracket-form} yields the stochastic-dissipative system
\begin{align}
    \begin{split}
    \diff \mu &= \left(\ad^*_{\delta h / \delta \mu} \mu - \frac{\delta h}{\delta a} \diamond a \right) \diff t + \sigma \sum_{i \geq 1} \left(\ad^*_{\xi_i} \mu + \frac{\delta \gamma_i}{\delta a} \diamond a \right) \circ \diff W_t^i \\
    &\quad + \theta \left(\ad^*_{\left(\ad^*_{\delta h/\delta \mu} \mu - (\delta h/\delta a) \diamond a\right)^\sharp} \mu - \left(T_e\Phi^*_{\delta h / \delta \mu} a\right)^\sharp \diamond a\right) \diff t \,,
    \end{split} \label{eq:semidirect-mu-eq}\\
    \diff a &= T_e\Phi^*_{\delta h / \delta \mu} a \,\diff t + \sigma \sum_{i 
    \geq 1} T_e \Phi_{\xi_i}^* a \circ \diff W_t^i + \theta T_e\Phi^*_{\left(\ad^*_{\delta h/\delta \mu} \mu - (\delta h/\delta a) \diamond a\right)^\sharp} \cdot a \, \diff t \,, \label{eq:semidirect-a-eq}
\end{align}
where $\sharp$ is the musical isomorphism defined by assigning inner products on the spaces $\mathfrak{g}$ and $V$, which is a modelling choice (note that the vectors $\{(\xi_i, \gamma'_i(a))\}_i$ must be orthonormal under this choice of inner product).
As in the standard Lie-Poisson case, this system preserves the Gibbs-measure $\mathbb{P}_\infty^{(\mu_0, a_0)}$ on the coadjoint orbit $\mathcal{O}_{(\mu_0, a_0)} \subset \mathfrak{s}^*$ (Corollary \ref{cor:lie-poisson-gibbs-measure}).

\subsection{Example: The stochastic-dissipative heavy top}

We consider a prototypical example of a semi-direct product system, namely the heavy top dynamics \cite{marsden1984semidirect}. The heavy top system has a Lie-Poisson structure on the semi-direct product group $S = SO(3) \ltimes (\mathbb{R}^3)^*$ and the its dynamics is generated by the Hamiltonian
\begin{align}
    H_{\boldsymbol{a}_0}(\alpha_{\boldsymbol{A}}) = \frac12 \|T_e^* L_{\boldsymbol{A}}\cdot\alpha_{\boldsymbol{A}}\|_{\mathbb{I}^{-1}}^2 + Mgl \boldsymbol{A}^{-1} \boldsymbol{a}_0 \cdot \boldsymbol{\chi} \,, \label{eq:heavy-top-hamiltonian}
\end{align}
for $\boldsymbol{a_0} \in (\mathbb{R}^3)^*$, $\alpha_{\boldsymbol{A}} \in T_{\boldsymbol{A}} SO(3)$ and $\|\cdot\|_{\mathbb{I}^{-1}}$ is the norm on $\mathfrak{so}(3) \cong \mathbb{R}^3$ induced by the inner product $\left<\xi, \eta \right>_{\mathbb{I}^{-1}} = \boldsymbol{\xi}^\top \mathbb{I}^{-1} \boldsymbol{\eta}$. Here, the constants $M, g, l$ are the mass of the body, gravitational acceleration and distance from the fixed point to the centre of mass, respectively. The constant vector $\boldsymbol{\chi} \in \mathbb{R}^3$ is the direction of the line connecting the fixed point to the centre of mass of the body. Note that the Hamiltonian \eqref{eq:heavy-top-hamiltonian} is invariant under the isotropy subgroup of $SO(3)$ that fixes the vector $\boldsymbol{a}_0$ and moreover can be expressed easily in terms of the extended Hamiltonian
\begin{align}
H(T_e^* L_{\boldsymbol{A}}\cdot\alpha_{\boldsymbol{A}}, \boldsymbol{A}^{-1} \boldsymbol{a}_0) = H_{\boldsymbol{a}_0}(\alpha_{\boldsymbol{A}}) \,.
\end{align}
Thus, semi-direct product reduction can take effect and this yields a Lie-Poisson system on $\mathfrak{s}^*$ equipped with the bracket
\begin{align}
    \{f, g\}(\boldsymbol{\Pi}, \boldsymbol{\Gamma}) = - \boldsymbol{\Pi} \cdot (\nabla_{\boldsymbol{\Pi}} f \times \nabla_{\boldsymbol{\Pi}} \, g) - \boldsymbol{\Gamma} \cdot (\nabla_{\boldsymbol{\Pi}} f \times \nabla_{\boldsymbol{\Gamma}} \, g + \nabla_{\boldsymbol{\Gamma}} f \times \nabla_{\boldsymbol{\Pi}} \,g) \,,
\end{align}
and with the reduced Hamiltonian
\begin{align}
    h(\boldsymbol{\Pi}, \boldsymbol{\Gamma}) = \frac12 \|\boldsymbol{\Pi}\|_{\mathbb{I}^{-1}}^2 + Mgl \boldsymbol{\Gamma} \cdot \boldsymbol{\chi} \,,
\end{align}
for $(\boldsymbol{\Pi}, \boldsymbol{\Gamma}) \in \mathfrak{so}^*(3) \times (\mathbb{R}^3)^* \cong \mathbb{R}^3 \times \mathbb{R}^3$.

Choosing our noise Hamiltonians to be of the form
\begin{align*}
    H_i(\alpha_{\boldsymbol{A}}, \boldsymbol{a}_0) &= \sigma \left(\left<T_e^* L_{\boldsymbol{A}}\cdot\alpha_{\boldsymbol{A}}, \xi_i\right>_{\mathfrak{so}^*(3) \times \mathfrak{so}(3)} + \left<\boldsymbol{A}^{-1} \boldsymbol{a}_0, \nabla_{\boldsymbol{a}_0} \gamma_i(\boldsymbol{a}_0) \right>_{\mathbb{R}^3 \times \mathbb{R}^3}\right) \\
    &= \sigma \left(\left<\boldsymbol{\Pi}, \xi_i\right>_{\mathfrak{so}^*(3) \times \mathfrak{so}(3)} + \left<\boldsymbol{\Gamma}, \nabla_{\boldsymbol{a}_0} \gamma_i(\boldsymbol{a}_0) \right>_{\mathbb{R}^3 \times \mathbb{R}^3}\right) \\
    &= h_i(\boldsymbol{\Pi}, \boldsymbol{\Gamma}) \,,
\end{align*}
for $i = 1, 2, 3$ for some $\sigma > 0$, $\xi_i \in \mathfrak{so}(3)$ and $\gamma_i : \mathbb{R}^3 \rightarrow \mathbb{R}$, we obtain from \eqref{eq:semidirect-mu-eq}--\eqref{eq:semidirect-a-eq} the corresponding stochastic-dissipative extension to the heavy top system:
\begin{align}
    \begin{split}
    \diff \boldsymbol{\Pi} &= \left(\boldsymbol{\Pi} \times \mathbb{I}^{-1} \boldsymbol{\Pi} + Mgl \boldsymbol{\Gamma} \times \boldsymbol{\chi}\right) \diff t + \sigma \sum_{i=1}^3 \left(\boldsymbol{\Pi} \times \boldsymbol{\xi}_i + \boldsymbol{\Gamma} \times \nabla_{\boldsymbol{a_0}} \gamma_i(\boldsymbol{a_0})\right) \circ \diff W_t^i \\
    &\quad - \theta \left(\boldsymbol{\Pi} \times \left(\boldsymbol{\Pi} \times \mathbb{I}^{-1} \boldsymbol{\Pi} + Mgl \boldsymbol{\Gamma} \times \boldsymbol{\chi}\right) + \boldsymbol{\Gamma} \times \left(\boldsymbol{\Gamma} \times \mathbb{I}^{-1} \boldsymbol{\Pi}\right)\right) \diff t \,,
    \end{split} \\
    \diff \boldsymbol{\Gamma} &= \boldsymbol{\Gamma} \times \mathbb{I}^{-1} \boldsymbol{\Pi} \,\diff t - \theta\boldsymbol{\Gamma} \times \left(\boldsymbol{\Pi} \times \mathbb{I}^{-1} \boldsymbol{\Pi} + Mgl \boldsymbol{\Gamma} \times \boldsymbol{\chi}\right)\diff t + \sigma \sum_{i=1}^3 \boldsymbol{\Gamma} \times \boldsymbol{\xi}_i \circ \diff W_t^i \,.
\end{align}
The coadjoint orbit on $\mathfrak{so}^*(3) \ltimes (\mathbb{R}^3)^*$ with $\mathbf{\Gamma}_0 \neq 0$ is given by the four dimensional submanifold 
\begin{align}
    \mathcal{O}_{(\boldsymbol{\Pi}_0, \boldsymbol{\Gamma}_0)} = \{(\boldsymbol{\Pi}, \boldsymbol{\Gamma}) \in \mathbb{R}^3 \times \mathbb{R}^3 : \|\boldsymbol{\Gamma}\|^2 = \|\boldsymbol{\Gamma}_0\|^2, \text{ and } \boldsymbol{\Pi} \cdot \boldsymbol{\Gamma} = \boldsymbol{\Pi}_0 \cdot \boldsymbol{\Gamma}_0\} \cong TS^2_{\|\mathbf{\Gamma}_0\|} \,,
\end{align}
and the corresponding KKS symplectic form by
\begin{align}
    \omega^{\text{KKS}}_{(\boldsymbol{\Pi}_0, \boldsymbol{\Gamma}_0)}\Big(\ad^*_{(\boldsymbol{\xi}_1, \boldsymbol{v}_1)} (\boldsymbol{\Pi}_0, \boldsymbol{\Gamma}_0), \ad^*_{(\boldsymbol{\xi}_2, \boldsymbol{v}_2)} (\boldsymbol{\Pi}_0, \boldsymbol{\Gamma}_0)\Big) = -\boldsymbol{\Pi}_0 \cdot \boldsymbol{\xi}_1 \times \boldsymbol{\xi}_1 + \boldsymbol{\Gamma}_0 \cdot (\boldsymbol{\xi}_1 \times \boldsymbol{v}_2 - \boldsymbol{\xi}_2 \times \boldsymbol{v}_1) \,,
\end{align}
where $\ad^*_{(\boldsymbol{\xi}, \boldsymbol{v})} (\boldsymbol{\Pi}, \boldsymbol{\Gamma}) = (\boldsymbol{\Pi} \times \boldsymbol{\xi} + \boldsymbol{\Gamma} \times \boldsymbol{v}, \, \boldsymbol{\Gamma} \times \boldsymbol{v})$. Thus, the Gibbs measure on $\mathcal{O}_{(\boldsymbol{\Pi}_0, \boldsymbol{\Gamma}_0)}$ reads
\begin{align}
    \mathbb{P}_\infty^{(\boldsymbol{\Pi}_0, \boldsymbol{\Gamma}_0)} = \frac{1}{Z} e^{-\frac{\beta}{2} h(\boldsymbol{\Pi}, \boldsymbol{\Gamma})} |(\omega^{\text{KKS}}_{(\boldsymbol{\Pi}_0, \boldsymbol{\Gamma}_0)})^2| \,, \quad Z = \int_{\mathcal{O}_{(\boldsymbol{\Pi}_0, \boldsymbol{\Gamma}_0)}} e^{-\frac{\beta}{2} h(\boldsymbol{\Pi}, \boldsymbol{\Gamma})} |(\omega^{\text{KKS}}_{(\boldsymbol{\Pi}_0, \boldsymbol{\Gamma}_0)})^2|\,.
\end{align}

\section{Derivation of the stochastic-dissipative point vortex system on the $2$-sphere}\label{app:point-vortex-derivation}
Recall that the KKS symplectic form on the point vortex coadjoint orbit is given by $\Omega = \sum_{n=1}^N \Gamma_n \mathrm{vol}^n_{S^2}$, where $\mathrm{vol}^n_{S^2}$ for $n=1, \ldots, N$ are identical copies of the area form on $S^2$ corresponding to vortex $n$.
The corresponding Poisson bracket reads
\begin{align}\label{eq:point-vortex-poisson}
    \{\cdot, \cdot\} = \sum_{n=1}^{N} \frac{1}{\Gamma_n} \{\cdot, \cdot\}_n \,,
\end{align}
where $\{\cdot, \cdot\}_n$ is Poisson bracket corresponding to the symplectic form $\mathrm{vol}_{S^2}$. By our discussion in Section \ref{sec:stoch-diss-rigid-body}, we know that under the embedding $S^2 \xhookrightarrow{} \mathbb{R}^3$, this bracket can further be identified with the Lie-Poisson bracket on $SO(3)$
\begin{align}\label{eq:lie-poisson-so3-vortex}
    \{f, g\}_n = -\vec{x}_n \cdot \frac{\partial f}{\partial \vec{x}_n} \times \frac{\partial g}{\partial \vec{x}_n} \,,
\end{align}
where $\vec{x} \in \mathbb{R}^3$ in boldfont denotes the extrinsic representation of $x \in S^2$ under the embedding $S^2 \xhookrightarrow{} \mathbb{R}^3$.
Now, given a point vortex ansatz $\omega = \sum_{n=1}^N \Gamma_n \delta(x;x_n)$, the corresponding streamfunction reads
\begin{align} \label{eq:point-vortex-streamfunction}
\psi(x) &= \Delta^{-1} \omega(x) = \sum_{n=1}^N \Gamma_n \Delta^{-1} \delta(x; x_n) = \sum_{n=1}^N \Gamma_n G_0(x, x_n) \,,
\end{align}
where $G_0$ is the Green's function for the Laplacian operator on the $2$-sphere. This has the explicit expression \cite{newton2002n} 
\begin{align}
    G_0(x, x') = -\frac{1}{4\pi R}\log(R^2 - \vec{x} \cdot \vec{x}'), \qquad R = \|\vec{x}\| = \|\vec{x}'\| \,,
\end{align}
where $\vec{x} \in \mathbb{R}^3$ in boldfont denotes the extrinsic representation of $x \in S^2$ under the embedding $S^2 \xhookrightarrow{} \mathbb{R}^3$. The expression $R^2 - \vec{x} \cdot \vec{x}'$ inside the logarithm represents the squared chordal distance $\|x - x'\|^2$ between points $x, x' \in S^2$ on the sphere.
The reduced Hamiltonian for the point vortex system can be derived by substituting the point vortex ansatz into the reduced Hamiltonian for ideal fluid dynamics \eqref{eq:euler-hamiltonians-2d}
\begin{align}
    h_0(\omega) &= \frac12 \int_{S^2} \omega(x) \psi(x) \mathrm{vol}_{S^2}(x) = \sum_{i=1}^N \Gamma_i \int \delta(x; x_i) \psi(x) \mathrm{vol}_{S^2}(x) \stackrel{\eqref{eq:point-vortex-streamfunction}}{=} \sum_{i,j=1}^N \frac{\Gamma_i \Gamma_j}{2} G_0(x_i, x_j) \,.
\end{align}
Similarly, the noise Hamiltonians for the point vortex system become
\begin{align}
h_i(x) = \sigma \int \omega(x) \psi_i(x) \mathrm{vol}_{S^2}(x) = \sigma \sum_{j=1}^N \Gamma_j \psi_i(x_j), \quad i = 1, 2, \ldots \,.
\end{align}
By the orthonormality of noise potentials $\psi_i$ with respect to the energy inner product, i.e.
\begin{align}
    \tilde{\gamma}(\psi_i, \psi_j) := \int_{S^2} \psi_i(x) \Delta \psi_j(x) \mathrm{vol}_{S^2} = \delta_{ij} \,,
\end{align}
we can express the streamfunction \eqref{eq:point-vortex-streamfunction} purely in terms of the noise potentials, as follows
\begin{align}
\psi(x) &= \sum_{i=1}^\infty \tilde{\gamma}(\psi, \psi_i) \psi_i(x) \nonumber \\
&= \sum_{i=1}^\infty \left(\int_{S^2} \psi_i(x) \Delta \psi(x) \mathrm{vol}_{S^2}(x)\right) \psi_i(x) \nonumber  \\
&= \sum_{i=1}^\infty \left(\int_{S^2} \psi_i(x) \omega(x) \mathrm{vol}_{S^2}(x)\right) \psi_i(x) \nonumber \\
&= \sum_{i=1}^\infty \sum_{j=1}^N \Gamma_j\psi_i(x_j) \psi_i(x) \nonumber \\
\stackrel{\eqref{eq:point-vortex-streamfunction}}{\Leftrightarrow} \quad &G_0(x, x_j) = \sum_{i=1}^\infty \psi_i(x_j) \psi_i(x) \,. \label{eq:G0-via-gamma}
\end{align}
Recall that the reduced dynamics on the coadjoint orbit is given by
\begin{align}\label{eq:reduced-stoch-diss-poisson}
    \diff f = \{f, h_0\} \,\diff t - \frac{\beta}{2} &\sum_{i=1}^\infty \{h_0, h_i\} \,\{f, h_i\} \,\diff t + \sum_{i=1}^\infty \{f, h_i\} \circ \diff W_t^i \,,
\end{align}
for the Poisson bracket $\{\cdot, \cdot\}$ given in \eqref{eq:point-vortex-poisson}.
We have
\begin{align}
\{f, h_0\} &= \sum_{n=1}^N \sum_{i,j=1}^N \frac{\Gamma_i\Gamma_j}{2\Gamma_n}\{f(x), G_0(x_i, x_j)\}_n = \sum_{i,j=1}^N \Gamma_j \{f(x), G_0(x_i, x_j)\}_i \,, \label{eq:f-h0}\\
\{h_0, h_k\} &= \sum_{i,j=1}^N \Gamma_j\{G_0(x_i, x_j), h_k(x)\}_i = \sigma \sum_{i,j=1}^N \sum_{l=1}^N \Gamma_j \Gamma_l\{G_0(x_i, x_j), \psi_k(x_l)\}_i \nonumber \\
&= \sigma \sum_{i,j=1}^N \Gamma_j \Gamma_i\{G_0(x_i, x_j), \psi_k(x_i)\}_i \,, \\
\{f, h_k\} &= \sigma \sum_{n=1}^N\sum_{l=1}^N \frac{\Gamma_l}{\Gamma_n} \{f(x), \psi_k(x_l)\}_n = \sigma \sum_{n=1}^N \{f(x), \psi_k(x_n)\}_n \,, \label{eq:f-hk} \\
\begin{split}\label{eq:h0-hk,f-hk}
\sum_{k=1}^\infty \{h_0, h_k\}\{f, h_k\} &= \sigma^2\sum_{k=1}^\infty\sum_{i,j=1}^N \Gamma_j \Gamma_i\{G_0(x_i, x_j), \psi_k(x_i)\}_i \sum_{n=1}^N \{f(x), \psi_k(x_n)\}_n \\&=  \sigma^2 \sum_{n=1}^N\sum_{i,j=1}^N \Gamma_i\Gamma_j \left\{f(x), \left\{G_0(x_i, x_j), \sum_{k=1}^\infty \psi_k(x_i) \psi_k(x_n)\right\}_i \right\}_n \\&\stackrel{\eqref{eq:G0-via-gamma}}{=} \sigma^2 \sum_{n=1}^N\sum_{i,j=1}^N \Gamma_i\Gamma_j \Big\{f(x), \big\{G_0(x_i, x_j), G_0(x_i, x_n)\big\}_i \Big\}_n \,.
\end{split}
\end{align}
Thus, denoting $\Psi_{ij}(x) := G_0(x_i, x_j)$ and plugging \eqref{eq:f-h0} -- \eqref{eq:h0-hk,f-hk}
into \eqref{eq:reduced-stoch-diss-poisson}, we arrive at the system
\begin{align}\label{eq:dxdt-point-vortex}
    \diff x_i = \sum_{j=1}^N \Gamma_j X_{\Psi_{ij}}(x_i) \,\diff t - \theta \sum_{j,k=1}^N \Gamma_j\Gamma_k X_{\{\Psi_{jk}, \Psi_{ik}\}_k}(x_i)\,\diff t + \sigma \sum_{k=1}^\infty X_{\psi_k}(x_i) \circ \diff W_t^k \,,
\end{align}
where $\theta = \beta \sigma^2/2$ and we used the relation $\{f, h\} = X_h f$.
More explicitly, using \eqref{eq:lie-poisson-so3-vortex}, we have
\begin{align}
    X_{\Psi_{ij}}(\vec{x}_i) &= \frac{\partial \psi_{ij}}{\partial \vec{x}_i} \times \vec{x}_i = \frac{\vec{x}_j \times \vec{x}_i}{4\pi R(R^2 - \vec{x}_j \cdot \vec{x}_i)} \,,\\
    \{\Psi_{jk}, \Psi_{ik}\}_k &= -\vec{x}_k \cdot \frac{\partial \Psi_{jk}}{\partial \vec{x}_k} \times \frac{\partial \Psi_{ik}}{\partial \vec{x}_k} \nonumber \\
    &= -\vec{x}_k \cdot \frac{\vec{x}_j}{4\pi R (R^2 - \vec{x}_j \cdot \vec{x}_k)} \times \frac{\vec{x}_i}{4\pi R (R^2 - \vec{x}_i \cdot \vec{x}_k)} \nonumber \\ &= -\frac{\vec{x}_k \cdot \vec{x}_j \times \vec{x}_i}{(4\pi R)^2 (R^2 - \vec{x}_j \cdot \vec{x}_k)(R^2 - \vec{x}_i \cdot \vec{x}_k)} \,,\\
    X_{\{\Psi_{jk}, \Psi_{ik}\}_k}(\vec{x}_i) &= \frac{\partial}{\partial \vec{x}_i}\{\Psi_{jk}, \Psi_{ik}\}_k \times \vec{x}_i \nonumber \\
    &= -\frac{(\vec{x}_k \times \vec{x}_j) \times \vec{x}_i}{(4\pi R)^2 (R^2 - \vec{x}_j \cdot \vec{x}_k)(R^2 - \vec{x}_i \cdot \vec{x}_k)} - \frac{[\vec{x}_k \cdot \vec{x}_j \times \vec{x}_i] \vec{x}_k \times \vec{x}_i}{(4\pi R)^2 (R^2 - \vec{x}_j \cdot \vec{x}_k)(R^2 - \vec{x}_i \cdot \vec{x}_k)^2} \,.
\end{align}
Now plugging these expressions into \eqref{eq:dxdt-point-vortex}, we arrive at the full system
\begin{align}
\begin{split}
\diff \boldsymbol{x}_i &= \frac{1}{4\pi R} \sum_{\substack{j = 1 \\ j \neq i}}^N \frac{\Gamma_j \,\boldsymbol{x}_j  \times \boldsymbol{x}_i}{R^2 - \boldsymbol{x}_i  \cdot \boldsymbol{x}_j}\diff t + \sigma \sum_{k=1}^\infty X_{\psi_k}(x_i) \circ \diff W_t^k \\
&\quad - \theta \sum_{\substack{k=1 \\ k \neq i}}^N \Bigg[\Bigg(\sum_{\substack{j = 1 \\ j \neq k}}^N \frac{\Gamma_j \,\boldsymbol{x}_j  \times \boldsymbol{x}_k}{4\pi R(R^2 - \boldsymbol{x}_j  \cdot \boldsymbol{x}_k)}\Bigg) \times \Bigg(\frac{\Gamma_k \vec{x}_i}{4\pi R(R^2 - \boldsymbol{x}_i  \cdot \boldsymbol{x}_k)} \Bigg) \\
&\qquad \qquad + \Bigg(\sum_{\substack{j = 1 \\ j \neq k}}^N \frac{\Gamma_j \,\vec{x}_i \cdot \boldsymbol{x}_j  \times \boldsymbol{x}_k}{4\pi R(R^2 - \boldsymbol{x}_j  \cdot \boldsymbol{x}_k)}\Bigg) \Bigg(\frac{\Gamma_k \vec{x}_k \times \vec{x}_i}{4\pi R(R^2 - \boldsymbol{x}_i  \cdot \boldsymbol{x}_k)^2}\Bigg)\Bigg] \,\diff t \,.
\end{split}
\end{align}

\end{appendices}

\end{document}